%% file: msmax-arXiv.tex
\documentclass[a4paper,11pt]{article}
\usepackage[T1]{fontenc}
\usepackage[top=1in,bottom=1in,left=1in,right=1in]{geometry}

\title{A geometric approach for the upper bound theorem for Minkowski
  sums of convex polytopes}

\author{Menelaos I. Karavelas$^{1,2}$
\hfil
Eleni Tzanaki$^{2}$\\[5pt]
\it{}$^1$Department of Mathematics \& Applied Mathematics\\
\it{}University of Crete\\
\it{}GR-700 13 Voutes, Heraklion, Greece\\[5pt]
\it{}$^2$Institute of Applied and Computational Mathematics,\\
\it{}Foundation for Research and Technology -- Hellas,\\
\it{}P.O. Box 1385, GR-711 10 Heraklion, Greece\\[5pt]
{\small\texttt{\{mkaravel,etzanaki\}@uoc.gr}}\\[5pt]}

\usepackage{anyfontsize} 

\usepackage{amsmath,amsthm}
\usepackage{amssymb,latexsym}
\usepackage[mathscr]{eucal}
\usepackage{MnSymbol} 
\usepackage{enumerate}
\usepackage{mathrsfs}
\usepackage{graphicx}
\usepackage[dvipsnames]{xcolor}
\usepackage{float}
\usepackage{array,arydshln}
\usepackage{lscape}
\usepackage[unicode,final=true,hidelinks]{hyperref}
\hypersetup{colorlinks=true,citecolor=PineGreen,linkcolor=Magenta}
\hypersetup{urlcolor=RoyalBlue}

\newcommand{\reals}{\mathbb{R}}
\newcommand{\naturals}{\mathbb{N}}

\newcommand{\cC}{\mathcal{C}}

\newcommand{\wW}{\mathcal{W}}
\newcommand{\fF}{\mathcal{F}}
\newcommand{\gG}{\mathcal{G}}

\newcommand{\qQ}{\mathcal{Q}}
\newcommand{\yY}{\mathcal{Y}}
\newcommand{\kK}{\mathcal{K}}

\newcommand{\zZ}{\mathcal{Z}}

\newcommand{\str}[2]{\mathrm{star}(#1,\allowbreak{}#2)}  
\newcommand{\Sl}[1]{\mathsf{S}(#1)}  
\newcommand{\CC}{\mathscr{C}}  

\newcommand{\Wavg}{\overline{W}}

\newcommand{\arrow}{\raisebox{-5pt}%
  {\includegraphics[height=10 pt]{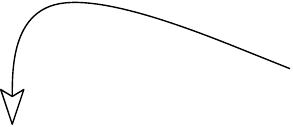}}}

\newcommand{\UU}{\mathscr{U}}
\newcommand{\VV}{\mathscr{V}}

\renewcommand{\to}{\rightarrow}

\newcommand{\conv}{{\rm conv}}
\newcommand{\sm}{\setminus}

\newcommand{\mb}[1]{{\boldsymbol{#1}}}
\newcommand{\me}{\mb{e}}
\newcommand{\mbv}{\mb{v}}
\newcommand{\mbu}{\mb{u}}
\newcommand{\mc}{\mb{\gamma}}

\newcommand{\lexp}[1]{\lfloor\frac{#1}{2}\rfloor}
\newcommand{\ltexp}[1]{\lfloor\tfrac{#1}{2}\rfloor}

\newcommand{\pddo}{(d-r+1)}
\newcommand{\bx}[1]{\partial{#1}}

\newcommand{\ptn}[1][V]{\mathsf{#1}}

\newcommand{\range}[2][r]{{#2}_1,\allowbreak{#2}_2,\allowbreak\ldots,\allowbreak{#2}_{#1}}

\newcommand{\MS}[1][P]{{#1}_{[r]}}

\newcommand{\GVD}{\text{GVD}}
\newcommand{\mesa}{\text{relint}}
\newcommand{\stel}[2]{\mathrm{st}(#1,#2)}

\newcommand{\z}{ {\color{black}{\zeta}}}

\newcommand{\hp}{$h$-polynomial } 
\newcommand{\fp}{$f$-polynomial } 
\newcommand{\h}[1]{ {\mathtt h}({#1};t)} 
\newcommand{\f}[1]{ {\mathtt f}({#1};t)}

\newcommand{\gmt}[3]{ {\mathtt g^{\tt (#1)}}({#2};{#3})} 
 
\newcommand{\hh}[1]{ {\mathtt h'}({#1};t)} 
 
\newcommand{\htt}[2]{ {\mathtt h}({#1};{#2})} 
\newcommand{\ft}[2]{ {\mathtt f}({#1};{#2})}

\newcommand{\Sgen}[2]{ {\mathtt S_{#1}}({#2}) }

\newcommand{\E}[2]{ {\mathtt E_{#1}}({#2})} 
\newcommand{\EE}[2]{ {\mathtt E'_{#1}}({#2})} 

\newcommand{\fs}{}

\newcommand{\Sh}{{\mathsf S}}
\newcommand{\SA}{{\mathsf A}}
\newcommand{\SB}{{\mathsf B}}
\newcommand{\dgr}[1]{\VV^{\scalebox{0.6}{$\triangle$}}(#1)}
\newcommand{\stirl}[2]{{\mathsf{}S}_{#1}^{#2}}
\newcommand{\eul}[2]{{\mathsf{}E}_{#1}^{#2}}
\newcommand{\D}[4]{{\mathsf D}(#1,#2,#3,#4)}

\newcommand{\0}{ \color{black!22}{0}\color{black}  }
\newcommand{\cd}{ \color{black!22}\cdots \color{black}  }
\newcommand{\dd}{ \color{black!22} \ddots \color{black}  }
\newcommand{\vd}{ \color{black!22} \vdots \color{black}  }

\allowdisplaybreaks[3]

  \newcommand\undermat[3][0pt]{%
  \makebox[0pt][l]{$\smash{\underbrace{\phantom{%
    \begin{matrix}\phantom{\rule{-10pt}{#1}}#3\end{matrix}}}_{\text{#2}}}$}#3}

\newcommand{\De}[2]{ \Delta_{ \scalebox{0.7}{$#1$}}(#2)}
\newcommand{\y}[2]{ y_{ \scalebox{0.8}{$#1,#2$}}}

\newcommand{\zz}[1]{ \color{blue}{\zeta^{#1}} \color{black}}
\newcommand{\ye}[2]{ y_{ \scalebox{1}{\hspace{-0.2 cm}  
$_{_\epsilon}$}\scalebox{0.8}{\hspace{-0.1 cm} ${#1},{#2}$}}}

\newcommand{\tubh}[3][0]{\Phi_{#2,d}^{({#1})}(\mb{n}_{#3})}
\newcommand{\tubk}[2]{\Psi_{#1,d}(\mb{n}_{#2})}

\newcommand{\spans}[2]{\Xi_{#1,d}(\mb{n}_{#2})}

\newcommand{\spsa}{\mathcal{A}}
\newcommand{\spsb}{\mathcal{B}}
\newcommand{\spsac}{{\sf{}A}}
\newcommand{\spsbc}{{\sf{}B}}

\newtheorem{theorem}{Theorem}

\newtheorem{corollary}[theorem]{Corollary}
\newtheorem{definition}[theorem]{Definition}

\newtheorem{lemma}[theorem]{Lemma}

\allowdisplaybreaks

\begin{document}

\maketitle

\begin{abstract}
  \input{abstract}
\end{abstract}


\setcounter{equation}{0}
\numberwithin{equation}{section}
\numberwithin{theorem}{section}

%
\input{intro}
\input{prelim}
\input{construction_of_Q}
\input{Dehn_Som}

\input{recurrenceF}
\input{ubFK}
\input{lbconstruction}

\section*{Acknowledgments}
The authors would like to thank Christos Konaxis for useful
discussions and comments on earlier versions of this paper,
as well as Vincent Pilaud for discussions related to the tightness
construction presented in the paper.

The work in this paper has been co-financed
by the European Union (European Social Fund -- ESF) and Greek
national funds through the Operational Program ``Education and
Lifelong Learning'' of the National Strategic Reference Framework
(NSRF) -- Research Funding Program: THALIS -- UOA (MIS 375891).

\phantomsection

\addcontentsline{toc}{section}{References}

\bibliographystyle{plain}
\bibliography{msmax}

\appendix
\renewcommand{\thesection}{\Alph{section}}
\renewcommand{\thetheorem}{\Alph{section}.\arabic{theorem}}
\input{appendix_Dehn_Som}
\input{appendix_recurrence}
\input{appendix_matrix}

\end{document}

%% file: abstract.tex
We derive tight expressions for the maximum number of $k$-faces,
$0\le{}k\le{}d-1$, of the Minkowski sum, $P_1+\cdots+P_r$, of $r$
convex $d$-polytopes $P_1,\ldots,P_r$ in $\reals^d$, where
$d\ge{}2$ and $r<d$, as a (recursively defined) function on the
number of vertices of the polytopes.
Our results coincide with those recently proved by Adiprasito and
Sanyal \cite{as-ubt-14}. In contrast to Adiprasito and Sanyal's
approach, which uses tools from Combinatorial Commutative Algebra,  
our approach is purely geometric and uses
basic notions such as $f$- and $h$-vector calculus and shellings,
and generalizes the methodology used in \cite{kt-mnfms-12} and
\cite{kkt-mnfms-13} for proving upper bounds on the $f$-vector of
the Minkowski sum of two and three convex polytopes, respectively.
The key idea behind our approach is to express the Minkowski sum
$P_1+\cdots+P_r$ as a section of the Cayley polytope $\cC$ of the
summands; bounding the $k$-faces of $P_1+\cdots+P_r$ reduces to
bounding the subset of the $(k+r-1)$-faces of $\cC$ that contain
vertices from each of the $r$ polytopes.
We end our paper with a sketch of an explicit construction that
establishes the tightness of the upper bounds.

%% file: intro.tex
\section{Introduction}
\label{sec:intro}

Given two sets $A$ and $B$ in $\reals^d$, $d\ge{}2$, their Minkowski
sum $A+B$ is the set $\{a+b\mid{}a\in{}A, b\in{}B\}$. The Minkowski
sum definition can be extended naturally to any number of summands:
$A_{[r]}:=A_1+A_2+\cdots+A_r=\{a_1+a_2+\cdots+a_r\mid{}a_i\in{}A_i,1\le{}i\le{}r\}$.
Minkowski sums have a wide range of applications, including
algebraic geometry, computational commutative algebra, collision
detection, computer-aided design, graphics, robot motion planning and
game theory, just to name a few
(see also \cite{as-ubt-14}, \cite{kkt-mnfms-13} and the references therein).

In this paper we focus on convex polytopes, and we are interested in
computing the worst-case complexity of their Minkowski sum.
More precisely, given $r$ $d$-polytopes $P_1,\ldots,P_r$
in $\reals^d$, we seek tight bounds on the number of $k$-faces
$f_k(\MS)$, $0\le{}k\le{}d-1$, of their Minkowski sum
$\MS:=P_1+P_2+\cdots+P_r$.
This problem, which can be seen as a generalization of the Upper Bound
Theorem (UBT) for polytopes \cite{m-mnfcp-70}, has a history of more than 20
years. Gritzmann and Sturmfels \cite{gs-mapcc-93} were the first to
consider the problem, and gave a complete answer to it, for any number
of $d$-polytopes in $\reals^d$, in terms of the number of non-parallel
edges of the $r$ polytopes. More than 10 years later, Fukuda and Weibel
\cite{fw-fmacp-07} proved tight upper bounds on the number of
$k$-faces of the Minkowski sum of two 3-polytopes, expressed either in
terms on the number of vertices or number of facets of the
summands. Fogel, Halperin, and Weibel \cite{fhw-emcms-09} extended one
of the results in \cite{fw-fmacp-07}, and expressed the number of
facets of the Minkowski sum of $r$ 3-polytopes in terms of the number
of facets of the summands. Quite recently Weibel \cite{w-mfmsl-12}
provided a relation for the number of $k$-faces of the Minkowski sum
of $r\ge{}d$ summands in terms of the $k$-faces of the Minkowski sums
of subsets of size $d-1$ of these summands. This result should be
viewed in conjunction with a result by Sanyal \cite{s-tovnm-09} stating that
the number of vertices of the Minkowski sum of $r$ $d$-polytopes,
where $r\ge{}d$, is strictly less than the product of the vertices of
the summands (whereas for $r\le{}d-1$ this is indeed possible).
About 3 years ago, the authors of this paper proved the first tight
upper bound on the number of $k$-faces for the Minkowski sum of two
$d$-polytopes in $\reals^d$, for any $d\ge{}2$ and for all $0\le{}k\le{}d-1$
(cf. \cite{kt-mnfms-12}), a result which was subsequently extended
to three summands in collaboration with Konaxis
(cf. \cite{kkt-mnfms-13}).

In a recent paper, Adiprasito and Sanyal
\cite{as-ubt-14} provide the complete resolution of the \emph{Upper Bound
Theorem for Minkowski sums (UBTM)}. In particular, they show that
there exists, what they call, a \emph{Minkowski-neighborly} family of
$r$ $d$-polytopes $N_1,\ldots,N_r$, with $f_0(N_i)=n_i$,
$1\le{}i\le{}r$, such that for any $r$ $d$-polytopes
$P_1,P_2,\ldots,P_r\subset{}\reals^d$ with $f_0(P_i)=n_i$, $1\le{}i\le{}r$, 
$f_k(\MS)$ is bounded by above by $f_k(\MS[N])$, for all $0\le{}k\le{}d-1$.
The majority of the arguments in the UBTM proof by Adiprasito and
Sanyal make use of powerful tools from Combinatorial Commutative
Algebra. The high-level layout of the proof is analogous to
McMullen's proof of the UBT, as well as the proofs of the UBTM in
\cite{kt-mnfms-12} and \cite{kkt-mnfms-13} for two and three
summands, respectively:
\begin{enumerate}[1.]
\item\label{step:cayley}
  Consider the Cayley polytope $\cC\subset\reals^{d+r-1}$ of the $r$
  polytopes $P_1,P_2,\ldots,P_r$, and identify their Minkowski sum as
  a section of $\cC$ with an appropriately defined $d$-flat $\Wavg$.
  Let $\fF\subset\reals^{d+r-1}$ be the faces of $\cC$ that intersect
  $\Wavg$, and let $\kK$ be the closure of $\fF$ under subface inclusion
  ($\kK$ is a $(d+r-1)$-polytopal complex).
  By the Cayley trick, there is a bijection between the faces of $\fF$ and
  the faces of $P_{[r]}$; as a result, to bound the number
  of faces of $P_{[r]}$ it suffices to bounds the number of
  faces of $\fF$.
\item\label{step:DS}
  Define the $h$-vector $\mb{h}(\fF)$ of $\fF$, and prove the
  Dehn-Sommerville equations for $\mb{h}(\fF)$, relating its elements
  to the elements of $\mb{h}(\kK)$.
\item\label{step:recurrence}
  Prove a recurrence relation for the elements of $\mb{h}(\fF)$.
\item\label{step:ubF}
  Use the recurrence relation above to prove upper bounds for
  $h_k(\fF)$, for all $0\le{}k\le{}\lexp{d+r-1}$.
\item\label{step:ubK}
  Prove upper bounds for $h_k(\kK)$, for all
  $0\le{}k\le\lexp{d+r-1}$.
\item\label{step:conditions}
  Provide necessary and sufficient
  conditions under which the elements of both $\mb{h}(\fF)$ and
  $\mb{h}(\kK)$ are maximized for all $k$. These conditions are
  conditions on the \emph{lower half} of the $h$-vector of $\fF$.
  Due to the relation between the $f$- and $h$-vectors of $\fF$, these
  are also conditions for the maximality of the elements of
  $\mb{f}(\fF)$.
\item\label{step:construction}
  Describe a family of polytopes for which the necessary and
  sufficient conditions hold; clearly, such a family establishes the
  tightness of the upper bounds.
\end{enumerate}
In Adiprasito and Sanyal's proof steps \ref{step:DS},
\ref{step:recurrence} and \ref{step:ubF} are proved by introducing a
powerful new theory that they call the \emph{relative Stanley-Reisner
  theory} for simplicial complexes. The focus of this theory is on
relative simplicial complexes, and is able to reveal properties of
such complexes not only under topological restrictions, but also
account for their combinatorial and geometric structure. To apply
their theory, Adiprasito and Sanyal consider the simplicial complex
$\kK$ and then define $\fF$ as a relative simplicial complex 
(they call them the Cayley and \emph{relative Cayley}
complex, respectively). They then apply their relative Stanley-Reisner theory to
$\fF$ to establish the Dehn-Sommerville equations of step
\ref{step:DS}, the recurrence relation of step \ref{step:recurrence}
and finally the upper bounds for $\mb{h}(\fF)$ in \ref{step:ubF}.
Steps \ref{step:ubK} and \ref{step:conditions} are done by clever algebraic
manipulation of the $h$-vectors of $\fF$ and $\kK$, by exploiting the
geometric properties of $\kK$, and by making use of the
recurrence relation in step \ref{step:recurrence}.
Step \ref{step:construction} is reduced to results by Matschke,
Pfeifle, and Pilaud \cite{mpp-pnp-11} and Weibel \cite{w-mfmsl-12}.

\emph{Our contribution.} In what follows, we provide a completely
geometric proof of the UBTM,
that generalizes the technique we used in \cite{kt-mnfms-12} and
\cite{kkt-mnfms-13} for two and three summands to the case of $r$
summands, when $r<d$.
Instead of relying on algebraic tools, we use basic notions from
combinatorial geometry, such as stellar subdivisions and shellings.
Our proof, in essence, differs from that of Adiprasito and Sanyal
in steps \ref{step:DS}, \ref{step:recurrence}, \ref{step:ubF} and
\ref{step:ubK} of the layout above (the remaining steps do not
use tools from Combinatorial Commutative Algebra anyway).

In more detail, to prove the various intermediate
results, towards the UBTM, we consider the Cayley polytope $\cC$ and we
perform a series of stellar subdivisions to get a simplicial polytope
$\qQ$. From the analysis of the combinatorial structure of $\qQ$, we
derive the Dehn-Sommerville equations of step \ref{step:DS} (see
Sections \ref{sec:construction} and \ref{sec:DS}), as well as the
recurrence relation of step \ref{step:recurrence} (see Section \ref{sec:rec}).
This recurrence relation is then used for establishing the upper bounds for
the elements of $\mb{h}(\fF)$ and $\mb{h}(\kK)$ (see Section \ref{sec:ub}).
We end with a construction similar to the one presented in
\cite[Theorem 2.6]{mpp-pnp-11}, that establishes the tightness
of the upper bounds (see Section \ref{sec:tbconstruction}).

%% file: prelim.tex
\section{Preliminaries}
\label{ssec:prelim} 
Let $P$ be a $d$-dimensional polytope, or $d$-polytope for short.
Its dimension is the dimension of its affine span.
The faces of $P$ are $\emptyset,P$, and the intersections of $P$ with its 
supporting hyperplanes. The $\emptyset$ and $P$ faces are called  
\emph{improper}, while the remaining faces are called \emph{proper}. Each face 
of $P$ is itself a polytope, and a face of dimension $k$ is called a $k$-face.
Faces of $P$ of dimension $0,1,d-2$ and $d-1$ are called vertices, edges, 
ridges, and facets, respectively.

A $d$-dimensional \emph{polytopal complex} or, simply,
\emph{$d$-complex}, $\CC$ is a
finite collection of polytopes in $\reals^d$ such that (i) $\emptyset\in\CC$,
(ii) if $P\in\CC$ then all the faces of $P$ are also in $\CC$ and (iii) the 
intersection $P\cap{}Q$ for two polytopes $P$ and $Q$ in $\CC$ is a face of 
both. The dimension $\dim(\CC)$ of $\CC$ is the largest dimension of a polytope 
in $\CC$. A polytopal complex is called \emph{pure} if all its maximal (with 
respect to inclusion) faces have the same dimension. In this case the maximal 
faces are called the \emph{facets} of $\CC$. A polytopal complex is 
\emph{simplicial} if all its faces are simplices. A polytopal complex 
$\CC'$ is called a \emph{subcomplex} of a polytopal complex $\CC$ if all faces 
of $\CC'$ are also faces of $\CC$.
For a polytopal complex $\CC$, the \emph{star} of $v$ in $\CC$, denoted by
$\str{v}{\CC}$, is the subcomplex of $\CC$ consisting of all faces that contain 
$v$, and their faces. The \emph{link} of $v$, denoted by $\CC/v$, is the 
subcomplex of $\str{v}{\CC}$ consisting of all the faces of $\str{v}{\CC}$ that 
do not contain $v$.

A $d$-polytope $P$, together with all its faces, forms a $d$-complex, denoted
by $\CC(P)$. The polytope $P$ itself is the only maximal face of $\CC(P)$, 
i.e., the only facet of $\CC(P)$, and is called the \emph{trivial} face of 
$\CC(P)$. Moreover, all proper faces of $P$ form a pure $(d-1)$-complex, called 
the \emph{boundary complex} $\CC(\bx{}P)$, or simply $\bx{}P$, of $P$. The 
facets of $\bx{}P$ are just the facets of $P$.

For a $(d-1)$-complex $\CC$, its $f$-vector is defined as
$\mb{f}(\CC) = (f_{-1},f_0, f_1,\dots,\allowbreak{}f_{d-1})$, where
$f_k = f_k(\CC)$
denotes the number of $k$-faces of $P$ and $f_{-1}(\CC):=1$ corresponds to the 
empty face of $\CC$. From the $f$-vector of $\CC$ we define its $h$-vector as 
the vector $\mb{h}(\CC)=(h_0,h_1,\ldots,h_d)$, where $h_k=h_k(\CC):=$ 
$\sum_{i=0}^k(-1)^{k-i}\binom{d-i}{d-k}f_{i-1}(\CC),$  $0\leq{}k\leq d$.

Denote by $\yY$ a generic subset of faces of a polytopal complex
$\CC$, and define its dimension $\dim(\yY)$ as the maximum of the
dimensions of its faces. Let $\dim(\yY)=\delta-1$; then we may define
(if not already properly defined), the $h$-vector $\mb{h}(\yY)$ of $\yY$ as:
	\begin{equation}
	\label{def:f-h}
	  h_k(\yY)=\sum_{i=0}^{\delta}(-1)^{k-i}
	  \binom{\delta-i}{\delta-k}f_{i-1}(\yY).
	\end{equation}
We can further define the $m$-order $g$-vector of $\yY$ according to
the following recursive formula:
\begin{equation}\label{equ:gm-def}
  g_k^{(m)}(\yY)=
  \begin{cases}
    h_k(\yY),&m=0,\\
    g_k^{(m-1)}(\yY)-g_{k-1}^{(m-1)}(\yY),&m>0.
  \end{cases}
\end{equation}
Clearly, $\mb{g}^{(m)}(\yY)$ is nothing but the backward $m$-order
finite difference of $\mb{h}(\yY)$; therefore:
\begin{equation}
g_k^{(m)}(\yY)=\sum_{i=0}^m(-1)^i\binom{m}{i}h_{k-i}(\yY),\qquad k,m\ge{}0.
\end{equation}
Observe that for $m=0$ we get the $h$-vector of $\yY$, while for $m=1$
we get what is typically defined as the $g$-vector.

The relation between the $f$- and  $h$-vector of $\yY$ is better manipulated  
using generating functions. We define the \fp and \hp of $\yY$ as follows:
\begin{equation*}
  \f{\yY}=\sum_{i=0}^\delta 
  f_{i-1}t^{\delta-i}=f_{\delta-1}+f_{\delta-2}t+\cdots+f_{-1}t^\delta,\quad
  \h{\yY}=\sum_{i=0}^\delta h_{i}t^{\delta-i}
  =h_\delta +h_{\delta-1}t+\cdots+ h_0 t^\delta,
  \label{hpol}
\end{equation*}
where, we simplified $f_i(\yY)$ and $h_i(\yY)$ to $f_i$ and $h_i$.
In this set-up, the relation between the $f$-vector and 
$h$-vector (cf. \eqref{def:f-h}) can be expressed as: 
\begin{equation}
  \ft{\yY}{t}=\htt{\yY}{t+1},
  \qquad\text{or, equivalently, as }\qquad
  \htt{\yY}{t}=\ft{\yY}{t-1}.
  \label{hf}
\end{equation}


\subsection{The Cayley embedding, the Cayley polytope and the Cayley trick}
\label{subsec:cayley}

\newcommand{\e}[1]{\mb e_{#1}}

Let $\range{P}$ be $r$ $d$-polytopes with vertex sets $\range{\VV}$, 
respectively. Let $\e{0},\e{1},\allowbreak\ldots,\allowbreak\e{r-1}$ be an affine basis of 
$\reals^{r-1}$ and call $\mu_i:\reals^d\to\reals^{r-1}\times\reals^d$  the 
affine inclusion given by $\mu_i(\mb{x})=(\e{i},\mb{x})$. The \emph{Cayley 
embedding} $\cC(\range{\VV})$ of the point sets $\range{\VV}$ is defined as 
$\cC(\range{\VV})=\bigcup_{i=1}^r\mu_i(\VV_i)$. The polytope corresponding to 
the convex hull $\conv\bigl(\cC(\range{\VV})\bigr)$ of the Cayley embedding
$\cC(\range{\VV})$ of $\range{\VV}$ is typically referred to as the 
\emph{Cayley polytope} of $\range{P}$.

The following lemma, known as \emph{the Cayley trick for Minkowski sums},
relates the Minkowski sum of the polytopes
$\range{P}$ with their Cayley polytope.
\begin{lemma}[{\cite[Lemma 3.2]{hrs-ctlsb-00}}]
  \label{lem:cayley-embedding}
  Let $\range{P}$ be $r$ $d$-polytopes with vertex sets   
  $\range{\VV}\subset\reals^d$. Moreover, let $\Wavg$ be the $d$-flat 
  defined as $\{\tfrac{1}{r}\e{1}+\cdots+\tfrac{1}{r}\e{r}\}\times
  \reals^d\subset\reals^{r-1}\times\reals^d$. Then, the Minkowski sum
  $\MS$ has the  following representation as a section of the Cayley 
  embedding  $\cC(\range{\VV})$ in $\reals^{r-1}\times\reals^d$:
  \begin{align*}
    \MS  &\cong  \cC(\range{\VV})\cap \Wavg\\
    &:=\Big\{\conv\{(\e{i},\mbv_i)\mid{}1\le{}i\le{}r\}\cap{}\Wavg \,:
    (\e{i},\mbv_i)\in{}\cC(\range{\VV}), 1\le{}i\le{}r\Big\}.
  \end{align*}
  Moreover, $F$ is a facet of $\MS$ if and only if it is of the 
  form   $F=F'\cap{}\Wavg$ for a facet $F'$ of $\cC(\range{\VV})$
  containing at least one point $(\e{i},\mbv_i)$ for all $1\le{}i\le{}r$.
\end{lemma}

Let $\cC_{[r]}$ be the Cayley polytope of $\range{P}$, and call $\fF_{[r]}$
the set of faces of $\cC_{[r]}$ that have non-empty intersection with the
$d$-flat $\Wavg$. A direct consequence of Lemma \ref{lem:cayley-embedding} is a 
bijection between the $(k-1)$-faces of $\Wavg$ and the $(k-r)$-faces of 
$\fF_{[r]}$, for $r\le{}k\le{}d+r-1$. This further implies that:
\begin{equation}\label{equ:fkW}
  f_{k-1}(\fF_{[r]})=f_{k-r}(\MS),\quad \mbox{for all } r\le{}k\le{}d+r-1.
\end{equation}

In what follows, to  keep the notation lean, we identify $V_i$
with its pre-image $\VV_i$.
For any $\emptyset\subset{}R\subseteq{}[r]$, we denote by $\cC_R$ the Cayley 
polytope of the polytopes $P_i$ where $i\in{}R$. In particular, if 
$R=\{i\}$ for some $i\in[r]$, then $\cC_{\{i\}}\equiv{}P_i$. We shall assume 
below that $\cC_{[r]}$ is \emph{``as simplicial as possible''}. This means that 
we consider all faces of $\cC_{[r]}$ to be simplicial, except possibly 
for the trivial faces $\{\cC_R\}$\footnote{We denote by $\{\cC_R\}$ the 
polytope $\cC_R$ as a trivial face itself (without its non-trivial faces).},   
$\emptyset\subset{}R\subseteq{}[r]$. Otherwise, we can employ the so called 
\emph{bottom-vertex triangulation} \cite[Section 6.5, pp.~160--161]{Mat02} to 
triangulate all proper faces of $\cC_{[r]}$ except for the trivial ones, i.e., 
$\{\cC_R\},$ $\emptyset\subset{}R\subseteq{}[r]$. The resulting complex is 
polytopal (cf.~\cite{EwSh74}) with all of its faces being simplicial, except 
possibly for the trivial ones. Moreover, it has the same number of vertices as 
$\cC_{[r]}$, while the number of its $k$-faces is never less than the number of 
$k$-faces of $\cC_{[r]}$. 

For each  $\emptyset\subset{}R\subseteq[r]$, we  denote by $\fF_R$ the set of 
faces of $\cC_R$ having at least one vertex from each $V_i$, $i\in{}R$ 
and we call it the set of \emph{mixed faces of $\cC_R$}. We trivially have that 
$\fF_{\{i\}}\equiv{}\bx{}P_i$. We define the dimension of $\fF_R$ to be the 
maximum dimension of the faces in $\fF_R$, i.e.,  
$\dim(\fF_R)=\max_{F\in\fF_R}\dim(F)=d+|R|-2$.
Under the \emph{``as simplicial as possible''} assumption above,
the faces in $\fF_R$ are simplicial.
We denote by $\kK_R$ the \emph{closure}, under subface inclusion, of $\fF_R$. 
By construction, $\kK_R$ contains: (1) all faces in $\fF_R$,
(2) all faces that are subfaces of faces in $\fF_R$, and (3) the
empty set. It is easy to see that $\kK_R$ does not contain any of the
trivial faces  $\{\cC_S\}$, $\emptyset\subset{}S\subseteq{}R$, and thus,
$\kK_R$ is a pure simplicial $(d+|R|-2$)-complex.
It is also easy to verify that
	\begin{equation}\label{equ:fkKR-def}
		f_k(\kK_R)=\sum_{\emptyset\subset{}S\subseteq{}R}f_k(\fF_S),
		\qquad -1\le{}k\le{}d+|R|-2,
	\end{equation}
where in order for the above equation to hold for $k=-1$, we set
$f_{-1}(\fF_S)=(-1)^{|S|-1}$ for all $\emptyset\subset{}S\subseteq{}R$.
 In what follows we use the convention that $f_k(\fF_R)=0$, for any $k<-1$ or 
 $k>d+|R|-2$.

A general form of the Inclusion-Exclusion Principle states that if 
$f$ and $g$ are two functions defined over the subsets of a finite
set $A$, such that $f(A)=\sum_{\emptyset\subset{}B\subseteq{}A}g(B)$, then
$g(A)=\sum_{\emptyset\subset{}B\subseteq{}A}(-1)^{|A|-|B|}f(B)$
\cite[Theorem 12.1]{ggl-hc-95}.
Applying this principle in \eqref{equ:fkKR-def}, we deduce 
that:
\begin{equation}\label{equ:fkFR-ie}
	f_k(\fF_R)=\sum_{\emptyset\subset{}S\subseteq{}R}(-1)^{|R|-|S|}\,f_k(\kK_S),
	\qquad -1\le{}k\le{}d+|R|-2.
\end{equation}


In the majority of our proofs that involve evaluation of 
$f$- and $h$-vectors, we use generating functions as they 
significantly simplify calculations. The starting point is to evaluate 
$\ft{\kK_R}{t}$ (resp., $\ft{\fF_R}{t}$) in terms of the generating functions 
$\ft{\fF_S}{t}$  (resp., $\ft{\kK_S}{t}$), $\emptyset\subset{}S\subseteq{}R$, 
for each fixed choice of $\emptyset\subset{}R\subseteq{}[r]$.
Then, using \eqref{hf} we derive the 
analogous relations between their $h$-vectors.

Recalling that $\dim(\kK_R) = d+|R|-2$  and $\dim(\fF_S) = d+|S|-2$ we have:
\begin{equation} 
  \begin{aligned}
    \ft{\kK_R}{t}&=\sum_{k=0}^{d+|R|-1}f_{k-1}(\kK_R)t^{d+|R|-1-k}
    \overset{\eqref{equ:fkKR-def}}{=}
    \sum_{k=0}^{d+|R|-1}\sum_{\emptyset\subset{}S\subseteq{}R}
    f_{k-1}(\fF_S)t^{d+|R|-1-k}\\
    &=\sum_{S\subseteq{}R}t^{|R|-|S|}\sum_{k=0}^{d+|R|-1}f_{k-1}(\fF_S)
    t^{d+|S|-1-k}
    =\sum_{\emptyset\subset{}S\subseteq{}R}t^{|R|-|S|}\ft{\fF_S}{t}.
  \end{aligned}\label{fKF}
\end{equation}
Rewriting the above relation as 
$t^{-|R|}\ft{\kK_R}{t}=
\sum_{\emptyset\subset{}S\subseteq{}R}t^{-|S|}\ft{\fF_S}{t}$ 
and using M\"obious inversion, we get:
\begin{align}
  \ft{\fF_R}{t}=\sum_{\emptyset\subset{}S\subseteq{}R}(-1)^{|R|-|S|}
  t^{|R|-|S|}\ft{\kK_S}{t}.\label{fFK}
\end{align}
Setting $t:=t-1$ in \eqref{fKF} we have:
\begin{equation}
  \begin{aligned} 
    \htt{\kK_R}{t}=\ft{\kK_R}{t-1}&=\sum_{\emptyset\subset{}S\subseteq{}R}
    (t-1)^{|R|-|S|}\ft{\fF_S}{t-1}\\
    &=\sum_{\emptyset\subset{}S\subseteq{}R}(t-1)^{|R|-|S|}
    \htt{\fF_S}{t}
    =\sum_{\emptyset\subset{}S\subseteq{}R}\gmt{{|R|-|S|}}{\fF_S}{t}.
  \end{aligned}
  \label{gen:hKF}
\end{equation}
And similarly, from \eqref{fFK}  we obtain: 	
\begin{align} 
\htt{\fF_R}{t}&=\sum_{\emptyset\subset{}S\subseteq{}R}(-1)^{|R|-|S|}\gmt{{|R|-|S|}}{\kK_S}{t}.
\label{gen:hFK}
\end{align}
Comparing coefficients in the above generating functions, we  deduce that:  	
\begin{align}
	h_k(\kK_R)& =\sum_{\emptyset\subset{}S\subseteq{}R}g_{k}^{(|R|-|S|)}(\fF_S),
	 && \mbox{ for all }0\le{}k\le{}d+|R|-1, \;\mbox{ and }	
	\label{hkKF} \\
	h_k(\fF_R)& =\sum_{\emptyset\subset{}S\subseteq{}R}(-1)^{|R|-|S|}
	g_{k}^{(|R|-|S|)}(\kK_S),
	 && \mbox{ for all }0\le{}k\le{}d+|R|-1.
\label{hkFK}
\end{align}	

%% file: construction_of_Q.tex
\section{\texorpdfstring{The construction of the auxiliary simplicial polytope $\qQ_{[r]}$.}{The construction of the auxiliary simplicial polytope Qr}}
\label{sec:construction}

The non-trivial faces of the Cayley polytope $\cC_{[r]}$ of 
$P_1,\ldots,P_r$ are the  faces in each $\fF_R$, 
$\emptyset\subset{}R\subseteq{}[r]$ as well as all trivial faces  $\{\cC_R\}$ 
with $\emptyset{}\subset{}R\subset{}[r]$.
Since the latter are not necessarily simplices, the 
Cayley polytope $\cC_{[r]}$ may not be simplicial.
In order to exploit the combinatorial structure of $\cC_{[r]}$,
we add auxiliary points on $\cC_{[r]}$ so that the resulting polytope,
denoted by $\qQ_{[r]}$, is simplicial.

The main tool for describing our construction is \emph{stellar subdivisions}.
Let $P\subset{}\mathbb R^d$ be a $d$-polytope, and consider a point
$y_F$ in the relative interior of a face $F$ of $\bx{}P$.
The \emph{stellar subdivision} $\stel{y_F}{\bx{}P}$ of $\bx{}P$ over
$F$, replaces $F$ by the set of faces $\{y_F,F'\}$ where $F'$ is a
non-trivial face of $F$. It is a well-known fact that stellar
subdivisions preserve polytopality (cf.~\cite[pp. 70--73]{e-ccag-96}),
in the sense that the newly constructed complex is combinatorially
equivalent to a polytope each facet of which lies on a
distinct supporting hyperplane.

Our goal is to triangulate each face $\{\cC_R\}$,
$\emptyset\subset{}R\subset{}[r]$, of $\cC_{[r]}$ so that
the boundaries of the resulting complexes, denoted by
$\qQ_S$, $\emptyset\subset{}S\subseteq{}[r]$, are simplicial polytopes.
We obtain this by performing a series of stellar subdivisions.
First set $\qQ_S:=\cC_{S}$, for all $\emptyset\subset{}S\subseteq{}[r]$.
Then, we add auxiliary vertices as follows:
\begin{equation}
  \label{algorithm}
    \begin{aligned}
      &\mbox{for }s\mbox{ from }1\mbox{ to }r-1\\
      &\hspace*{0.75 cm}
      \mbox{for all }S\subseteq{}[r]\mbox{ with }|S|=s\\
      &\hspace*{1.5cm}
      \mbox{choose }y_S\in\mesa(\qQ_S)\\
      &\hspace*{1.5cm}
      \mbox{for all }T\mbox{ with }S\subset{}T\subseteq[r]\\
      &\hspace*{2.25cm}
      \qQ_T:=\stel{y_S}{\qQ_T}
  \end{aligned}
\end{equation}
The recursive step of the previous definition is well defined due to
the fact that for any fixed $s$, the order in which we add the
auxiliary points $y_S$ is independent of the $S$ chosen, since the
relative interiors of all $\qQ_S$ with $|S|=s$ are pairwise disjoint.
At the end of the $s$-th iteration, the faces of
each $\qQ_T$ of dimension less than $d+s-1$ are simplices.
At the end of the iterative procedure above, and in view of the fact that
stellar subdivisions preserve polytopality, the above construction
results in simplicial $(d+|R|-1)$-polytopes $\qQ_{R}$, for all
$\emptyset\subset{}R\subseteq{}[r]$.

The following two lemmas express the faces of $\bx\qQ_R$ in terms of the
sets $\fF_S$, $\kK_S$, $\emptyset{}\subset{}S\subseteq{}R$, and the auxiliary
vertices added. Unless otherwise stated, all set unions are disjoint. 
\begin{lemma}
\label{lem:QR-FSjoin}
For $\emptyset{}\subset{}R\subseteq{}[r],$  the non-trivial faces of the 
simplicial polytope $\qQ_R$ are: 
\begin{equation} 
  \label{equ:QR-FSjoin}
  \bx{}\qQ_R=\bigcup\limits_{\emptyset\subset{}S\subseteq{}R}\fF_S
  \bigcup_{\scalebox{0.7}{$\substack{\emptyset\subset{}S\subset{}R\\
	S\subseteq{}S_1{}\subset{}S_2\subset{}\cdots\subset{}S_\ell\subset{}R}$}}
  \{y_{S_1},y_{S_2},\ldots,y_{S_\ell},\fF_S\}
  \bigcup_{\scalebox{0.7}{$\substack{
	\emptyset{}\subset{}S_1{}\subset{}S_2\subset{}\cdots\subset{}
	S_\ell\subset{}R}$}}\{y_{S_1},y_{S_2},\ldots,y_{S_\ell}\},
\end{equation}
where $\{y_{S_1},y_{S_2},\ldots,y_{S_\ell},\fF_S\}$ is the set of faces 
formed by the vertices $y_{S_1},\ldots,y_{S_\ell}$ and a face in $\fF_S$. 
\end{lemma}

\begin{proof}
We use induction on the size of $|R|$, the case $|R|=1$ being trivial. 
We next assume that our result holds true for $|R|=\rho$ and we prove it for 
$|R|=\rho+1.$

When $|R|=\rho+1$ the  recursion in \eqref{algorithm} coincides with that of the
case $|R|=\rho$, until the last but one step, i.e., when $s=\rho-1$.  Thus, 
before doing the last recursion, we have:
\begin{enumerate}
\item[(a)] By induction:
  \begin{equation*}
    \bx{}\qQ_{S'}=\bigcup\limits_{\emptyset\subset{}S\subseteq{}S'}
    \fF_{S}
    \bigcup\limits_{\scalebox{0.7}
      {$\substack{\emptyset\subset{}S\subset{}S'\\
	S\subseteq{}S_1{}\subset{}S_2\subset{}\cdots\subset{}S_\ell\subset{}S'}$}}
    \{y_{S_1},y_{S_2},\ldots,y_{S_\ell},\fF_S\}
    \bigcup\limits_{\scalebox{0.7}{$\substack{\emptyset{}\subset{}S_1{}\subset{}
	  S_2\subset{}\cdots\subset{}S_\ell\subset{}S'}$}}
    \{y_{S_1},y_{S_2},\ldots,y_{S_\ell}\},
  \end{equation*}
  for all $\emptyset\subset{}S'\subseteq{}R$ with $|S'|\leq\rho-1$.
\item[(b)] By our construction, the faces in $\bx{}\qQ_R$ are: 
  \begin{enumerate}[1)]
  \item faces in each $\bx\qQ_{R'},|R'|=\rho-1$, 
  \item the (trivial) faces 
    $\{\qQ_{R'}\}$ for $|R'|=\rho-1$, and
  \item  faces in $\fF_R.$	
  \end{enumerate}
\end{enumerate}	
The faces in (b.1)-(b.3) are not necessarily disjoint. However, 
using (a) we can write them  disjointly as follows: 
\begin{equation} 
  \bx{}\qQ_R=\bigcup\limits_{\emptyset\subset{}S\subseteq{}R}\fF_S
  \bigcup_{\scalebox{0.7}{$\substack{\emptyset\subset{}S\subset{}R\\
	S\subseteq{}S_1{}\subset{}S_2\subset{}\cdots\subset{}S_\ell\subset{}R\\
	|S|\leq{}\rho-1}$}}
  \mbox{\hspace{-0.4cm}}\{y_{S_1},y_{S_2},\ldots,y_{S_\ell},\fF_S\}
  \bigcup_{\scalebox{0.7}{$\substack{\emptyset\subset{}S_1{}\subset{}S_2
	\subset{}\cdots\subset{}S_\ell\subset{}R\\|S|\leq{}\rho-1}$}}
  \mbox{\hspace{-0.4cm}}\{y_{S_1},y_{S_2},\ldots,y_{S_\ell}\}
  \bigcup\limits_{\scalebox{0.7}{$\substack{S\subset{}R\\|S|=\rho-1}$}}\{\qQ_S\}.
  \label{algo:1}
\end{equation}
The faces in \eqref{algo:1} that will be stellarly subdivided in the last
recursion of \eqref{algorithm} are all in some $\{\qQ_{S}\}$ with
$|S|=\rho-1$. These, will be replaced by:
\begin{equation}
  \underset{\underset{|S|=\rho}{S\subseteq{}R}}{\bigcup}
  \biggl(
  \underset{\scalebox{0.6}{$\emptyset\subset{}X\subseteq{}S$}}{\bigcup}
  \{y_{S},\fF_{X}\}\mbox{\hspace{-0.5cm}}
  \bigcup_{\scalebox{0.7}{$\substack{\emptyset\subset{}X\subset{}S\\
	X\subseteq{}S_1{}\subset{}S_2\subset{}\cdots\subset{}S_\ell\subset{}S}$}}
  \mbox{\hspace{-0.4cm}}
  \{y_{S_1},y_{S_2},\ldots,y_{S_\ell},y_{S},\fF_{X}\}
  \bigcup\limits_{\scalebox{0.7}{$\substack{\emptyset{}\subset{}S_1{}\subset{}
	S_2\subset{}\cdots\subset{}S_\ell\subset{}S}$}}
  \mbox{\hspace{-0.4cm}}\{y_{S_1},y_{S_2},\ldots,y_{S_\ell},y_{S}\}
  \biggr).
  \label{algo:2}
\end{equation}
Combining \eqref{algo:1} and \eqref{algo:2} and recalling that 
$|R|=\rho+1$ we conclude that indeed 
\begin{equation*} 
  \bx{}\qQ_R=\bigcup
  \limits_{\emptyset\subset{}S\subseteq{}R}\fF_S\mbox{\hspace{-0.5cm}}
  \bigcup_{\scalebox{0.7}{$\substack{\emptyset\subset{}S\subset{}R\\
 	S\subseteq{}S_1{}\subset{}S_2\subset{}\cdots\subset{}S_\ell\subset{}R}$}}
  \mbox{\hspace{-0.4cm}}\{y_{S_1},y_{S_2},\ldots,y_{S_\ell},\fF_S\}
  \bigcup_{\scalebox{0.7}{$\substack{
	\emptyset{}\subset{}S_1{}\subset{}S_2\subset{}\cdots\subset{}
	S_\ell\subset{}R}$}}
  \mbox{\hspace{-0.4cm}}\{y_{S_1},y_{S_2},\ldots,y_{S_\ell}\}.\qedhere
\end{equation*}
\end{proof}

\medskip

\begin{lemma}
\label{lem:QR-KSjoin} 
For $\emptyset{}\subset{}R\subseteq{}[r],$ the non-trivial faces of the 
simplicial polytope 
$\qQ_R$ are:
\begin{equation}
  \label{equ:QR-KSjoin} 
  \bx{}\qQ_R=\kK_R\bigcup_{\scalebox{0.7}
    {$\substack{\emptyset\subset{}S\subset{}R\\		
	S=S_1{}\subset{}S_2\subset{}\cdots\subset{}S_\ell\subset{}R}$}}
  \{y_{S_1},y_{S_2},\ldots,y_{S_\ell},\kK_S\}.
\end{equation}
\end{lemma}

\begin{proof}
Recall that the faces of $\kK_R$ are all faces in
$\bigcup_{\emptyset\subset{}S\subseteq{}R}\fF_S$ together with the empty set. 
We can therefore  write the right-hand side of \eqref{equ:QR-KSjoin} as:
\begin{align*}
  &\underset{\emptyset\subset{}S\subseteq{}R}{\bigcup}\fF_S
  \bigcup_{\scalebox{0.7}{$\substack{\emptyset\subset{}S\subset{}R\\
	S=S_1{}\subset{}S_2\subset{}\cdots\subset{}S_\ell\subset{}R}$}}
  \{y_{S_1},y_{S_2},\ldots,y_{S_\ell},
  \underset{\emptyset\subset{}S'\subseteq{}S}{\bigcup}\fF_{S'}\}
  \bigcup_{\scalebox{0.7}{$\substack{\emptyset\subset{}S\subset{}R\\
	S=S_1{}\subset{}S_2\subset{}\cdots\subset{}S_\ell\subset{}R}$}}
  \{y_{S_1},y_{S_2},\ldots,y_{S_\ell}\}
  \\
  &=\underset{\emptyset\subset{}S\subseteq{}R}{\bigcup}\fF_S
  \bigcup_{
    \scalebox{0.7}{$\substack{\emptyset\subset{}S'\subseteq{}S\subset{}R\\
	S'\subseteq{}S=S_1{}\subset{}S_2\subset{}\cdots\subset{}S_\ell\subset{}R}$}}
  \{y_{S_1},y_{S_2},\ldots,y_{S_\ell},\fF_{S'}\}
  \bigcup_{\scalebox{0.7}
    {$\substack{\emptyset\subset{}S\subset{}R\\
	S=S_1{}\subset{}S_2\subset{}\cdots\subset{}S_\ell\subset{}R}$}}
  \{y_{S_1},y_{S_2},\ldots,y_{S_\ell}\}
  \\
  &=\underset{\emptyset\subset{}S\subseteq{}R}{\bigcup}\fF_S
  \bigcup_{\scalebox{0.7}{
      $\substack{\emptyset\subset{}S'\subset{}R\\
	S'\subseteq{}S_1{}\subset{}S_2\subset{}\cdots\subset{}S_\ell\subset{}R}$}}	
  \{y_{S_1},y_{S_2},\ldots,y_{S_\ell},\fF_{S'}\}
  \bigcup_{\scalebox{0.7}{$\substack{\emptyset\subset{}S\subset{}R\\
	S=S_1{}\subset{}S_2\subset{}\cdots\subset{}S_\ell\subset{}R}$}}
  \{y_{S_1},y_{S_2},\ldots,y_{S_\ell}\},
\end{align*}
which is precisely the quantity in \eqref{equ:QR-FSjoin} and thus equal to 
the set of faces of $\bx{}\qQ_R.$
\end{proof}

The next lemma shows how the iterated stellar subdivisions performed 
in \eqref{algorithm} are captured in the enumerative structure of
$\qQ_R$.
\begin{lemma}\label{lem:stirling}
  For any $\emptyset{}\subset{}R\subseteq[r]$ and
  $-1\le{}k\le{}d+|R|-2$, we have:
  \begin{align}
    f_k(\bx{}\qQ_{R})&=f_{k}(\fF_R)+\sum_{\emptyset\subset{}S\subset{}R}
    \left(\sum_{i=0}^{|R|-|S|}i!\,\stirl{|R|-|S|+1}{i+1}\,f_{k-i}(\fF_S)
    \right),
    \label{equ:fkQR-FS}\\
    f_k(\bx{}\qQ_{R})&=f_k(\kK_{R})+
    \sum_{\emptyset\subset{}S\subset{}R}\left(\sum_{i=0}^{|R|-|S|-1}
    (i+1)!\,\stirl{|R|-|S|}{i+1}\,f_{k-1-i}(\kK_S)\right),
    \label{equ:fkQR-KS}
  \end{align}
  where $\stirl{m}{k}$ are the Stirling numbers of the second kind
  \cite{stirling-2ndkind}:
  \begin{equation*}
    \stirl{m}{k}=\frac{1}{k!}\sum_{i=0}^{k}(-1)^{k-i}\binom{k}{i}i^m,
    \qquad
    m\ge{}k\ge{}0.
  \end{equation*}
\end{lemma}

\begin{proof}
  To prove \eqref{equ:fkQR-FS}, we count the $(k+1)$-element subsets of the set 
  in relation \eqref{equ:QR-FSjoin} of Lemma \ref{lem:QR-FSjoin}. This gives:  
  \begin{align} 
    f_{k}(\bx{}\qQ_R)&=\sum\limits_{\emptyset\subset{}S\subseteq{}R}
    f_k(\fF_R)+\sum\limits_{\emptyset\subset{}S\subset{}R}
    \sum_{i=1}^{|R|-|S|}\left|\{S\subseteq{}S_1{}\subset{}S_2\subset{}
    \cdots\subset{}S_i\subset{}R\}\right|f_{k-i}(\fF_S)
    \label{ch:1}\\
    &=\sum\limits_{\emptyset\subset{}S\subseteq{}R}f_k(\fF_R)+
    \sum\limits_{\emptyset\subset{}S\subset{}R}
    \sum_{i=1}^{|R|-|S|}\spsac_{|R|-|S|}(\emptyset,i)f_{k-i}(\fF_S)
    \label{ch:2}\\
    &=\sum\limits_{\emptyset\subset{}S\subseteq{}R}f_k(\fF_R)+
    \sum\limits_{\emptyset\subset{}S\subset{}R}\sum_{i=1}^{|R|-|S|}
    i!\,\stirl{|R|-|S|+1}{i+1}f_{k-i}(\fF_S)\label{ch:3}\\
    &=f_k(\fF_R)+\sum\limits_{\emptyset\subset{}S\subset{}R}
    \sum_{i=0}^{|R|-|S|}i!\,\stirl{|R|-|S|+1}{i+1}f_{k-i}(\fF_S),
    \label{ch:4}
  \end{align} 
  where, 
	\begin{itemize}
		\item 
		the value $i=k+1$ in \eqref{ch:1} combined with the fact that 
		$f_{-1}(\fF_S)=(-1)^{|S|-1}$, counts precisely the elements in 
		$\bigcup\limits_{\scalebox{0.7}{$\substack{
		\emptyset{}\subset{}S_1{}\subset{}S_2\subset{}\cdots\subset{}
		S_\ell\subset{}R}$}}\{y_{S_1},y_{S_2},\ldots,y_{S_\ell}\}$ 
		in relation \eqref{equ:QR-FSjoin} of Lemma \ref{lem:QR-FSjoin}
		via inclusion exclusion,    
		\item to go from \eqref{ch:2} to \eqref{ch:3} we used Lemma 
		\ref{lem:chain_counting}(ii), and 
		\item from \eqref{ch:3} to \eqref{ch:4} 
		we used the fact that $\stirl{m}{1}=1$ for all $m\geq{}1.$
	\end{itemize}

  To prove \eqref{equ:fkQR-KS}, we utilize Lemma \ref{lem:QR-KSjoin}:
  \begin{align} 			
    f_{k}(\bx{}\qQ_R)&=f_k(\kK_R)+\sum
    \limits_{\emptyset\subset{}S\subset{}R}\sum_{i=1}^{|R|-|S|}
    \left|\{S=S_1{}\subset{}S_2\subset{}\cdots\subset{}S_i\subset{}R\}\right|
    f_{k-i}(\kK_S)\notag\\
    &=f_k(\kK_R)+\sum\limits_{\emptyset\subset{}S\subset{}R}
    \sum_{i=1}^{|R|-|S|}\spsbc_{|R|-|S|} (\emptyset,i)\,f_{k-i}(\kK_S)
    \label{ch:5}\\
    &=f_k(\kK_R)+\sum\limits_{\emptyset\subset{}S\subset{}R}
    \sum_{i=1}^{|R|-|S|}i!\,\stirl{|R|-|S|}{i}\,f_{k-i}(\kK_S)
    \label{ch:6}\\
    &=f_k(\kK_R)+\sum\limits_{\emptyset\subset{}S\subset{}R}
    \sum_{i=0}^{|R|-|S|-1}(i+1)!\,\stirl{|R|-|S|}{i+1}\,f_{k-i-1}(\kK_S),
    \notag
  \end{align} 
  where, to go from \eqref{ch:5} to \eqref{ch:6} we used Lemma 
  \ref{lem:chain_counting}(i). 
\end{proof}

Restating relations \eqref{equ:fkQR-FS} and \eqref{equ:fkQR-KS} in
terms of  generating functions, we arrive at Lemma \ref{lem:stirl}.
These relations will be used to transform \eqref{equ:fkQR-FS} and
\eqref{equ:fkQR-KS} in their $h$-vector equivalents. 
\begin{lemma} 
\label{lem:stirl}
 For all $\emptyset\subset{}R\subseteq{}[r]$ we have:
 \begin{align}
   \ft{\bx{}\qQ_R}{t}&=\ft{\fF_R}{t}+
   \sum_{\emptyset\subset{}S\subset{}R}
   \sum_{i=0}^{|R|-|S|}i!\,\stirl{|R|-|S|+1}{i+1}t^{|R|-|S|-i}\ft{\fF_S}{t},
   \label{equ:stirl1}\\
   \ft{\bx{}\qQ_R}{t}&=\ft{\kK_R}{t}+\sum_{\emptyset\subset{}S\subset{}R}
   \sum_{i=0}^{|R|-|S|-2}(i+1)!\,\stirl{|R|-|S|}{i+1}t^{|R|-|S|-i}
   \ft{\kK_S}{t}.
   \label{equ:stirl2}
 \end{align}
\end{lemma}
\begin{proof}
Using  relation \eqref{equ:fkQR-FS} and recalling that 
$\dim(\bx{}\qQ_R)=d+|R|-2$, we have:
\begin{align*}
  &\ft{\bx{}\qQ_R}{t}=\sum_{k=0}^{d+|R|-1}f_{k-1}(\bx{}\qQ_R)t^{d+|R|-1-k}\\
  &\hspace*{7mm}
  =\sum_{k=0}^{d+|R|-1}f_{k-1}(\fF_R)t^{d+|R|-1-k}
  +\sum_{k=0}^{d+|R|-1}\sum_{\emptyset\subset{}S\subset{}R}\left(
  \sum_{i=0}^{|R|-|S|}i!\,\stirl{|R|-|S|+1}{i+1}\,f_{k-1-i}(\fF_S)
  \right)t^{d+|R|-1-k}\\
  &\hspace*{7mm}
  =\ft{\fF_R}{t}+\sum_{\emptyset\subset{}S\subset{}R}\sum_{i=0}^{|R|-|S|}i!\,
  \stirl{|R|-|S|+1}{i+1}t^{|R|-|S|-i}
  \sum_{k=0}^{d+|R|-1}f_{k-i-1}(\fF_S)\,t^{d+|S|-1-k+i}\\ 
  &\hspace*{7mm}
  =\ft{\fF_R}{t}+
  \sum_{\emptyset\subset{}S\subset{}R}\sum_{i=0}^{|R|-|S|}i!\,\stirl{|R|+|S|+1}{i+1}
  t^{|R|-|S|-i}\sum_{k=i}^{d+|S|-1+i}f_{k-i-1}(\fF_S)\,t^{d+|S|-1-k+i}\\ 
  &\hspace*{7mm}
  =\ft{\fF_R}{t}+
  \sum_{\emptyset\subset{}S\subset{}R}\sum_{i=0}^{|R|-|S|}i!\,\stirl{|R|-|S|+1}{i+1}
  t^{|R|-|S|-i}\sum_{k=0}^{d+|S|-1}f_{k-1}(\fF_S)\,t^{d+|S|-1-k}\\ 
  &\hspace*{7mm}
  =\ft{\fF_R}{t}+
  \sum_{\emptyset\subset{}S\subset{}R}\sum_{i=0}^{|R|-|S|}i!\,
  \stirl{|R|-|S|+1}{i+1}t^{|R|-|S|-i}\ft{\fF_S}{t}.
\end{align*}

Analogously, converting \eqref{equ:fkQR-KS} into its generating function
equivalent, we get:
\begin{align*}
  &\ft{\bx{}\qQ_R}{t}\\
  &\hspace*{4mm}
  =\sum_{k=0}^{d+|R|-1}\hspace{-0.1cm}f_{k-1}(\kK_R)t^{d+|R|-1-k}
  +\hspace{-0.1cm}\sum_{k=0}^{d+|R|-1}
  \sum_{\emptyset\subset{}S\subset{}R}\biggl(\sum_{i=0}^{|R|-|S|-1}(i+1)!\,
  \stirl{|R|-|S|}{i+1}\,f_{k-1-i}(\kK_S)\biggr)t^{d+|R|-1-k}\\
  &\hspace*{4mm}
  =\ft{\kK_R}{t}+\sum_{\emptyset\subset{}S\subset{}R}
  \sum_{i=0}^{|R|-|S|-1}(i+1)!\,\stirl{|R|-|S|}{i+1}
  t^{|R|-|S|-i}\sum_{k=0}^{d+|R|-1}f_{k-i-1}(\kK_S)\,t^{d+|S|-1-k+i}\\
  &\hspace*{4mm}
  =\ft{\kK_R}{t}+\sum_{\emptyset\subset{}S\subset{}R}\sum_{i=0}^{|R|-|S|-1}(i+1)!\,
  \stirl{|R|-|S|}{i+1}t^{|R|-|S|-i}\sum_{k=0}^{d+|S|-1}f_{k-1}(\kK_S)\,t^{d+|S|-1-k}\\
  &\hspace*{4mm}
  =\ft{\kK_R}{t}+\sum_{\emptyset\subset{}S\subset{}R}\sum_{i=0}^{|R|-|S|-1}(i+1)!\,
  \stirl{|R|-|S|}{i+1}t^{|R|-|S|-i}\ft{\kK_S}{t},
\end{align*}
where, in order to go from the third to the fourth line,
we changed variables (in the last sum) and we used the fact that 
$f_{k-1}(\kK_S)=0$ for $k>d+|S|-1.$
\end{proof}

The $h$-vector relations stemming from the $f$-vector relations above
are the subject of the following lemma.
\begin{lemma} 
\label{lem:hQ}
For all $\emptyset\subset{}R\subseteq{}[r]$ we have:
\begin{align} 
  \htt{\bx{}\qQ_R}{t}&=\htt{\fF_R}{t}+
  \sum_{\emptyset\subset{}S\subset{}R}\sum_{j=0}^{|R|-|S|-1}
  \eul{|R|-|S|}{j}\,t^{j+1}\,\htt{\fF_S}{t},\label{hQF}\\
  \htt{\bx{}\qQ_R}{t}&=\htt{\kK_R}{t}+
  \sum_{\emptyset\subset{}S\subset{}R}\sum_{j=0}^{|R|-|S|-1}
  \eul{|R|-|S|}{j}\,t^{j}\,\htt{\kK_S}{t},\label{hQK}
\end{align} 
where $\eul{m}{k}$ are the Eulerian numbers \cite{oeis:A008292,gkp-cm-89}:
\begin{equation*}
  \eul{m}{k}=\sum_{i=0}^{k}(-1)^{i}\binom{m+1}{i}(k+1-i)^m,
  \qquad
  m\ge{}k+1>0.
\end{equation*}
\end{lemma}
\begin{proof}
Using  \eqref{hf}, \eqref{equ:stirl1} and the symmetry of Eulerian numbers, 
we get:
\begin{align*}
  \htt{\bx{}\qQ_R}{t}&=\ft{\bx{}\qQ_R}{t-1}\\
  &=\ft{\fF_R}{t-1}+\sum_{\emptyset\subset{}S\subset{}R}
  \underbrace{\sum_{i=0}^{|R|-|S|-1}i!\,
    \stirl{|R|-|S|+1}{i+1}(t-1)^{|R|-|S|-i}}_{{\rm{}cf.}~\eqref{WE1}}\ft{\fF_S}{t-1}\\
  &=\htt{\fF_R}{t}+\sum_{\emptyset\subset{}S\subset{}R}\sum_{j=0}^{|R|-|S|-1}
  \eul{|R|-|S|}{j}t^{|R|-|S|-j}\htt{\fF_S}{t}\\
  &=\htt{\fF_R}{t}+\sum_{\emptyset\subset{}S\subset{}R}\sum_{j=0}^{|R|-|S|-1}
  \eul{|R|-|S|}{|R|-|S|-1-j}t^{|R|-|S|-j}\htt{\fF_S}{t}\\
  &=\htt{\fF_R}{t}+\sum_{\emptyset\subset{}S\subset{}R}\sum_{j=0}^{|R|-|S|-1}
  \eul{|R|-|S|}{j}t^{j+1}\htt{\fF_S}{t}.
\end{align*}
Analogously, using \eqref{hf}, \eqref{equ:stirl2} 
and the symmetry of Eulerian numbers, we deduce that: 
\begin{align*}
  \htt{\bx{}\qQ_R}{t}&=\ft{\bx{}\qQ_R}{t-1}\\
  &=\ft{\kK_R}{t-1}
  +\sum_{\emptyset\subset{}S\subset{}R}\sum_{i=0}^{|R|-|S|-2}(i+1)!\,
  \stirl{|R|-|S|}{i+1}t^{|R|-|S|-i-1}\ft{\kK_S}{t-1}\\
  &=\htt{\kK_R}{t}+\sum_{\emptyset\subset{}S\subset{}R}\underbrace{
    \sum_{i=0}^{|R|-|S|-2}(i+1)!\,\stirl{|R|-|S|}{i+1}(t-1)^{|R|-|S|-i}
  }_{{\rm{}cf.}~\eqref{WE2}}\htt{\kK_S}{t}\\
  &=\htt{\kK_R}{t}+\sum_{\emptyset\subset{}S\subset{}R}\sum_{i=0}^{|R|-|S|-2}
  \eul{|R|-|S|}{i}t^{|R|-|S|-i-1}\htt{\kK_S}{t}\\
  &=\htt{\kK_R}{t}+\sum_{\emptyset{}\subset{}S\subset{}R}
  \sum_{i=0}^{|R|-|S|-2}\eul{|R|-|S|}{|R|-|S|-1-i}t^{|R|-|S|-i-1}
  \htt{\kK_S}{t}\\
  &=\htt{\kK_R}{t}+\sum_{\emptyset\subset{}S\subset{}R}\sum_{i=0}^{|R|-|S|-2}
  \eul{|R|-|S|}{i}t^{i}\htt{\kK_S}{t}.\qedhere
\end{align*}
\end{proof} 

%% file: Dehn_Som.tex
\section{The Dehn-Sommervile equations}
\label{sec:DS}

A very important structural property of the Cayley polytope $\cC_R$ is,
what we call, the \emph{Dehn-Sommervile equations}. For a single
polytope they reduce to the well-known Dehn-Sommerville equations,
whereas for two or more summands they relate the $h$-vectors of the
sets $\fF_R$ and $\kK_R$. The Dehn-Sommerville equations for $\cC_R$
are one of the major key ingredients for establishing our upper
bounds, as they permit us to reason for the maximality of the elements
of $\mb{h}(\fF_R)$ and $\mb{h}(\kK_R)$ by considering only the lower
halves of these vectors.

\begin{theorem}[Dehn-Sommerville equations]
  \label{theor:DS} 
  Let $\cC_R$ be the Cayley polytope of the $d$-polytopes 
  $P_i,i\in{}R$. Then, the following relations hold:
  \begin{equation} 
    t^{d+|R|-1}\htt{\fF_R}{\tfrac{1}{t}} = \htt{\kK_R}{t} 
    \label{equ:DF1} 
  \end{equation} 
  or, equivalently, 
  \begin{equation}
    h_{d+|R|-1-k}(\fF_{R})=h_k(\kK_{R}),
    \qquad
    0\leq{}k\leq{}d+|R|-1.
    \label{equ:DF2}
  \end{equation}
\end{theorem}

\begin{proof} 
  We prove our claim by induction on the size of  $R$, the case $|R|=1$ being 
  the Dehn-Somerville equations for a $d$-polytope. We next assume that our 
  claim holds for all $\emptyset\subset{}S\subset{}R$ and prove it for
  $R$. The ordinary Dehn-Somerville relations, written in generating
  function form, for the (simplicial) $(d+|R|-1)$-polytope $\qQ_R$
  imply that:
  \begin{equation} 
    \htt{\bx{}\qQ_R}{t}=t^{d+|R|-1}\htt{\bx{}\qQ_R}{\tfrac{1}{t}}.
    \label{prDS1} 
  \end{equation} 
  In view of  relation \eqref{hQF} of Lemma \ref{lem:hQ}, 
  the right-hand side  of \eqref{prDS1}  becomes:
  \begin{equation} 
    t^{d+|R|-1}\htt{\fF_R}{\tfrac{1}{t}}+t^{d+|R|-1}
    \sum_{\emptyset\subset{}S\subset{}R}\sum_{j=0}^{|R|-|S|-1}\eul{|R|-|S|}{j}
    t^{-j-1}\htt{\fF_S}{\tfrac{1}{t}}. \label{RR}
  \end{equation}
  Using relation \eqref{hQK}, along with the induction
  hypothesis, the left-hand side of \eqref{prDS1} becomes:
  \begin{align} 
    \htt{\kK_R}{t}&+\sum_{\emptyset\subset{}S\subset{}R} 
    \sum_{j=0}^{|R|-|S|-1}\eul{|R|-|S|}{j}t^{j}\htt{\kK_S}{t}
    \label{E1}\\
    &=\htt{\kK_R}{t}+\sum_{\emptyset\subset{}S\subset{}R} 
    \sum_{j=0}^{|R|-|S|-1}\eul{|R|-|S|}{j} t^{|R|-|S|-j-1}\htt{\kK_S}{t}
    \label{E2}\\
    &=\htt{\kK_R}{t}+\sum_{\emptyset\subset{}S\subset{}R}
    \sum_{j=0}^{|R|-|S|-1}\eul{|R|-|S|}{j} t^{|R|-|S|-j-1} t^{d+|S|-1} 
    \htt{\fF_S}{\tfrac{1}{t}}\notag\\ 
    &=\htt{\kK_R}{t}+\sum_{\emptyset\subset{}S\subset{}R}
    \sum_{j=0}^{|R|-|S|-1}\eul{|R|-|S|}{j}t^{d+|R|-j-2} 
    \htt{\fF_S}{\tfrac{1}{t}},\label{LL}
  \end{align} 
  where to go from \eqref{E1} to \eqref{E2} we changed variables and used the 
  well-known symmetry of the Eulerian numbers, namely,
  $\eul{m}{k}=\eul{m}{m-k-1}$, for all $m\ge{}k+1>0$.
	
  Now, substituting \eqref{RR} and \eqref{LL} in \eqref{prDS1}, we deduce that 
  $t^{d+|R|-1}\htt{\fF_R}{\tfrac{1}{t}}=$ $\htt{\kK_R}{t}$,
  which is, coefficient-wise, equivalent to \eqref{equ:DF2}.
\end{proof}

%% file: recurrenceF.tex
\section{\texorpdfstring{The recurrence relation for $\mb{h}(\fF_R)$}%
{The recurrence relation for h(FR)}}
\label{sec:rec}

The subject of this section is the generalization, for the $h$-vector of
$\fF_R$, $\emptyset\subset{}R\subseteq[r]$, of the recurrence relation 
\begin{equation}\label{equ:links}
  (k+1)h_{k+1}(\bx{}P)+(d-k)h_k(\bx{}P)\leq{}n\,h_{k}(\bx{}P),
  \quad 0\leq{}k\leq{}d-1,
\end{equation}
that holds true for any simplicial $d$-polytope $P\subset\reals^d$.
This is the content of the next theorem. Its proof is postponed until
Section \ref{sec:proof-hkFRrec}. In the next five subsections we
build upon the necessary intermediate results for proving this theorem.

\begin{theorem}[Recurrence inequality]
\label{thm:hkFRrec}
For any $\emptyset\subset{}R\subseteq[r]$ we have: 
\begin{equation}\label{hkFRrec}
  h_{k+1}(\fF_R)\le\frac{n_{R}-d-|R|+1+k}{k+1}h_k(\fF_R)+
  \sum_{i\in{}R}\frac{n_i}{k+1}g_k(\fF_{R\sm\{i\}}),\qquad 0\le{}k\le{}d+|R|-2,
\end{equation}
where: (1) $n_R=\sum_{i\in{}R}n_i$, $n_\emptyset=\emptyset$, and,
(2) $g_k(\fF_\emptyset)=g_k(\emptyset)=0$, for all $k$.
\end{theorem}

\subsection{\texorpdfstring{Relating the $h$-vector of $\qQ_R/v$ with the 
$h$-vectors of $\fF_R/v$ and $\kK_R/v$}{Relating the h-vector of QR/v
    with the h-vectors of FR/v and KR/v}}
\label{sec:hFR-hFRv}

For any $\emptyset\subset{}R\subseteq{}[r]$, let $V_R:=\cup_{i\in{}R}V_i$.
We define the link of a vertex $v\in{}V_R$ in $\fF_R$ as the intersection of
the link $\kK_R/v$ with $\fF_R$. The following lemma relates the
$h$-vector of $\qQ_R/v$ with the $h$-vectors of $\fF_R/v$ and $\kK_R/v$.

\begin{lemma} \label{lem:link Q_R}
  For any $v\in{}V_R$ we have:
  \begin{equation}
    \htt{\bx{}\qQ_R/v}{t}=\htt{\fF_R/v}{t}+
    \sum_{\{i\}\subseteq{}S\subset{}R}\sum_{j=0}^{|R|-|S|-1}
    \eul{|R|-|S|}{j}t^{j+1}\htt{\fF_S/v}{t},
    \label{hQF/v}
  \end{equation}
  and
  \begin{equation}
    \htt{\bx{}\qQ_R/v}{t}=\htt{\kK_R/v}{t}+
    \sum_{\{i\}\subseteq{}S\subset{}R}
    \sum_{j=0}^{|R|-|S|-1}\eul{|R|-|S|}{j}t^{j}\htt{\kK_S/v}{t}.
    \label{hQK/v}
  \end{equation} 
\end{lemma} 
\begin{proof}
Let us fix some $v\in{}V_j$, $j\in{}R$. In view of relation 
\eqref{equ:QR-FSjoin} in Lemma \ref{lem:QR-FSjoin} we can write:
\begin{equation} 
  \bx{}\qQ_R/v=\bigcup\limits_{\emptyset\subset{}S\subseteq{}R}
  \fF_S/v\mbox{\hspace{-0.5cm}}
  \bigcup_{\scalebox{0.7}{$\substack{\emptyset\subset{}S\subset{}R\\
	S\subseteq{}S_1{}\subset{}S_2\subset{}\cdots\subset{}S_\ell\subset{}R}$}}
  \mbox{\hspace{-0.4cm}}\{y_{S_1},y_{S_2},\ldots,y_{S_\ell},\fF_S/v\}, 
  \label{recal3.2}
\end{equation}
where it is understood that both $\fF_S/v$ and 
$\{y_{S_1},y_{S_2},\ldots,y_{S_\ell},\fF_S/v\}$ are empty if 
$v\not\in{}V_S$. Taking this into account, we simplify \eqref{recal3.2} as 
follows:
\begin{equation}
  \bx{}\qQ_R/v=\bigcup_{\scalebox{0.7}{$\{j\}\subseteq{}S\subseteq{}R$}}
  \fF_S/v\bigcup_{\scalebox{0.7}{$\substack{\{j\}\subseteq{}S\subset{}R\\
	S\subseteq{}S_1	\subset{}S_2\subset\cdots\subset{}S_\ell\subset{}R}$}}
  \{y_{S_1},y_{S_2},\ldots,y_{S_\ell},\fF_S/v\}.
	\label{tag:Q/v}
\end{equation}
Since each auxiliary point of a face  in
$\{y_{S_1},y_{S_2},\ldots,y_{S_\ell},\fF_S/v\}$ increases the dimension by one,
from \eqref{tag:Q/v} we can  write : 
\begin{align*} 
	f_k(\bx{}\qQ_R/v)
	&=
	\sum_{\{j\}\subseteq{}S\subseteq{}R}f_{k}(\fF_S/v)
	+\sum_{\{j\}\subseteq{}S\subset{}R}\sum_{i=1}^{|R|-|S|}\,
	\sum_{S\subseteq{}S_1\subset{}S_2\subset\cdots\subset{}S_{i}\subset{}R}    
	f_{k-i}(\fF_S/v).
\end{align*}
In view of Lemma \ref{lem:chain_counting}(i) the above can be written as:
\begin{align*}
	f_k(\bx{}\qQ_R/v)&=\sum_{\{j\}\subseteq{}S\subseteq{}R}f_{k}(\fF_S/v)
	+\sum_{\{j\}\subseteq{}S\subset{}R}\sum_{i=1}^{|R|-|S|}i!\,
	\stirl{|R|-|S|+1}{i+1}\,f_{k-i}(\fF_S/v)\\
	&=f_k(\fF_R/v)+\sum_{\{j\}\subseteq{}S\subset{}R}\sum_{i=0}^{|R|-|S|}i!\,
		\stirl{|R|-|S|+1}{i+1}\,f_{k-i}(\fF_S/v),
\end{align*}
where in the last step we used the fact that $\stirl{m}{1}=1$
for all $m\geq{}1$. 

Recalling that $\dim(\fF_S/v)=d+|S|-3$ and converting the above relation into 
generating function we get:
\begin{equation}
	\ft{\bx{}\qQ_R/v}{t}=\ft{\fF_R/v}{t}+\sum_{\{j\}\subseteq{}S\subset{}R}
	\sum_{i=0}^{|R|-|S|}i!\,\stirl{|R|-|S|+1}{i+1} t^{|R|-|S|-i}\ft{\fF_S/v}{t}.
	\label{equ:stirl1/v}
\end{equation}
We thus have: 
\begin{align}
  \htt{\bx{}\qQ_R/v}{t}&=\ft{\bx{}\qQ_R/v}{t-1}\notag\\
  &=\ft{\fF_R/v}{t-1}+\sum_{\{j\}\subseteq{}S\subset{}R}	
  \sum_{i=0}^{|R|-|S|}i!\,\stirl{|R|-|S|+1}{i+1}
  t^{|R|-|S|-i}\ft{\fF_S/v}{t-1}
  \notag\\
  &=\htt{\fF_R/v}{t}+\sum_{\{j\}\subseteq{}S\subset{}R}\sum_{i=0}^{|R|-|S|}i!\,
  \stirl{|R|-|S|+1}{i+1}(t-1)^{|R|-|S|-i}\ft{\fF_S/v}{t-1}\notag\\
  &=\htt{\fF_R/v}{t}+\sum_{\{j\}\subseteq{}S\subset{}R}
  \sum_{i=0}^{|R|-|S|}i!\,\stirl{|R|-|S|+1}{i+1}(t-1)^{|R|-|S|-i}\htt{\fF_S/v}{t}
  \label{S-E:1}\\
  &=\htt{\fF_R/v}{t}+\sum_{\{j\}\subseteq{}S\subset{}R}\sum_{j=0}^{|R|-|S|-1}
  \eul{|R|-|S|}{j}t^{|R|-|S|-j}\htt{\fF_S/v}{t}\label{S-E:2}\\
  &=\htt{\fF_R/v}{t}+\sum_{\{j\}\subseteq{}S\subset{}R}\sum_{j=0}^{|R|-|S|-1}
  \eul{|R|-|S|}{|R|-|S|-1-j}t^{|R|-|S|-j}\htt{\fF_S/v}{t}\notag\\
  &=\htt{\fF_R/v}{t}+\sum_{\{j\}\subseteq{}S\subset{}R}\sum_{j=0}^{|R|-|S|-1}
  \eul{|R|-|S|}{j}t^{j+1}\htt{\fF_S/v}{t},\notag
\end{align}
where to go from \eqref{S-E:1} to \eqref{S-E:2} we used relation
\eqref{WE1} from Lemma \ref{stirl-euler}.

Let us now turn our attention to relation \eqref{hQK/v} .
In view of \eqref{equ:QR-KSjoin} of Lemma \ref{lem:QR-KSjoin} we have:
\begin{align*}
	\bx{}\qQ_R/v=\kK_{R}/v
	\bigcup_{\substack{\{j\}\subseteq{}S\subset{}R\\
	S=S_1\subset{}S_2\subset\cdots\subset{}S_\ell\subset{}R}}
	\{y_{S_1},y_{S_2},\ldots,y_{S_\ell},\kK_S/v\}, 
\end{align*}
which in turn gives
\begin{align*} 
  f_k(\bx{}\qQ_R/v)
  &=f_k(\kK_R/v)+\sum_{\{j\}\subseteq{}S\subset{}R}\sum_{i=1}^{|R|-|S|}\,
  \sum_{S=S_1\subset{}S_2\subset\cdots\subset{}S_{i}\subset{}R}
  f_{k-i}(\kK_S/v)\\
  &=f_k(\kK_R/v)+\sum_{\{j\}\subseteq{}S\subset{}R}\sum_{i=1}^{|R|-|S|}i!\,
  \stirl{|R|-|S|}{i}f_{k-i}(\kK_S/v)\\
  &=f_k(\kK_R/v)+\sum_{\{j\}\subseteq{}S\subset{}R}
  \sum_{i=0}^{|R|-|S|-1}(i+1)!\,\stirl{|R|-|S|}{i+1}f_{k-i-1}(\kK_S/v).   
\end{align*}
Recalling that $\dim(\kK_S/v)=d+|S|-3$ and converting the above relation into 
generating function, we get:
\begin{align}
	\ft{\bx{}\qQ_R/v}{t}
	&=\ft{\kK_R/v}{t}+\sum_{\{j\}\subseteq{}S\subset{}R}
	\sum_{i=0}^{|R|-|S|-1}(i+1)!\,\stirl{|R|-|S|}{i+1}t^{|R|-|S|-i}\,
	\ft{\kK_S/v}{t},
\end{align}
which further implies that
\begin{align}
  \htt{\bx{}\qQ_R/v}{t}&=\ft{\bx{}\qQ_R/v}{t-1}\notag\\
  &=\ft{\kK_R/v}{t-1}+\sum_{\{j\}\subseteq{}S\subset{}R}
  \sum_{i=0}^{|R|-|S|-1}(i+1)!\,\stirl{|R|-|S|}{i+1}(t-1)^{|R|-|S|-1-i}
  \ft{\kK_S/v}{t-1}\label{S-E:3}\\
  &=\htt{\kK_R/v}{t}+\sum_{\{j\}\subseteq{}S\subset{}R}\sum_{j=0}^{|R|-|S|-1}
  \eul{|R|-|S|}{i}t^{|R|-|S|-1-i}\htt{\kK_S/v}{t}\label{S-E:4}\\
  &=\htt{\kK_R/v}{t}+\sum_{\{j\}\subseteq{}S\subset{}R}\sum_{j=0}^{|R|-|S|-1}
  \eul{|R|-|S|}{|R|-|S|-1-j}t^{|R|-|S|-1-j}\,\htt{\kK_S/v}{t}\notag\\
  &=\htt{\kK_R/v}{t}+\sum_{\{j\}\subseteq{}S\subset{}R}\sum_{j=0}^{|R|-|S|-1}
  \eul{|R|-|S|}{j}t^{j}\,\htt{\kK_S/v}{t},\notag
\end{align}
where to go from \eqref{S-E:3} to \eqref{S-E:4} we used \eqref{WE2}
from Lemma \ref{stirl-euler}.
\end{proof}

\subsection{\texorpdfstring{The link of $y_S$ in $\bx\qQ_R$}%
{The link of yS in QR}}
\label{sec:linkYS}

Our next goal is to find an expression analogous to those of 
Lemma \ref{lem:link Q_R}, but now involving links of type $\bx{}\qQ_R/y_S$, 
where $\emptyset\subset{}S\subseteq{}R.$ To do this, we first need to express 
$f_k(\bx{}\qQ_R/y_S) $ in terms of sums of $f_i(\fF_X)$ with $i\leq{}k$ and 
$X\subseteq{}S$. This is the content of the next Lemma.
In order to state it we need to introduce a new set. Let
$X\subseteq{}T\subset{}R$ and $\ell$ be a positive integer. We define
the set
\begin{equation} 
  \mathcal{D}(R,T,X,\ell):=\{(S_1,\ldots,S_\ell):
  X\subseteq{}S_1\subset{}S_2\subset\cdots\subset{}S_{\ell}\subset{}R
  \mbox{ and }S_i=T\mbox{ for some }1\leq{}i\leq\ell\},
\end{equation} 
and denote by $\D{R}{T}{X}{\ell}$ its cardinality. 

\begin{lemma}
\label{lem:link{}of{}yS}
For every $\emptyset\subset{}S\subset{}R$ we have:
\begin{align}\label{equ:link{}of{}yS}
\ft{\bx{}\qQ_R/y_S}{t}=\sum_{\emptyset\subset{}X\subseteq{}S}\sum_{\ell=1}^{|R|-|X|}
\D{R}{S}{X}{\ell}\,t^{|R|-|X|-\ell}\ft{\fF_X}{t}.
\end{align} 
\end{lemma} 
\begin{proof}
First of all, notice that, in view of relation \eqref{equ:QR-FSjoin}, 
if we denote by $y_S\ast\bx\qQ_R$ the set of all faces in 
$\bx{}\qQ_R$ containing $y_S$, we have:
\begin{equation*}
	y_S\ast\bx\qQ_R=	
	\bigcup_{\scalebox{0.7}{$\substack{\emptyset\subset{}X\subseteq{}S\\
	X\subseteq{}S_1{}\subset{}S_2\subset{}\cdots\subset{}S_\ell\subset{}R\\
	S_i=S\,\mbox{\tiny  for some }1\leq{}i\leq\ell}$}}
	\mbox{\hspace{-0.4cm}}\{y_{S_1},y_{S_2},\ldots,y_{S_\ell},\fF_X\}.
\end{equation*}
Then clearly, 
\begin{align*}
	f_k(\bx{}\qQ_R/y_S)=f_{k+1}(y_S\ast\bx\qQ_R)&=
	\sum_{\emptyset\subset{}X\subseteq{}S}
	\sum_{\ell=1}^{|R|-|X|}\sum_{\substack{X\subseteq{}S_1\subset{}S_2\subset
	\cdots\subset{}S_\ell\subset{}R\\S_i=S\,\mbox{\tiny  for some }
	1\leq{}i\leq\ell}}f_{k-\ell+1}(\fF_X)\\
	&=\sum_{\emptyset\subset{}X\subseteq{}S}
	\sum_{\ell=1}^{|R|-|X|}\D{R}{S}{X}{\ell}\,f_{k-\ell+1}(\fF_X).
\end{align*}   
Using the fact that $\dim(\bx{}\qQ/y_S)=d+|R|-3$ and rewriting in terms of 
generating functions, the above becomes: 
\begin{align*}
  \ft{\bx{}\qQ_R/y_S}{t}&=\sum_{k=0}^{d+|R|-|S|-1}
  \sum_{\emptyset\subset{}X\subseteq{}S}\sum_{\ell=1}^{|R|-|X|}
  \D{R}{S}{X}{\ell}\,f_{k-\ell+1}(\fF_X)t^{d+|R|-2-k}\\
  &=\sum_{k=0}^{d+|R|-|S|-1}\sum_{\emptyset\subset{}X\subseteq{}S}
  \sum_{\ell=1}^{|R|-|X|}t^{|R|-|X|-\ell}\D{R}{S}{X}{\ell}\,
  f_{k-\ell+1}(\fF_X)t^{d+|X|-1-(k-\ell+1)}\\
  &=\sum_{\emptyset\subset{}X\subseteq{}S}\sum_{\ell=1}^{|R|-|X|}\!\!\!\!
  \D{R}{S}{X}{\ell}t^{|R|-|X|-\ell}\sum_{k-\ell+1=|X|-|R|+1}^{d+|R|-|S|-1}
  \!\!\!f_{k-\ell+1}(\fF_X)t^{d+|X|-1-(k-\ell+1)}\\
  &=\sum_{\emptyset\subset{}X\subseteq{}S}\sum_{\ell=1}^{|R|-|X|}
  \D{R}{S}{X}{\ell}t^{|R|-|X|-\ell}\sum_{j=0}^{d+|R|-|S|-1}f_{j}(\fF_X)
  t^{d+|X|-1-j}\\
  &=\sum_{\emptyset\subset{}X\subseteq{}S}\sum_{\ell=1}^{|R|-|X|}
  \D{R}{S}{X}{\ell}t^{|R|-|X|-\ell}\sum_{j=0}^{d+|R|-|X|-1}f_{j}(\fF_X)
  t^{d+|X|-1-j}\\
  &=\sum_{\emptyset\subset{}X\subseteq{}S}\sum_{\ell=1}^{|R|-|X|}
  \D{R}{S}{X}{\ell}t^{|R|-|X|-\ell}\ft{\fF_X}{t}.\qedhere
\end{align*}
\end{proof}


\noindent
Converting relation \eqref{equ:link{}of{}yS} of the above lemma
to its $h$-vector equivalent we get:
\begin{equation}
  \label{equ:link{}of{}yS-h}
  \begin{aligned}
    \htt{\bx{}\qQ_R/y_S}{t}
    &=\ft{\bx{}\qQ_R/y_S}{t-1}\\
    &=\sum_{\emptyset\subset{}X\subseteq{}S}
    \sum_{\ell=1}^{|R|-|X|}\D{R}{S}{X}{\ell}(t-1)^{|R|-|X|-j}\htt{\fF_X}{t}.
  \end{aligned}
\end{equation}

The following lemma expresses the sum of the $h$-vectors of the links
$\qQ_R/y_S$ to the $h$-vectors of the sets $\fF_X$.
\begin{lemma} 
  \label{lem:hQR/S}
  For every $\emptyset\subset{}R\subseteq{}[r]$ we have:
  \begin{equation}
    \sum_{\emptyset\subset{}S\subset{}R}\htt{\bx{}\qQ_R/y_S}{t}=
    \sum_{\emptyset\subset{}X\subset{}R}\sum_{j=0}^{|R|}
    \bigl(\eul{|R|-|X|+1}{j}-\eul{|R|-|X|}{j} 
    \bigr)t^{|R|-|X|-j}\htt{\fF_X}{t}.
    \label{hQR/S}
  \end{equation} 
\end{lemma}
\begin{proof}
By means of relation \eqref{equ:link{}of{}yS-h}, the sum
$\sum_{\emptyset\subset{}S\subset{}R}\htt{\bx{}\qQ_R/y_S}{t}$ is equal to:
\begin{align}
  \sum_{\emptyset\subset{}X\subseteq{}S\subset{}R}&
  \sum_{\ell=1}^{|R|-|X|}\D{R}{S}{X}{\ell}(t-1)^{|R|-|X|-l}\htt{\fF_X}{t}
  \notag\\
  =&\sum_{\ell=1}^{|R|}\sum_{\emptyset\subset{}X\subset{}R}
  \sum_{X\subseteq{}S\subset{}R}\D{R}{S}{X}{\ell}(t-1)^{|R|-|X|-l}
  \,\htt{\fF_X}{t}\label{Q/y_S:1}\\
  =&\sum_{\ell=1}^{|R|}
  \sum_{\emptyset\subset{}X\subset{}R}\ell\,\ell!\,\stirl{|R|-|X|+1}{\ell+1}
  (t-1)^{|R|-|X|-l}\htt{\fF_X}{t}\label{Q/y_S:2}\\
  =&\sum_{\emptyset\subset{}X\subset{}R}\biggl(\sum_{\ell=0}^{|R|}\ell\,\ell!\,
  \stirl{|R|-|X|+1}{\ell+1}(t-1)^{|R|-|X|-\ell}\biggr)\htt{\fF_X}{t}
  \notag\\
  =&\sum_{\emptyset\subset{}X\subset{}R}
  \biggl(\sum_{\ell=0}^{|R|}(\ell+1)\,\ell!\,
  \stirl{|R|-|X|+1}{\ell+1}(t-1)^{|R|-|X|-\ell}
  -\sum_{\ell=1}^{|R|}\stirl{|R|-|X|+1}{\ell+1}(t-1)^{|R|-|X|-l}\biggr)\,      
  \htt{\fF_X}{t}\label{Q/y_S:3}\\
  =&\sum_{\emptyset\subset{}X\subset{}R}\biggl(\sum_{j=0}^{|R|}
  \eul{|R|-|X|+1}{j}(t-1)^{|R|-|X|-j}-\sum_{j=0}^{|R|}\eul{|R|-|X|}{j}(t-1)^{|R|-|X|-j}
  \biggr)\,\htt{\fF_X}{t}\label{Q/y_S:4}\\
  =&\sum_{\emptyset\subset{}X\subset{}R}\;
  \sum_{j=0}^{|R|}\bigl(\eul{|R|-|X|+1}{j}-\eul{|R|-|X|}{j}\bigr)
  (t-1)^{|R|-|X|-j}\htt{\fF_X}{t},\notag
\end{align}
where, to go from \eqref{Q/y_S:1} to \eqref{Q/y_S:2} and from
\eqref{Q/y_S:3} to \eqref{Q/y_S:4}  we used Lemma \ref{lem:D-S} and
Lemma\;\ref{stirl-euler}, respectively.
\end{proof}

\subsection{Links and non-links}
\label{sec:links-nonlinks}

The following theorem generalizes Lemma \ref{lem:McM} in the context
of Cayley polytopes.
\begin{theorem}
\label{theor:rec1} 
For any $\emptyset\subset{}R\subseteq[r]$, 
\begin{equation} 
  (d+|R|-1)\htt{\fF_R}{t}+(1-t)\hh{\fF_R}=\sum_{v\in{}V_R}\htt{\fF_R/v}{t},
  \label{equ:rec1}
\end{equation} 
where $V_R=\cup_{i\in{}R}V_i$. 
\end{theorem}

\begin{proof}
We proceed by induction on the size of $R$. 
The case $|R|=1$ is considered in Lemma \ref{lem:McM}. 
Assume now that \eqref{equ:rec1} holds for all 
$\emptyset\subset{}S\subset{}R$. By applying Lemma \ref{lem:McM} to
the simplicial polytope $\qQ_R$ we have:
\begin{equation}\label{equ:McM-QR}
  (d+|R|-1)\,\h{\bx{\qQ_R}}+(1-t)\hh{\bx{\qQ_R}}
  =\sum_{v\in\text{vert}(\bx{\qQ_R})}\h{\bx{\qQ_R}/v}.
\end{equation}
Recall from Lemma \ref{lem:hQ} that: 
\begin{align}
  \htt{\bx{}\qQ_R}{t}=\htt{\fF_R}{t}+\sum_{\emptyset\subset{}S\subset{}R}
  \sum_{j=0}^{|R|-|S|-1}\eul{|R|-|S|}{j}t^{|R|-|S|-j}\htt{\fF_S}{t}. 
  \label{main1}
\end{align}
Multiplying both sides of \eqref{main1} by $d+|R|-1$  we get: 
\begin{align*} 
  (d+|R|-1)\,&\h{\bx{}\qQ}=\\
  &(d+|R|-1)\htt{\fF_R}{t}+(d+|R|-1)
  \sum_{\emptyset{}\subset{}S\subset{}R}\sum_{j=0}^{|R|-|S|-1}\eul{|R|-|S|}{j}
  t^{|R|-|S|-j}\htt{\fF_S}{t}
\end{align*} 
Differentiating both sides of \eqref{main1} and multiplying by $(1-t)$ 
we get: 
\begin{align*}
  (1-t)\hh{\bx{}\qQ_R}=(1-t)&\hh{\fF_R}\\
  &+(1-t)\sum_{\emptyset\subset{}S\subset{}R}
  \sum_{j=0}^{|R|-|S|-1}(|R|-|S|-j)\eul{|R|-|S|}{j}t^{|R|-|S|-j-1}\htt{\fF_S}{t}\\
  &+(1-t)\sum_{\emptyset\subset{}S\subset{}R}
  \sum_{j=0}^{|R|-|S|-1}\eul{|R|-|S|}{j}t^{|R|-|S|-j}\hh{\fF_S}.
\end{align*} 
Summing up the above two relations and using Lemma \ref{McM} for the 
$(d+|R|-1)$-polytope $\qQ_R$, we conclude that the right-hand side of
\eqref{equ:McM-QR} is equal to:
\begin{align*}
  (d&+|R|-1)\,\h{\bx{}\qQ_R}+(1-t)\hh{\bx{}\qQ_R}\\
  &=(d+|R|-1)\htt{\fF_R}{t}+(d+|R|-1)
  \sum_{\emptyset\subset{}S\subset{}R}\sum_{j=0}^{|R|-|S|-1}
  \eul{|R|-|S|}{j}t^{|R|-|S|-j}\htt{\fF_S}{t}\\
  &\qquad+(1-t)\hh{\fF_R}+(1-t)\sum_{\emptyset\subset{}S\subset{}R}
  \sum_{j=0}^{|R|-|S|-1}(|R|-|S|-j)\eul{|R|-|S|}{j}t^{|R|-|S|-j-1}\htt{\fF_S}{t}\\
  &\qquad+(1-t)\sum_{\emptyset\subset{}S\subset{}R}\sum_{j=0}^{|R|-|S|-1}
  \eul{|R|-|S|}{j}t^{|R|-|S|-j}\\
  &=A+\sum_{\emptyset\subset{}S\subset{}R}
  \sum_{j=0}^{|R|-|S|-1}(d+|S|-1)\eul{|R|-|S|}{j}t^{|R|-|S|-j}\htt{\fF_S}{t}\\
  &\qquad+\sum_{\emptyset\subset{}S\subset{}R}\sum_{j=0}^{|R|-|S|-1}j\,
  \eul{|R|-|S|}{j}t^{|R|-|S|-j}\htt{\fF_S}{t}\\
  &\qquad+\sum_{\emptyset\subset{}S\subset{}R}
  \sum_{j=0}^{|R|-|S|-1}(|R|-|S|-j)\eul{|R|-|S|}{j}t^{|R|-|S|-j-1}\htt{\fF_S}{t}\\
  &\qquad+(1-t)\sum_{\emptyset\subset{}S\subset{}R}
  \sum_{j=0}^{|R|-|S|-1}\eul{|R|-|S|}{j}t^{|R|-|S|-j}\hh{\fF_S},
\end{align*}
where $A=(d+|R|-1)\htt{\fF_R}{t}+(1-t)\hh{\fF_R}$. 
In order to use our induction hypothesis, we regroup the terms of the above 
expression as follows:
\begin{align*}
  A&+\sum_{\emptyset\subset{}S\subset{}R}
  \sum_{j=0}^{|R|-|S|-1}\eul{|R|-|S|}{j}\biggl((d+|S|-1)
  \htt{\fF_S}{t}-(1-t)\hh{\fF_S}\biggr)t^{|R|-|S|-j}\\
  &\quad+\sum_{\emptyset\subset{}S\subset{}R}\sum_{j=0}^{|R|-|S|-1}
  \biggl((j+1)\eul{|R|-|S|}{j+1}+(|R|-|S|-j)\eul{|R|-|S|}{j}\biggr)
  t^{|R|-|S|-j-1}\htt{\fF_S}{t}.
\end{align*}
Using the well known recurrence relation for the Eulerian numbers
(cf. \cite{gkp-cm-89}):
\[\eul{m}{i}=(m-i)\,\eul{m-1}{i}+(i+1)\,\eul{m-1}{i},\]
and the induction hypothesis, the above expression simplifies to:
\begin{equation}\label{main5}
  \begin{aligned}
    A&+\sum_{\emptyset\subset{}S\subset{}R}\sum_{j=0}^{|R|-|S|-1}
    \sum_{v\in{}V_S}\eul{|R|-|S|}{j}t^{|R|-|S|-j}\htt{\fF_S/v}{t}\\
    &\quad+\sum_{\emptyset\subset{}S\subset{}R}\sum_{j=0}^{|R|-|S|-1}
    \biggl(\eul{|R|-|S|+1}{j+1}-\eul{|R|-|S|}{j+1}\biggr)
    t^{|R|-|S|-j-1}\htt{\fF_S}{t}.
  \end{aligned}
\end{equation} 
Since the vertices of $\qQ_R$ are either vertices of some polytope 
$P_i,i\in{}R$, or auxiliary points $y_S$, $\emptyset\subset{}S\subseteq{}R$,
we split the sum in the right-hand side of \eqref{equ:McM-QR} as follows: 
\begin{equation*}
  \sum_{v\in\text{vert}(\bx{}\qQ_R)}\h{\bx{}\qQ_R/v}=
  \sum_{v\in{}V_R}\h{\bx{}\qQ_R/v}+\sum_{S\subset{}R}\h{\bx{}\qQ_R/y_S}
\end{equation*}
Using relations \eqref{hQR/S} and \eqref{hQF/v}, the
right-hand side of the above equation is equal to:
\begin{align}
  \underbrace{\sum_{v\in{}V_R}\htt{\fF_R/v}{t}}_{B}&+
  \sum_{\underset{v\in{}V_i}{i\in{}R}}
  \sum_{\{i\}\subseteq{}S\subset{}R}\sum_{j=0}^{|R|-|S|-1}\eul{|R|-|S|}{j}t^{j+1}
  \htt{\fF_S/v}{t}\notag\\
  &\qquad\qquad+\sum_{\emptyset\subset{}X\subset{}R}\;\sum_{\ell=0}^{|R|}
  \bigl(\eul{|R|-|X|+1}{\ell}-\eul{|R|-|X|}{\ell}\bigr)
  t^{|R|-|X|-\ell}\htt{\fF_X}{t}\notag\\
  =B+\sum_{\emptyset\subset{}S\subset{}R}&\sum_{v\in{}V_S}\;
  \sum_{j=0}^{|R|-|S|-1}\eul{|R|-|S|}{j}t^{j+1}\htt{\fF_S/v}{t}\notag\\
  &+\sum_{\emptyset\subset{}X\subset{}R}\;\sum_{\ell=0}^{|R|}\;
  \bigl(\eul{|R|-|X|+1}{\ell}-\eul{|R|-|X|}{\ell}\bigr)t^{|R|-|X|-l}
  \htt{\fF_X}{t}\notag\\
  =B+ \sum_{\emptyset\subset{}S\subset{}R} &
  \sum_{v\in{}V_S}\;\sum_{j=0}^{|R|-|S|-1}
  \underbrace{\eul{|R|-|S|}{|R|-|S|-j-1}}_{\eul{|R|-|S|}{j}}t^{|R|-|S|-j}
  \htt{\fF_S/v}{t}\notag\\
  &+\sum_{\emptyset\subset{}X\subset{}R}\;\sum_{\ell=0}^{|R|-|S|}
  \bigl(\eul{|R|-|X|+1}{\ell+1}-\eul{|R|-|X|}{\ell+1}\bigr)
  t^{|R|-|X|-\ell-1}\htt{\fF_X}{t}.\label{main6}
\end{align}
Equating \eqref{main5} and \eqref{main6} we conclude that  $A=B,$ 
which is precisely relation \eqref{equ:rec1}.
\end{proof} 

Comparing coefficients in \eqref{equ:rec1} we conclude the following:
\begin{corollary} 
\label{cor:rec}
  For any  $\emptyset\subset{}R\subseteq{}[r]$
  and all $0\leq{}k\leq{}d+|R|-2$ we have:
  \begin{equation}\label{recur-relation-FR-wrt-F/v}
    (k+1)h_{k+1}(\fF_R)+(d+|R|-1-k)h_k(\fF_R)=\sum_{v\in{}V_R} h_k(\fF_R/v),
  \end{equation}
  or equivalently 
  \begin{align}\label{recur-relation-FR-wrt-K}
    (k+1)h_{k+1}(\fF_R)+(d+|R|-1-k)h_k(\fF_R)
    &=\sum_{\emptyset\subset{}S\subseteq{}R}(-1)^{|R|-|S|}
    \sum_{v\in{}V_S}g_k^{(|R|-|S|)}(\kK_S/v), 
  \end{align}
  where $V_S=\cup_{i\in{}S}V_i$.
\end{corollary} 

\begin{proof}
  Relation \eqref{recur-relation-FR-wrt-F/v} is immediate from
  \eqref{equ:rec1}; it suffices to compare the coefficients of the
  generating functions of left- and right-hand sides of \eqref{equ:rec1}.

  To go from \eqref{recur-relation-FR-wrt-F/v} to
  \eqref{recur-relation-FR-wrt-K} we use the Inclusion-Exclusion
  principle, and notice that $\kK_S/v$ is the empty set for
  $v\nin\kK_S$:
  \begin{align*}
    \sum_{v\in{}V_R} h_k(\fF_R/v)
    &=\sum_{v\in{}V_R}
    \sum_{\emptyset\subset{}S\subseteq{}R}(-1)^{|R|-|S|}g_k^{(|R|-|S|)}(\kK_S/v)\\
    &=\sum_{\emptyset\subset{}S\subseteq{}R}(-1)^{|R|-|S|}
    \sum_{v\in{}V_R}g_k^{(|R|-|S|)}(\kK_S/v)\\
    &=\sum_{\emptyset\subset{}S\subseteq{}R}(-1)^{|R|-|S|}
    \sum_{v\in{}V_S}g_k^{(|R|-|S|)}(\kK_S/v).\qedhere
  \end{align*}
\end{proof}

\subsection{Using shellings to bound the
  \texorpdfstring{$g$}{g}-vectors of links}
\label{sec:shellings}

The main result of this subsection is Theorem \ref{thm:link}, which is 
essential for proving the recursive relation in Theorem \ref{thm:hkFRrec}.  
Before proving it, some more lemmas are in order. 
The first two (Lemmas \ref{lem:in-degree} and 
\ref{lem:ap_link_nonlink}) concern inequalities of $h$-vectors, 
which are proved using their interpretation as in-degrees of the dual graph
of shellable simplicial complexes (cf. \cite{k-swtsp-88}).
The third (Lemma \ref{lem:special_shelling}) shows that there exists a 
particular shelling of the polytope $\bx{}\qQ_R$, for which the previous two 
lemmas are applicable.

We start with some definitions. 
\begin{definition}
  \label{def:shell1}
  Let  $\cC$ be a pure $d$-dimensional complex. A shelling of $\cC$ 
  is a linear ordering $F_1,\ldots,F_s$ of its facets such that  either $\cC$ 
  is $0$-dimensional, or it satisfies the following conditions:
  \begin{enumerate}[(a)]
  \item the boundary complex $\bx{}F_1$ of the first facet has a shelling,
  \item for $1<j\le{}s$ the intersection of the facet $F_j$ with the 		
    previous facets is nonempty and is a beginning segment of a shelling of 
    the	$(d-1)$-dimensional boundary complex $\bx{}F_j$, that is 
    $F_j\bigcap_{i=1}^{j-1}F_i=G_1\cup{}G_2\cup\cdots\cup{}G_r$  
    for some shelling $G_1,\ldots,G_r,\ldots,G_t$ of $\bx{}F_j.$
  \end{enumerate}
  A complex is \emph{shellable} if it is pure and has a shelling.
\end{definition}


\begin{definition}
	The dual graph $\dgr{\cC}$ of a shellable simplicial $d$-complex $\cC$ is 
	the graph   whose vertices are the maximal simplices (i.e.,	facets) and 
	whose edges correspond to adjacent facets. 
    If, in addition, we consider a linear ordering  $F_1,\ldots,F_{\ell}$ of the
    facets of $\cC$, we can impose an orientation on the graph $\dgr{\cC}$ as 
    follows: an edge connecting two facets $F_i,F_j$ is oriented from $F_i$ to 
    $F_j$ if $F_i$ precedes $F_j$ in the above order. 
\end{definition}

In the case where $\cC$ is shellable, the $h$-vector of $\cC$ encodes 
information about the in-degrees of the dual graph $\dgr{\cC}$. This is the 
content of the next theorem.  
\begin{theorem}{\rm\cite{k-swtsp-88}} 
     Let $\cC$ be a shellable simplicial $d$-complex and consider the dual graph
	 $\dgr{\cC}$ of $\cC$ oriented according to a shelling order of the facets
	 of $\cC.$ Then, $h_k(\cC),\,0\leq{}k\leq{}d,$ counts the number of 
	 vertices of the dual graph of $\cC$ with in-degree $k$ {\rm(}and is 
	 independent of  the shelling chosen{\rm)}.
\label{theor:Kalai}
\end{theorem}

Let $\Sh$ be a shellable simplicial complex and assume that 
$F_1,\ldots,F_\ell,F_{\ell+1},\ldots,F_s$ is a shelling order of its facets. 
Let $\SA$ be the subcomplex of $\Sh$ whose facets are  $F_1,\ldots,F_\ell$. 
Clearly, $\SA$ is shellable  as an initial  segment of  a shelling of  $\Sh.$
Consider now the set $\SB$ containing all faces in $\Sh\sm\SA.$ Notice 
that $\SB$ has no complex structure since it contains the facets 
$F_{\ell+1},\ldots,F_{s}$ but not all their subfaces. We can however naturally 
define its $f$-vector and, since all its maximal faces are facets of $\Sh$, 
make the convention that $\dim(\SB)=\dim(\Sh)$. Moreover, as the following 
lemma suggests, the $h$-vector of $\SB$ admits a combinatorial interpretation. 
 
\begin{lemma} 
	\label{lem:in-degree}
	$h_k(\SB)$ counts the number of vertices in $\dgr{\Sh}\sm\dgr{\SA}$
	of in-degree $k.$
\end{lemma}
\begin{proof}
	In view of Theorem \ref{theor:Kalai} we have
	that: {\sf(i)} $h_k(\Sh)$ counts the number of vertices of the dual graph 
	$\dgr{\Sh}$ of $\Sh$ with in-degree $k$ and {\sf(ii)} $h_k(\SA)$ counts the 
	number of vertices of the dual graph $\dgr{\SA}$ of $\SA$ with in-degree  
	$k$. However, since the facets in $\SA$ are an initial segment of a 
	shelling of $\Sh$, their in-degree in $\dgr{\SA}$ as well as in $\dgr{\Sh}$ 
	is the same (the out-degrees of vertices in $\dgr{\SA}$ might 
	be greater when seen as vertices in $\dgr{\Sh}$). Thus, the difference 
	$h_k(\SB)=h_k(\Sh)-h_k(\SA)$ counts the vertices in 
	$\dgr{\Sh}\sm\dgr{\SA}$ 	with in-degree $k$.
\end{proof} 
In the case where $\Sh$ is a  simplicial polytope, $\SA$ a beginning segment of 
its 
shelling and $\SB$ the set theoretical difference of their faces,  the above  
interpretation helps us compare the $h$-vector of $\SB$ with that of its 
link $\SB/v$ on $v$, for any vertex $v$ in $\SB$.
\begin{lemma}
\label{lem:ap_link_nonlink}
	$h_k(\Sh/v)-h_k(\SA/v)\leq{}h_k(\Sh)-h_k(\SA)$ 
	or equivalently, $h_k(\SB/v)\leq{}h_k(\SB)$.
\end{lemma}
\begin{proof}
To prove our claim, we use the fact that for any vertex $v$ of a polytope  
$\Sh$ there exists a shelling such that the facets that contain $v$, i.e.,\,
the facets in $\str{\Sh}{v}$, appear first in this shelling 
\cite[Corollary 8.13]{z-lp-95}. Applying Lemma \ref{lem:in-degree} 
for $\Sh$ as well as for $\Sh/v$ we have that:
\begin{itemize}
\item $h_k(\SB)$ counts the number of vertices in $\dgr{\Sh}\sm\dgr{\SA}$
	of in-degree $k$, 
\item $h_k(\SB/v)$ counts the number of vertices in $\dgr{\Sh/v}\sm\dgr{\SA/v}$
	of in-degree $k$. 
\end{itemize}
Moreover, since in the above mentioned shellings the link is shelled first, 
the in-degree of a vertex in $\dgr{\Sh}\sm\dgr{\SA}$ can only but 
be greater with respect to its in-degree in $\dgr{\Sh/v}\sm\dgr{\SA/v}$. 
This immediately implies the statement of the lemma. 
\end{proof}

Using the machinery developed above, we may now show that 
$\bx{}\qQ_R$ admits a particular shelling, as stated in the following
lemma.

\newcommand{\qQp}{\bx\qQ'}
\begin{lemma}
  \label{lem:special_shelling}
  There exists a shelling of $\bx{}\qQ_R$
  starting from facets in  
  $\bigcup_{j\in{}R\sm\{1\}}\str{y_{R\sm\{j\}}}{\bx\qQ_R}$,
  and finishing with facets in  $\str{y_{R\sm\{1\}}}{\bx\qQ_R}$.
\end{lemma}
\begin{proof}
  Let us start with some definitions: 
  we denote by $\zZ$ the $(d+|R|-1)$-complex we get by performing the 
  recursion in \eqref{algorithm} until the last but one step, i.e., after 
  having added all the auxiliary vertices $y_S$ with $|S|\leq{}|R|-2.$ 
  Clearly, the facets of $\zZ$ are the $(d+|R|-2)$-polytopes 
  $\qQ_{R\sm\{i\}}$, $i\in{}R$, as well as all facets in $\fF_R.$ Since $\zZ$ 
  is polytopal, each line in general position induces a shelling order of 
  its facets (cf. \cite[Section 8.2]{z-lp-95}). We will chose a line in such 
  a way, so that the induced line shelling of $\zZ$ leads us (after adding 
  all vertices $y_{R\sm\{i\}}$) to the sought-for	shelling of 
  $\bx{}\qQ_R$.

  Notice that, by the definition of the Cayley 
  embedding, there exists a hyperplane in $\mathbb R^{d+|R|-1}$ containing 
  $P_1$ and being parallel to $\cC_{R\sm\{1\}}$ (and thus to 
  $\qQ_{R\sm\{1\}}$). We can therefore choose a line $\ell$ beyond 
  $y_1$ in $\zZ$ and intersecting $\qQ_{R\sm\{1\}}$ in its interior. This 
  line $\ell$ yields a shelling $\Sl{\zZ}$ of $\zZ$ starting from facets in 
  $\str{y_1}{\zZ}$ and finishing with $\qQ_{R\sm\{1\}}$. Since the facets in 
  $\str{y_1}{\zZ}$ are nothing but the polytopes $\qQ_{R\sm\{i\}}$, 
  \;$i\in{}R\sm\{1\}$, the shelling $\Sl{\zZ}$ starts with all 	
  $\qQ_{R\sm\{i\}}$, \;$i\in{}R\sm\{1\}$, (continues with the facets in 
  $\fF_R$) and ends with $\qQ_{R\sm\{1\}}$. Our next goal is to replace each 
  facet $\qQ_{R\sm\{i\}}$ in $\Sl{\zZ}$ by all facets in  
  $\str{y_{R\sm\{i\}}}{\bx{}\qQ_{R\sm\{i\}}}$, ordered so that the 
  conditions in Definition \ref{def:shell1} are satisfied.
  
  We do this by induction. If $\qQ_{R\sm\{2\}}$ is the first facet in the 
  shelling order $\Sh(\zZ)$ then we can replace it by the facets  in 
  $\str{y_{R\sm\{2\}}}{\bx\qQ_{R\sm\{2\}}}$, in any order ``inherited'' from 
  a shelling of $\bx\qQ_{R\sm\{2\}}$. Without loss of generality,  we 
  assume that  the facets $\qQ_{R\sm\{j\}}$ with 
  $2\leq{}j<i$ are those preceding $\qQ_{R\sm\{i\}}$ in the shelling order 
  $\Sh(\zZ)$. By our induction hypothesis, we have replaced all 
  $\qQ_{R\sm\{j\}}$ by $\str{y_{R\sm\{j\}}}{\bx\qQ_{R\sm\{j\}}}$ in a way 
  that the conditions of our claim are satisfied; we
  want to prove the same for $j=i$.

  Indeed,  notice that the intersection of  $\qQ_{R\sm\{i\}}$ with the 
  union  of the previous facets, is the union of all $\qQ_{R\sm\{i,j\}}$
  with $2\leq{}j<i$, whether we consider ``previous'' in the shelling 
  $\Sh(\zZ)$ or in the shelling until the current inductive step 
  (i.e., when each $\qQ_{R\sm\{j\}}$ with $2\leq{}j<i$ is stellarly 
  subdivided).
  As a result, the second condition of Definition \ref{def:shell1},
  namely that that there exists a shelling order of the facets of
  $\bx\qQ_{R\sm\{i\}}$ starting with all facets of
  $\bigcup_{2\leq{}j<i}\bx\qQ_{R\sm\{i,j\}}$, holds.
  It suffices to choose a shelling order  of  $\bx{}\qQ_{R\sm\{j\}}$ 
  that respects the common shelling order with 
  $\bigcup_{2\leq{}j<i}\bx\qQ_{R\sm\{i,j\}}$.
  Using this shelling order, we may replace the facet
  $\qQ_{R\sm\{i\}}$ by those in
  $\str{y_{R\sm\{i\}}}{\bx\qQ_{R\sm\{i\}}}$ (the shelling orders of
  each $\str{y_{S}}{\bx\qQ_{S}}$ are inherited from those for 		
  $\bx\qQ_{S}$) and arrive at a shelling order of $\qQ_R$ with the 
  desired properties. The last facet $\qQ_{R\sm\{1\}}$ can be replaced 
  by $\str{y_{R\sm\{1\}}}{\bx\qQ_{R\sm\{1\}}}$ without any further 
  concern, since the shelling conditions are already fulfilled  from the 
  shelling $\Sh(\zZ)$.
\end{proof}


%
Exploiting Lemmas \ref{lem:in-degree}, \ref{lem:ap_link_nonlink} and
\ref{lem:special_shelling} we arrive at the following theorem, where
we bound the right-hand side of \eqref{recur-relation-FR-wrt-K} by an
expression that does not involve the links $\kK_S/v$.

\begin{theorem}
\label{thm:link}
For all $v\in{}V_R$ and all $k\ge{}0$ we have:
  \begin{equation}
    \sum_{\emptyset\subset{}S\subseteq{}R}(-1)^{|R|-|S|}
    \sum_{v\in{}V_S}g_k^{(|R|-|S|)}(\kK_S/v)\le
    \sum_{\emptyset\subset{}S\subseteq{}R}(-1)^{|R|-|S|}
    \sum_{v\in{}V_S}g_k^{(|R|-|S|)}(\kK_S),
    \label{equ:link-non link for F}
  \end{equation}
  where $V_S=\cup_{i\in{}S}V_i$.
\end{theorem}
\begin{proof} 
Let us first observe that, by rearranging terms, we can rewrite relation
\eqref{equ:link-non link for F}  as:
\begin{align}
  \sum_{i\in{}R}\sum_{v\in{}V_i}\sum_{\{i\}\subseteq{}S\subseteq{}R}
  (-1)^{|R|-|S|}\,g_k^{(|R|-|S|)}(\kK_S/v)
  \le\sum_{i\in{}R}\sum_{v\in{}V_i}
  \sum_{\{i\}\subseteq{}S\subseteq{}R}(-1)^{|R|-|S|}\,g_k^{(|R|-|S|)}(\kK_S).
  \label{non:1}
\end{align}
Clearly, to show that relation \eqref{non:1} holds, it
suffices to prove that:
\begin{align}
  \sum_{\{i\}\subseteq{}S\subseteq{}R}(-1)^{|R|-|S|}\,g_k^{(|R|-|S|)}(\kK_S/v)
  \le\sum_{\{i\}\subseteq{}S\subseteq{}R}(-1)^{|R|-|S|}\,
  g_k^{(|R|-|S|)}(\kK_S),
  \label{non:2}
\end{align}
for any arbitrary fixed $i\in{}R.$

Without loss of generality we may assume that $i=1.$ 
Define $\gG_1=\fF_R\cup\fF_{R\sm\{1\}}$. Since 
$\fF_R$ and $\fF_{R\sm\{1\}}$ are disjoint, we can write:
\begin{align} 
  f_k(\gG_1)&=f_k(\fF_R)+f_k(\fF_{R\sm\{1\}})\notag\\
  &=\sum_{S\subseteq{}R}(-1)^{|R|-|S|}f_k(\kK_S)+
  \sum_{S\subseteq{}R\sm\{1\}}(-1)^{|R|-1-|S|}f_k(\kK_S)\notag\\
  &=\sum_{S\subseteq{}R}(-1)^{|R|-|S|}f_k(\kK_S)
  -\sum_{S\subseteq{}R\sm\{1\}}(-1)^{|R|-|S|}f_k(\kK_S)\notag\\
  &=\sum_{\{1\}\subseteq{}S\subseteq{}R}(-1)^{|R|-|S|}f_k(\kK_S).
  \label{non:3}
\end{align} 
Similarly, for all $v\in{}V_1$: 
\begin{equation}
  f_k(\gG_1/v)=\sum_{\{1\}\subseteq{}S\subseteq{}R}(-1)^{|R|-|S|}f_k(\kK_S/v).
  \label{non:4}
\end{equation}

Converting the above relations into $h$-vector relations (using generating 
functions and comparing coefficients) we deduce that:
\begin{equation} 
  h_k(\gG_1)=
  \sum_{\{1\}\subseteq{}S\subseteq{}R}(-1)^{|R|-|S|}g^{(|R|-|S|)}_k(\kK_S),
  \label{non:5}
\end{equation}
and	
\begin{equation}	
  h_k(\gG_1/v)=
  \sum_{\{1\}\subseteq{}S\subseteq{}R }(-1)^{|R|-|S|}g^{(|R|-|S|)}_k(\kK_S/v).
  \label{non:6}
\end{equation}

Thus, in view of \eqref{non:5} and \eqref{non:6}, proving \eqref{non:2} 
reduces to showing that $h_k(\gG_1/v)\leq{}h_k(\gG_1)$.
Define $\bx{}\qQ'_R$ to be the polytopal $(d+|R|-1)$-complex whose
facets are the facets of $\bx{}\qQ_R$ not incident to $y_{R\sm\{1\}}.$
To understand the face structure of $\bx{}\qQ'_R$, we use Lemma 
\ref{lem:QR-FSjoin} to rewrite $\bx{}\qQ_R$ as the union:
\[\fF_R\bigcup\limits_{i\in{}R} 
\str{y_{R\sm\{i\}}}{\bx{}\qQ_R}\]
of, not necessarily disjoint, faces. After removing all faces of 
$\bx{}\qQ_R$ incident to $y_{R\sm\{1\}}$ we are left with the 
following set of faces:
\begin{align*} 
  \bx{}\qQ'_R=
  \overbrace{
    \bigcup_{i\in{}R\sm\{1\}}
    \str{y_{R\sm\{i\}}}{\bx{}\qQ_R}
  }^{\SA}
  \cup\overbrace{\fF_R\cup{}\fF_{R\sm\{1\}}}^{\SB}.
\end{align*}
Although the face sets in the above union are not disjoint, the face sets 
$\SA$ and $\SB$ are. This further implies that the facets of $\bx{}\qQ'_R$ 
are the facets in $\SA$ and those in  $\SB$. We next claim that 
$\bx{}\qQ'_R$ is shellable and that there exists a shelling of 
$\bx{}\qQ'_R$ in which all facets in $\SA$ come first.

Indeed, according to Lemma \ref{lem:special_shelling}, there
exists a shelling 
of $\bx{}\qQ_R$ starting from facets in $\bigcup_{i\in{}R\sm\{1\}} 
\str{y_{R\sm\{i\}}}{\bx{}\qQ_R}$, continuing with 
those in $\fF_R$ and ending with facets in $\str{y_{R\sm\{1\}}}{\bx\qQ_R}$. 
Discarding the facets in $\str{y_{R\sm\{1\}}}{\bx\qQ_R}$ we obtain a 	
shelling of $\bx{}\qQ'_R$  starting from facets in 
$\bigcup_{i\in{}R\sm\{1\}}
\str{y_{R\sm\!\{i\}}}{\bx{}\qQ_R}$ and ending with facets 
in $\fF_R$. We then apply Lemma \ref{lem:ap_link_nonlink} with 
$\Sh:=\bx{}\qQ'_R$
and $\SA:=\bigcup_{i\in{}R\sm\{1\}} 
\str{y_{R\sm\!\{i\}}}{\bx{}\qQ_R}$
and we deduce that $h_k(\SB/v)\leq{}h_k(\SB),$ or equivalently that  
$h_k(\gG_1/v)\leq{}h_k(\gG_1)$. This completes our proof. 
\end{proof} 

\subsection{The last step towards the recurrence relation}
\label{sec:rec-last-step}

The last step for proving  Theorem \ref{thm:hkFRrec}, is the following 
lemma that  involves calculations which simplify the right-hand side of 
\eqref{equ:link-non link for F}.
\begin{lemma}
  \label{lem:simplify}
  Let $\emptyset\subset{}R\subseteq[r]$, and $V_S=\cup_{i\in{}S}V_i$,
  for all $\emptyset{}\subset{}S\subseteq{}R$. Then, for all $k\ge{}0$
  we have:
  \begin{equation}
    \sum_{\emptyset\subset{}S\subseteq{}R}(-1)^{|R|-|S|}\sum_{v\in{}V_S}g_k^{(|R|-|S|)}(\kK_S)
    =n_{R}\,h_k(\fF_R)+\sum_{i\in{}R}n_i\,g_k(\fF_{R\sm\{i\}}),
    \label{equ:simplify}
  \end{equation}
  where $n_R=\sum_{i\in{}R}n_i$ and $n_\emptyset=0$.
\end{lemma}
\begin{proof}
From relation \eqref{hkKF} and the definition of the $m$-order
$g$-vector (cf.~\eqref{equ:gm-def}), we can easily show that, for any
$\emptyset{}\subset{}R\subseteq[r]$,
\begin{equation*}
  g^{(m)}_k(\kK_R)=\sum_{\emptyset\subset{}S\subseteq{}R}g_{k}^{(|R|-|S|+m)}(\fF_S).
\end{equation*}
Hence, for all $0\le{}k\le{}d+|R|-1$, we get:
\begin{equation*}
  g_k^{(|R|-|S|)}(\kK_S)
  =\sum_{\emptyset\subset{}X\subseteq{}S}g_k^{(|S|-|X|+(|R|-|S|))}(\fF_X)
  =\sum_{\emptyset\subset{}X\subseteq{}S}g_k^{(|R|-|X|)}(\fF_X).
\end{equation*}
Thus, the left-hand side of \eqref{equ:simplify} becomes:
\begin{align} 
  &\sum_{\emptyset\subset{}S\subseteq{}R}(-1)^{|R|-|S|}\sum_{v\in{}V_S}
  g_k^{(|R|-|S|)}(\kK_S)
  =\sum_{\emptyset\subset{}S\subseteq{}R}(-1)^{|R|-|S|}n_S\,
  g_k^{(|R|-|S|)}(\kK_S)
  \notag\\
  &=\sum_{\emptyset\subset{}S\subseteq{}R}(-1)^{|R|-|S|}n_S
  \sum_{\emptyset\subset{}X\subseteq{}S}
  g_k^{(|R|-|X|)}(\fF_X)
  =\sum_{\emptyset\subset{}X\subseteq{}R}
  \bigl(\sum_{X\subseteq{}S\subseteq{}R}
  (-1)^{|R|-|S|}n_S\bigr)g_k^{(|R|-|X|)}(\fF_X).
  \label{simpl:1}
\end{align} 
We next evaluate the coefficient of $g_k^{(|R|-|X|)}(\fF_X)$ in 
\eqref{simpl:1}, i.e., the quantity
\begin{align}
  \sum_{X\subseteq{}S\subseteq{}R}(-1)^{|R|-|S|}n_S. 
  \label{simpl:2}	
\end{align} 
We separate cases:
\begin{enumerate}[\fs (a)] 
\item If $X=R$ the sum in \eqref{simpl:2}
  simplifies to $n_{R}$.
\item If $|X|=|R|-1$, then $X=R\sm\{i\}$ for some $i\in{}R$
and the sum in \eqref{simpl:2} simplifies to
\begin{equation*}
  \sum_{X\subseteq{}S\subseteq{}R}(-1)^{|R|-|S|}n_S
  =(-1)^{|R|-(|R|-1)}n_{R\sm\{i\}}+(-1)^{|R|-|R|}n_R
  =n_R-n_{R\sm\{i\}}=n_i.
\end{equation*}
\item If $|X|<|R|-1$  then for every $i\in{}R\sm{}X$ and every
  $0\leq{}j\leq|R|-|X|-1$ there exist $\binom{|R|-|X|-1}{j}$ 
  sets of size $|X|+j+1$ containing $i.$ We therefore have: 
  \begin{align*} 
    \sum_{X\subseteq{}S\subseteq{}R}(-1)^{|R|-|S|}n_S=
    \sum_{i\in{}R\sm{}X}\sum_{j=0}^{|R|-|X|-1}
    (-1)^{|R|-|X|-j-1}\tbinom{|R|-|X|-1}{j}n_i
    =\sum_{i\in{}R\sm{}X}0^{|R|-|X|-1}n_i=0.
  \end{align*} 
\end{enumerate}
From {\fs(a)}-{\fs(c)} we deduce  that the only non-zero coefficients of 
$g_k^{(|R|-|X|)}(\fF_X)$ in \eqref{simpl:1} are those for which
$|X|=|R|$ or $|R|-1.$ Thus, the sum in \eqref{simpl:1} simplifies to
\begin{equation*} 
	n_{R}h_k(\fF_R)+\sum_{i\in{}R}n_i\,g_k(\fF_{R\sm\{i\}}),
\end{equation*}
which is precisely the right-hand side of \eqref{equ:simplify}.
\end{proof}


\subsection{The proof of Theorem \ref{thm:hkFRrec}}
\label{sec:proof-hkFRrec}
\begin{proof}[Proof of Theorem \ref{thm:hkFRrec}]
To prove the inequality in the statement of the theorem, 
we generalize McMullen's steps in the proof of his Upper Bound 
theorem \cite{m-mnfcp-70}.

Our starting point is relation \eqref{equ:links} applied to the
simplicial $(d+|R|-1)$-polytope $\qQ_R$, expressed in terms of generating
functions:
\begin{equation}\label{equ:main:McM-QR}
  (d+|R|-1)\,\h{\bx{\qQ_R}}+(1-t)\hh{\bx{\qQ_R}}
  =\sum_{v\in\text{vert}(\bx{\qQ_R})}\h{\bx{\qQ_R}/v}.
\end{equation}
Exploiting the combinatorial structure of $\qQ_R$ in order to express:
(1) $\mb{h}(\bx\qQ_R)$ 
in terms of $\mb{h}(\fF_S)$, $\emptyset\subset{}S\subset{}R$, and
(2) $\mb{h}(\bx\qQ_R/v)$ in terms of $\mb{h}(\fF_R/v)$ and $\mb{h}(\fF_S)$,
$\emptyset\subset{}S\subset{}R$,
relation \eqref{equ:main:McM-QR} yields 
(see Sections \ref{sec:hFR-hFRv}--\ref{sec:links-nonlinks}):
\begin{equation*}
  (d+|R|-1)\htt{\fF_R}{t}+(1-t)\hh{\fF_R}=\sum_{v\in{}V_R}\htt{\fF_R/v}{t},
\end{equation*}
the element-wise form of which is:
\begin{equation*}
  (k+1)h_{k+1}(\fF_R)+(d+|R|-1-k)h_k(\fF_R)=\sum_{v\in{}V_R} h_k(\fF_R/v),
  \quad 0\le{}k\le{}d+|R|-2.
\end{equation*}
Noticing that $h_k(\fF_R/v)$ is equal to
$\sum_{\emptyset\subset{}S\subseteq{}R}(-1)^{|R|-|S|}\sum_{v\in{}V_S}g_k^{(|R|-|S|)}(\kK_S/v)$
(by the Inclusion-Exclusion Principle), we have that
(see Section \ref{sec:shellings}):
\begin{equation*}
  \sum_{\emptyset\subset{}S\subseteq{}R}(-1)^{|R|-|S|}
  \sum_{v\in{}V_S}g_k^{(|R|-|S|)}(\kK_S/v)
  \le
  \sum_{\emptyset\subset{}S\subseteq{}R}(-1)^{|R|-|S|}
  \sum_{v\in{}V_S}g_k^{(|R|-|S|)}(\kK_S).
\end{equation*}
The right-hand side of the above relation simplifies to
$n_{R}\,h_k(\fF_R)+\sum_{i\in{}R}n_i\,g_k(\fF_{R\sm\{i\}})$
(cf.~Section \ref{sec:rec-last-step}), which in
turn suggests the following inequality:
\begin{equation}
  (k+1)h_{k+1}(\fF_R)+(d+|R|-1-k)h_k(\fF_R)
  \le
  n_{R}\,h_k(\fF_R)+\sum_{i\in{}R}n_i\,g_k(\fF_{R\sm\{i\}})
\end{equation}
that holds true for all $0\le{}k\le{}d+|R|-2$. Solving in terms of
$h_{k+1}(\fF_R)$ results in \eqref{hkFRrec}.
\end{proof}

%% file: ubFK.tex
\section{Upper bounds}
\label{sec:ub}

Let $S_1,\ldots,S_r$ be a partition of a set $S$ into $r$
sets. We say that $A\subseteq{}\bigcup\limits_{1\leq{}i\leq{}r}S_i$
is a \emph{spanning subset of $S$} if 
$A\cap{}S_i\ne\emptyset$ for all $1\leq{}i\leq{}r$.

\begin{definition}
  Let $P_i,i\in{}R$, be $d$-polytopes with vertex sets $V_i,i\in{}R.$
  We say that their Cayley polytope $\cC_R$ is \emph{$R$-neighborly} 
  if every spanning subset of $\bigcup_{i\in{}R}V_i$ of size 
  $|R|\leq{}\ell\leq\lexp{d+|R|-1}$ is a face of $\cC_R$ (or,
  equivalently, a face of $\fF_R$). We say that the Cayley polytope 
  $\cC_R$ is \emph{Minkowski-neighborly} if, for every  
  $\emptyset\subset{}S\subseteq{}R$, the Cayley polytope $\cC_S$ is 
  $S$-neighborly.
\label{def:multi_nbh}
\end{definition}
The following characterizes $R$-neighborly Cayley polytopes in terms
of the $f$- and $h$-vector of $\fF_R$. 
\begin{lemma} 
\label{lem:eq_claim} 
The following are equivalent: 
\begin{enumerate}[(i)]
\item $\cC_R$ is $R$-neighborly, 
\item $f_{\ell-1}(\fF_R)=\sum_{\emptyset\subset{}S\subseteq{}R}(-1)^{|R|-|S|}
  \tbinom{n_S}{\ell}$,\; for all $0\leq{}\ell\leq{}\lexp{d+|R|-1}$,
\item $h_{\ell}(\fF_R)=\sum_{\emptyset\subset{}S\subseteq{}R}
  (-1)^{|R|-|S|}\tbinom{n_S-d-|R|+\ell}{\ell}$,\; for all 
  $0\leq{}\ell\leq{}\lexp{d+|R|-1}$,
\end{enumerate}
where $n_i$ is the number of vertices of $P_i$ and $n_S=\sum_{i\in{}S}n_i$. 
\end{lemma}

\begin{proof}
To show the equivalence between {\rm (i)} and {\rm (ii)}, notice, from  
the definition of spanning subsets, that every spanning subset of  
$V_R=\bigcup_{i\in{}R}V_i$ of size $\ell\geq{}|R|$ has:
\begin{equation*}
	\sum_{\scalebox{0.6}{$\begin{array}{c}
	\sum_{i\in{}R}k_i=\ell\\
	1\leq{}k_i\leq{}n_i
	\end{array} $}}
	\prod_{i\in{}R}\tbinom{n_i}{k_i}
\end{equation*}
elements.  Using induction on the size of $R$, one can check that the above sum 
of products is equal to the expression on the right-hand side of {\rm(ii)}.
Moreover, in the case where $\ell<|R|$, the expression on the right-hand side 
of {\rm(ii)} is $0$. This, agrees with the fact that  there do not exist any 
spanning subsets of $\bigcup\limits_{i\in{}R}V_i$ of size $\ell<|R|$.

We next show the equivalence between {\rm(ii)} and {\rm(iii)}.
Taking the $(d-k)$-th derivative of relation \eqref{hf} for 
$\fF_R$, it suffices to show that the values for  $f_{\ell-1}(\fF_R)$ and 
$h_{\ell}(\fF_R),$ $0\leq\ell\leq{}k$, in the statement of the theorem satisfy 
	\begin{equation} 
		\sum_{i=0}^k{}f_{i-1}(\fF_R)\frac{(d-i)!}{(k-i)!}t^{k-i}=
		\sum_{i=0}^k\frac{(d-i)!}{(k-i)!}h_i(\fF_R)(t+1)^{k-i}.
		\label{equ:trunc}
	\end{equation}
 
Indeed, we have: 
\begin{align} 
  \sum_{i=0}^k h_{i}(\fF_R)&\tfrac{(d+|R|-1-i)!}{(k-i)!}(t+1)^{k-i}\notag\\
  &=\sum_{i=0}^k\sum_{\emptyset\subset{}S\subseteq{}R}(-1)^{|R|-|S|} 
  \tbinom{n_S-d-|R|+i}{i}\tfrac{(d+|R|-1-i)!}{(k-i)!}(t+1)^{k-i}\notag\\
  &=\sum_{\emptyset\subset{}S\subseteq{}R}(-1)^{|R|-|S|}\sum_{i=0}^k
  \tbinom{n_S-d-|R|+i}{i}\tfrac{(d+|R|-1-i)!}{(k-i)!}
  \sum_{j=0}^{k-i}\tbinom{k-i}{j}t^{j}\notag\\
  &=\sum_{\emptyset\subset{}S\subseteq{}R}(-1)^{|R|-|S|}
  \sum_{i=0}^k\sum_{j=0}^{k}\tfrac{(d+|R|-1-k+j)!}{j!}
  \tbinom{n_S-d-|R|+i}{i}t\tbinom{d+|R'|-1-i}{d+|R|-1-k+j}t^{j}\notag\\
  &=\sum_{\emptyset\subset{}S\subseteq{}R}(-1)^{|R|-|S|}\sum_{j=0}^{k} 
  \tfrac{(d+|R|-1-k+j)!}{j!}\sum_{i=0}^k\tbinom{n_S-d-|R|+i}{n_S-d-|R|}
  \tbinom{d+|R|-1-i}{k-i-j}t^{j}\label{sth}\\
  &=\sum_{\emptyset\subset{}S\subseteq{}R}(-1)^{|R|-|S|}
  \sum_{j=0}^{k}\tfrac{(d+|R|-1-k+j)!}{j!}\tbinom{n_S}{n_S-k+j}t^{j}
  \label{sth2}\\
  &=\sum_{\emptyset\subset{}S\subseteq{}R}(-1)^{|R|-|S|}\sum_{j=0}^{k} 
  \tfrac{(d+|R|-1-j)!}{(k-j)!}\tbinom{n_S}{j}t^{k-j}\notag\\
  &=\sum_{j=0}^{k}\sum_{\emptyset\subset{}S\subseteq{}R}(-1)^{|R|-|S|}
  \tbinom{n_S}{j}\tfrac{(d+|R|-1-j)!}{(k-j)!} t^{k-j},\notag
\end{align} 
where to go from \eqref{sth} to \eqref{sth2}   we used 
Relation 5.26 from \cite{gkp-cm-89}: 
\[\sum\limits_{0\leq{}k\leq{}l}\tbinom{l-k}{m}\tbinom{q+k}{n}=
\tbinom{l+q+1}{m+n+1},\]
holding for all non negative integers $l,m,n\geq{}q$.
\end{proof}

\subsection{\texorpdfstring{Upper bounds for the lower half of $\mb{h}(\fF_R)$}%
{Upper bounds for the lower half of h(FR)}}

From the recurrence relation in Theorem \ref{thm:hkFRrec}
we arrive at the following theorem.
\begin{theorem}
  \label{theor:ub_from_rec}
  For any $\emptyset\subset{}R\subseteq{}[r]$ and $0\le{}k\le{}d+|R|-1$, 
  we have:
  \begin{align}
    g_k(\fF_R)&\leq\sum_{\emptyset\subset{}S\subseteq{}R}(-1)^{|R|-|S|}
    \tbinom{n_S-d-|R|-1+k}{k}, \quad\text{and}
    \label{equ:ub_from_rec1}
    \\
    h_k(\fF_R)&\le\sum_{\emptyset\subset{}S\subseteq{}R}(-1)^{|R|-|S|}
    \tbinom{n_S-d-|R|+k}{k},
    \label{equ:ub_from_rec2}
  \end{align}
  where $n_S=\sum_{i\in{}S}n_i$. 
  Equalities hold for all $0\le{}k\le{}\lexp{d+|R|-1}$ if and only if
  the Cayley polytope $\cC_R$ is $R$-neighborly.
\label{theor:UB_hF}
\end{theorem}
\begin{proof}
We are going to show the wanted bounds by induction on $|R|$ and $k$.
Clearly the bounds hold for $|R|=1$ and for any $0\le{}k\le{}d$ (this
is the case of one $d$-polytope and the bounds of the lemma refer to the
well-known bounds on the elements of the $h$- and $g$-vector of a polytope).

Suppose now that the bounds for $g_k(\fF_R)$ and $h_k(\fF_R)$ hold for
all $|R|<m$ and for all $0\le{}k\le{}d+|R|-1$. Consider an $R$ with
$|R|=m$. Then, for $k=0$ we have:
\begin{align*}
  h_0(\fF_R)&=f_{-1}(\fF_R)=(-1)^{|R|-1}=-(-1)^{|R|}
  =\sum_{i=1}^{|R|}(-1)^{|R|-i}\tbinom{|R|}{i}
  =\sum_{i=1}^{|R|}(-1)^{|R|-i}
  \sum_{\substack{\emptyset\subset{}S\subseteq{}R\\|S|=i}}1\\
  &=\sum_{i=1}^{|R|}
  \sum_{\substack{\emptyset\subset{}S\subseteq{}R\\|S|=i}}(-1)^{|R|-i}
  =\sum_{\emptyset\subset{}S\subseteq{}R}(-1)^{|R|-|S|}
  =\sum_{\emptyset\subset{}S\subseteq{}R}(-1)^{|R|-|S|}\tbinom{n_S-d-|R|}{0},
\end{align*}
and
\begin{equation*}
  g_0(\fF_R)=h_0(\fF_R)-h_{-1}(\fF_R)
  =\sum_{\emptyset\subset{}S\subseteq{}R}(-1)^{|R|-|S|}\tbinom{n_S-d-|R|}{0}
  =\sum_{\emptyset\subset{}S\subseteq{}R}(-1)^{|R|-|S|}\tbinom{n_S-d-|R|-1}{0}.
\end{equation*}
For $k\ge{}1$ we have:
\begin{align}
   g_k(\fF_R)&=h_k(\fF_R)-h_{k-1}(\fF_R)\notag\\
   &\le\tfrac{n_{R}-d-|R|+k}{k}h_{k-1}(\fF_R)+\sum_{i\in{}R}\tfrac{n_i}{k}
   g_{k-1}(\fF_{R\sm\{i\}})-h_{k-1}(\fF_R)\notag\\
   &=\tfrac{n_{R}-d-|R|}{k}h_{k-1}(\fF_R)
   +\sum_{i\in{}R}\tfrac{n_i}{k}g_{k-1}(\fF_{R\sm\{i\}}).
   \label{ind:1}
\end{align}
By our inductive hypotheses, we have:
\begin{equation}
  h_{k-1}(\fF_R)\le\sum_{\emptyset\subset{}S\subseteq{}R}(-1)^{|R|-|S|}
  \tbinom{n_S-d-|R|+k-1}{k-1}
  \label{ind:2}
\end{equation}
and also, for all $i\in{}R$:
\begin{equation}
  \begin{aligned}
    g_{k-1}(\fF_{R\sm\{i\}})
    &\le\sum_{\emptyset\subset{}S\subseteq{}R\sm\{i\}}(-1)^{|R\sm\{i\}|-|S|}
    \tbinom{n_S-d-|R\sm\{i\}|-1+k-1}{k-1}\\
    &=-\sum_{\emptyset\subset{}S\subseteq{}R\sm\{i\}}(-1)^{|R|-|S|}
    \tbinom{n_S-d-|R|-1+k}{k-1}.
  \end{aligned}\label{ind:3}
\end{equation}
Substituting \eqref{ind:2} and \eqref{ind:3} in \eqref{ind:1} we get:
\begin{equation}
  g_k(\fF_R)\le\tfrac{n_{R}-d-|R|}{k}
  \sum_{\emptyset\subset{}S\subseteq{}R}(-1)^{|R|-|S|}
  \tbinom{n_S-d-|R|+k-1}{k-1}
  -\sum_{i\in{}R}\tfrac{n_i}{k}
  \sum_{\emptyset\subset{}S\subseteq{}R\sm\{i\}}(-1)^{|R|-|S|}
  \tbinom{n_S-d-|R|-1+k}{k-1}.\label{ind:4}
\end{equation}
Consider the sum 
$\sum_{i\in{}R}\tfrac{n_i}{k}\sum_{\emptyset\subset{}S\subseteq{}R\sm\{i\}}
(-1)^{|R|-|S|}\tbinom{n_S-d-|R|-1+k}{k-1}$; observe that for any given
$\emptyset\subset{}S\subset{}R$ we get a contribution of $\frac{n_i}{k}$ for 
$\tbinom{n_S-d-|R|-1+k}{k-1}$, for any
$i\not\in{}S$. In other words, we have the equality:
\begin{equation}
  \sum_{i\in{}R}\tfrac{n_i}{k}
  \sum_{\emptyset\subset{}S\subseteq{}R\sm\{i\}}(-1)^{|R|-|S|}
  \tbinom{n_S-d-|R|-1+k}{k-1}
  =\sum_{\emptyset\subset{}S\subset{}R}(-1)^{|R|-|S|}\tfrac{n_{R\sm{}S}}{k}
  \tbinom{n_S-d-|R|-1+k}{k-1}. \label{ind:5}
\end{equation}
In view of \eqref{ind:5} the inequality in \eqref{ind:4} becomes:
{\allowdisplaybreaks
\begin{align*}
  g_k(\fF_R)&\le\tfrac{n_{R}-d-|R|}{k}
  \sum_{\emptyset\subset{}S\subseteq{}R}(-1)^{|R|-|S|}
  \tbinom{n_S-d-|R|+k-1}{k-1}
  -\sum_{\emptyset\subset{}S\subset{}R}(-1)^{|R|-|S|}
  \tfrac{n_{R\sm{}S}}{k}\tbinom{n_S-d-|R|-1+k}{k-1}\\
  &=\tfrac{n_R-d-|R|}{k}\tbinom{n_R-d-|R|-1+k}{k-1}
  +\tfrac{n_R-d-|R|}{k}\sum_{\emptyset\subset{}S\subset{}R}
  (-1)^{|R|-|S|}\tbinom{n_S-d-|R|-1+k}{k-1}\\
  &\qquad
  -\sum_{\emptyset\subset{}S\subset{}R}(-1)^{|R|-|S|}
  \tfrac{n_{R\sm{}S}}{k}\tbinom{n_S-d-|R|-1+k}{k-1}\\
  &=\tfrac{n_R-d-|R|}{k}\tbinom{n_R-d-|R|-1+k}{k-1}
  +\sum_{\emptyset\subset{}S\subset{}R}(-1)^{|R|-|S|}
  \left(\tfrac{n_R-d-|R|}{k}-\tfrac{n_{R\sm{}S}}{k}\right)
  \tbinom{n_S-d-|R|-1+k}{k-1}\\
  &=\tfrac{n_R-d-|R|}{k}\tbinom{n_R-d-|R|-1+k}{k-1}
  +\sum_{\emptyset\subset{}S\subset{}R}(-1)^{|R|-|S|}
  \tfrac{n_S-d-|R|}{k}\tbinom{n_S-d-|R|-1+k}{k-1}\\
  &=\tfrac{n_R-d-|R|+k}{k}\tbinom{n_R-d-|R|-1+k}{k-1}
  -\tbinom{n_R-d-|R|-1+k}{k-1}\\
  &\qquad
  +\sum_{\emptyset\subset{}S\subset{}R}(-1)^{|R|-|S|}
  \biggl(\tfrac{n_S-d-|R|+k}{k}\tbinom{n_S-d-|R|-1+k}{k-1}
    -\tbinom{n_S-d-|R|-1+k}{k-1}\biggr)\\
  &=\tbinom{n_R-d-|R|+k}{k}-\tbinom{n_R-d-|R|-1+k}{k-1}
  +\sum_{\emptyset\subset{}S\subset{}R}(-1)^{|R|-|S|}
  \biggl(\tbinom{n_S-d-|R|+k}{k}-\tbinom{n_S-d-|R|-1+k}{k-1}\biggr)\\
  &=\tbinom{n_R-d-|R|+k-1}{k}+\sum_{\emptyset\subset{}S\subset{}R}
  (-1)^{|R|-|S|}\tbinom{n_S-d-|R|+k-1}{k}\\
  &=\sum_{\emptyset\subset{}S\subseteq{}R}(-1)^{|R|-|S|}
  \tbinom{n_S-d-|R|-1+k}{k}.
\end{align*}}

We can now turn our attention to proving the bound for $h_k(\fF_R)$.
Using the recursive relation \eqref{hkFRrec} and the upper bound for 
$g_k(\fF_R)$ that we just proved, we get:
{\allowdisplaybreaks
\begin{align*}
	h_k(\fF_R)&\le\tfrac{n_{R}-d-|R|+k}{k}h_{k-1}(\fF_R)
	+\sum_{i\in{}R}\tfrac{n_i}{k}g_{k-1}(\fF_{R\sm\{i\}})\\
	&\le\tfrac{n_{R}-d-|R|+k}{k}
	\sum_{\emptyset\subset{}S\subseteq{}R}(-1)^{|R|-|S|}
	\tbinom{n_S-d-|R|+k-1}{k-1}\\
	&\qquad
	+\sum_{i\in{}R}\tfrac{n_i}{k}
	\sum_{\emptyset\subset{}S\subseteq{}R\sm\{i\}}(-1)^{|R\sm\{i\}|-|S|}
	\tbinom{n_S-d-|R\sm\{i\}|-1+k-1}{k-1}\\
	&=\tfrac{n_{R}-d-|R|+k}{k}\tbinom{n_R-d-|R|+k-1}{k-1}
	+\tfrac{n_{R}-d-|R|+k}{k}
	\sum_{\emptyset\subset{}S\subset{}R}(-1)^{|R|-|S|}
	\tbinom{n_S-d-|R|+k-1}{k-1}\\
	&\qquad+\sum_{i\in{}R}\tfrac{n_i}{k}
	\sum_{\emptyset\subset{}S\subseteq{}R\sm\{i\}}(-1)^{|R|-1-|S|}
	\tbinom{n_S-d-|R|+k-1}{k-1}\\
	&=\tbinom{n_R-d-|R|+k}{k}+\sum_{\emptyset\subset{}S\subset{}R}
	(-1)^{|R|-|S|}\tfrac{n_{R}-d-|R|+k}{k}\tbinom{n_S-d-|R|+k-1}{k-1}\\
	&\qquad+\sum_{S\subset{}R}(-1)^{|R|-|S|-1}
	\tfrac{n_{R\sm{}S}}{k}\tbinom{n_S-d-|R|+k-1}{k-1}\\
	&=\tbinom{n_R-d-|R|+k}{k}+\sum_{\emptyset\subset{}S\subset{}R}
	(-1)^{|R|-|S|}\left(\tfrac{n_{R}-d-|R|+k}{k}
	-\tfrac{n_{R\sm{}S}}{k}\right)\tbinom{n_S-d-|R|+k-1}{k-1}\\
	&=\tbinom{n_R-d-|R|+k}{k}
	+\sum_{\emptyset\subset{}S\subset{}R}(-1)^{|R|-|S|}
	\tfrac{n_S-d-|R|+k}{k}\tbinom{n_S-d-|R|+k-1}{k-1}\\
	&=\tbinom{n_R-d-|R|+k}{k}+\sum_{\emptyset\subset{}S\subset{}R}
	(-1)^{|R|-|S|}\tbinom{n_S-d-|R|+k}{k}\\
	&=\sum_{\emptyset\subset{}S\subseteq{}R}(-1)^{|R|-|S|}
	\tbinom{n_S-d-|R|+k}{k}.\qedhere
\end{align*}}

Finally, the equality claim is immediate from Lemma \ref{lem:eq_claim} .
\end{proof}

\subsection{\texorpdfstring{Upper bounds for $h_k(\fF_R)$ and
    $h_k(\kK_R)$ for all $k$}%
{Upper bounds for hk(FR) and hk(KR) for all k}}
\label{app:sec:ub-hKR}

Before proceeding with proving upper bounds for the $h$-vectors of
$\fF_R$ and $\kK_R$ we need to define the following functions.
\begin{definition}
  Let $d\ge{}2$, $\emptyset\subset{}R\subseteq[r]$, $m\ge{}0$,
  $0\le{}k\le{}d+|R|-1$, and $n_i\in\naturals$, $i\in{}R$, with
  $n_i\ge{}d+1$.
  We define the functions $\tubh[m]{k}{R}$ and $\tubk{k}{R}$ by the
  following conditions:
  \begin{enumerate}[1.]
  \item
    $\tubh{k}{R}
    =\sum_{\emptyset\subset{}S\subseteq{}R}(-1)^{|R|-|S|}\tbinom{n_S-d-|R|+k}{k}$,
    $0\le{}k\le\lexp{d+|R|-1}$,
    \vspace*{4pt}
  \item
    $\tubh[m]{k}{R}=\tubh[m-1]{k}{R}-\tubh[m-1]{k-1}{R}$, $m>0$,
    \vspace*{4pt}
  \item
    $\tubk{k}{R}=\sum_{\emptyset\subset{}S\subseteq{}R}\tubh[|R|-|S|]{k}{R}$,
    \vspace*{4pt}
  \item
    $\tubh{k}{R}=\tubk{d+|R|-1-k}{R}$,
  \end{enumerate}
  where $\mb{n}_R$ stands for the $|R|$-dimensional vector whose
  elements are the values $n_i$, $i\in{}R$.
\end{definition}

Notice that $\tubh{k}{R}$ and $\tubk{k}{R}$ are well defined,
though in a recursive manner (in the size of $R$), since for
any $k>\lexp{d+|R|-1}$, we have:
\begin{align}
  \tubh{k}{R}&=\tubk{d+|R|-1-k}{R}
  =\sum\limits_{\emptyset\subset{}S\subseteq{}R}\tubh[|R|-|S|]{d+|R|-1-k}{S}\notag\\
  &=\tubh{d+|R|-1-k}{R}
  +\sum\limits_{\emptyset\subset{}S\subset{}R}\tubh[|R|-|S|]{d+|R|-1-k}{S}\notag\\
  &=\sum\limits_{\emptyset\subset{}S\subseteq{}R}(-1)^{|R|-|S|}
  \tbinom{n_S-k-1}{d+|R|-1-k}
  +\sum\limits_{\emptyset\subset{}S\subset{}R}\tubh[|R|-|S|]{d+|R|-1-k}{S},
  \label{equ:Phi-upper-half-last}
\end{align}
where the second sum in \eqref{equ:Phi-upper-half-last} is to be
understood as 0 when $|R|=1$. In other words, $\tubh{k}{R}$, and, thus,
also $\tubh[m]{k}{R}$ for any $m>0$, is fully defined for some $R$ and
any $k$, once we know the values $\tubh[\ell]{k}{S}$ for all
$\emptyset\subset{}S\subset{}R$, for all $0\le{}k\le{}d+|S|-1$, and for
all $1\le{}\ell\le{}|R|-1$. Moreover, it is easy to verify that
$\tubh{k}{R}$ satisfies the following recurrence relation:
\begin{equation}
  \tubh{k+1}{R}=\frac{n_R-d-|R|+k+1}{k+1}\tubh{k}{R}
  +\sum_{i\in{}R}\frac{n_i}{k+1}\tubh[1]{k}{R\sm\{i\}},
  \quad 0\le{}k\le{}\ltexp{d+|R|-2}.
\end{equation}

\begin{lemma}
  \label{lem:hk-alpha}
  For any $\emptyset\subset{}R\subseteq[r]$, any $k$ with
  $0\le{}k\le{}\lexp{d+|R|-1}$, and any $\alpha$ with
  $0\le\alpha\leq{}\frac{d+1}{d-1}$, we have:
  \begin{equation}
    \label{equ:hk-alpha}
    h_k(\fF_R)-\alpha \sum_{i\in{}R}h_{k-1}(\fF_{R\sm\{i\}})
    \leq \tubh{k}{R}-\alpha \sum_{i\in{}R}\tubh{k-1}{R\sm\{i\}}.
  \end{equation}
\end{lemma}

To prove Lemma \ref{lem:hk-alpha} we need the following intermediate result.

\begin{lemma}
  \label{lem:gk-alpha}
  For any $\emptyset\subset{}R\subseteq[r]$, any $k$ with
  $0\le{}k\le{}\lexp{d+|R|-1}$, and any $\alpha$ with
  $0\le\alpha\leq{}\frac{d+1}{d-1}$, we have:
  \begin{equation*}
    g_k(\fF_R)-\alpha \sum_{i\in{}R}g_{k-1}(\fF_{R\sm\{i\}})
    \leq \tubh[1]{k}{R}-\alpha \sum_{i\in{}R}\tubh[1]{k-1}{R\sm\{i\}}.
  \end{equation*}
\end{lemma}
\begin{proof}
  Let us recall the recurrence relation from Theorem \ref{thm:hkFRrec}:
  \begin{equation*}
    h_k(\fF_R)\le\frac{n_{R}-d-|R|+k}{k}h_{k-1}(\fF_R)+
    \sum_{i\in{}R}\frac{n_i}{k}g_{k-1}(\fF_{R\sm\{i\}}).
  \end{equation*}
  Subtracting $h_{k-1}(\fF_R)+\alpha\sum_{i\in{}R}g_{k-1}(\fF_{R\sm\{i\}})$
  from both sides of the inequality we get:
  \begin{equation}\label{equ:gk-alpha}
    g_k(\fF_R)-\alpha \sum_{i\in{}R}g_{k-1}(\fF_{R\sm\{i\}})
    \le\frac{n_{R}-d-|R|}{k}h_{k-1}(\fF_R)+
    \sum_{i\in{}R}\left(\frac{n_i}{k}-\alpha\right)g_{k-1}(\fF_{R\sm\{i\}}).
  \end{equation}
  Observe that the coefficient of $h_{k-1}(\fF_R)$ in \eqref{equ:gk-alpha}
  is non-negative:
  \begin{equation*}
    n_R-d-|R|\ge{}|R|(d+1)-d-|R|
    =d|R|+|R|-d-|R|
    =d(|R|-1)\ge{}0.
  \end{equation*}
  The same holds for the coefficient of $g_{k-1}(\fF_{R\sm\{i\}})$ in
  \eqref{equ:gk-alpha}, since:
  \begin{equation}
    \frac{n_i}{k}
    \ge\frac{d+1}{\lexp{d+|R|-1}}
    \ge\frac{d+1}{\frac{d+|R|-1}{2}}
    =\frac{2d+2}{d+|R|-1}
    \ge\frac{2d+2}{d+(d-1)-1}
    =\frac{2d+2}{2d-2}
    \ge\alpha,
  \end{equation}
  where we used the fact that $|R|\le{}r\le{}d-1$.
  Hence, we can bound \eqref{equ:gk-alpha} from above by substituting
  $g_{k-1}(\fF_R)$ and $g_{k-1}(\fF_{R\sm\{i\}})$, $i\in{}R$, by
  $\tubh[1]{k-1}{R}$ and $\tubh[1]{k-1}{R\sm\{i\}}$, $i\in{}R$,
  respectively. This gives:
  \begin{align*}
    g_k(\fF_R)-\alpha \sum_{i\in{}R}g_{k-1}(\fF_{R\sm\{i\}})
    &\le
    \frac{n_R-d-|R|}{k}\tubh[1]{k-1}{R}
    +\sum_{i\in{}R}\left(\frac{n_i}{k}-\alpha\right)\tubh[1]{k-1}{R\sm\{i\}}\\
    &=\frac{n_R-d-|R|}{k}\tubh[1]{k-1}{R}
    +\sum_{i\in{}R}\frac{n_i}{k}\tubh[1]{k-1}{R\sm\{i\}}\\
    &\qquad
    -\alpha\sum_{i\in{}R}\tubh[1]{k-1}{R\sm\{i\}}\\
    &=\tubh[1]{k}{R}-\alpha\sum_{i\in{}R}\tubh[1]{k-1}{R\sm\{i\}}.\qedhere
  \end{align*}
\end{proof}

Having established Lemma \ref{lem:gk-alpha}, it is now straightforward
to prove Lemma \ref{lem:hk-alpha}.

\begin{proof}[Proof of Lemma \ref{lem:hk-alpha}]
  First observe that $h_i(\fF_X)$ may be written as a telescopic sum
  as follows:
  \begin{equation}\label{equ:telescopic}
    h_i(\fF_X)=h_0(\fF_X)+\sum_{\ell=0}^{i-1}g_{i-\ell}(\fF_X).
  \end{equation}
  Since $h_0(\fF_X)=g_0(\fF_X)$, the above expansion may be written in
  the more concise form:
  \begin{equation}\label{equ:telescopic2}
    h_i(\fF_X)=\sum_{\ell=0}^{i}g_{i-\ell}(\fF_X).
  \end{equation}
  Using relations \eqref{equ:telescopic} and \eqref{equ:telescopic2},
  and applying Lemma \ref{lem:gk-alpha}, we get:
  \begin{align*}
    h_k(\fF_R)-\alpha \sum_{i\in{}R}h_{k-1}(\fF_{R\sm\{i\}})
    &=h_0(\fF_R)+\sum_{\ell=0}^{k-1}g_{k-\ell}(\fF_R)
    -\alpha\sum_{i\in{}R}\sum_{\ell=0}^{k-1}g_{k-1-\ell}(\fF_{R\sm\{i\}})\\
    &=h_0(\fF_R)+\sum_{\ell=0}^{k-1}\left(g_{k-\ell}(\fF_R)
    -\alpha\sum_{i\in{}R}g_{k-1-\ell}(\fF_{R\sm\{i\}})\right)\\
    &\le{}\tubh{0}{R}+\sum_{\ell=0}^{k-1}\left(\tubh[1]{k-\ell}{R}
    -\alpha\sum_{i\in{}R}\tubh[1]{k-1-\ell}{R\sm\{i\}}\right)\\
    &=\tubh{0}{R}+\sum_{\ell=0}^{k-1}\tubh[1]{k-\ell}{R}
    -\alpha\sum_{i\in{}R}\sum_{\ell=0}^{k-1}\tubh[1]{k-1-\ell}{R\sm\{i\}}\\
    &=\tubh{k}{R}-\alpha\sum_{i\in{}R}\tubh{k-1}{R\sm\{i\}},
  \end{align*}
  where we also used the fact that $h_0(\fF_R)=(-1)^{|R|-1}=\tubh{0}{R}$.
\end{proof}

The next theorem provides upper bounds for $h$-vectors of $\fF_R$ and
$\kK_R$, as well as necessary and sufficient conditions for these upper 
bounds to be attained. 
\begin{theorem}\label{thm:ub-hkF-hkK}
  For all $0\le{}k\le{}d+|R|-1$, we have:
  \begin{enumerate}[(i)]
  \item $h_k(\fF_R)\le{}\tubh{k}{R}$,
  \item $h_k(\kK_R)\le{}\tubk{k}{R}$.
  \end{enumerate}
  Equalities hold for all $k$ if and only if the Cayley polytope
  $\cC_R$ is Minkowski-neighborly.
\end{theorem}
\begin{proof}
  To prove the upper bounds use recursion on the size of $|R|$.
  For $|R|=1$, the result for both $h_k(\fF_R)$ and $h_k(\kK_R)$ comes from
  the UBT for $d$-polytopes.
  For $|R|>1$, we assume that the bounds hold for all $S$ with
  $\emptyset\subset{}S\subset{}R$, and for all $k$ with $0\le{}k\le{}d+|S|-1$.
  Furthermore, the upper bound for $h_k(\fF_R)$ for $k\le{}\lexp{d+|R|-1}$ is
  immediate from Theorem \ref{theor:ub_from_rec}. To prove the upper
  bound for $h_k(\kK_R)$, $0\le{}k\le{}\lexp{d+|R|-1}$, we use the
  following expansion for $h_k(\kK_R)$ (cf.~\cite[Lemma 5.14]{as-ubt-14}):
  \begin{equation}\label{equ:hKR-expansion}
    \begin{aligned}
      h_k(\kK_R)&=\sum_{j=0}^{\lexp{|R|}}
      \sum_{s=c-2j-1}^{|R|-2j}\sum_{\substack{S\subseteq{}R\\|S|=s}}
      \binom{|R|-s}{2j}\left(h_{k-2j}(\fF_S)
      -\frac{1}{2j+1}\sum_{i\in{}S}h_{k-2j-1}(\fF_{S\sm\{i\}})\right)\\
      &\quad
      +\sum_{j=0}^{\lexp{|R|}}\sum_{\substack{S\subset{}R\\|S|=c-2j+1}}\binom{|R|-|S|}{2j}
      \left(h_{k-2j}(\fF_S)
      -\frac{1}{2j+1}\sum_{i\in{}S}h_{k-2j-1}(\fF_{S\sm\{i\}})\right),
    \end{aligned}
  \end{equation}
  where $c$ depends on $k$, $d$ and $|R|$.
  Under the assumption that $r<d$, it is easy to show that (see
  Lemma \ref{lem:hk-alpha} in Section \ref{app:sec:ub-hKR} below):
  \begin{equation}\label{equ:expansion-ub}
    h_{k-2j}(\fF_S)
    -\frac{1}{2j+1}\sum_{i\in{}S}h_{k-2j-1}(\fF_{S\sm\{i\}})
    \le
    \tubh{k-2j}{S}-\frac{1}{2j+1}\sum_{i\in{}S}\tubh{k-2j-1}{S\sm\{i\}}.
  \end{equation}
  Substituting the upper bound from \eqref{equ:expansion-ub} in
  \eqref{equ:hKR-expansion}, and reversing the derivation logic for
  \eqref{equ:hKR-expansion}, we deduce that $h_k(\kK_R)\le{}\tubk{k}{R}$.

  For $k>\lexp{d+|R|-1}$ we have:
  \begin{align*}
    h_k(\fF_R)&=h_{d+|R|-1-k}(\kK_R)\le\tubk{d+|R|-1-k}{R}=\tubh{k}{R},
    \quad\text{and},\\
    h_k(\kK_R)&=h_{d+|R|-1-k}(\fF_R)\le\tubh{d+|R|-1-k}{R}=\tubk{k}{R}.
  \end{align*}

  The necessary and sufficient conditions are easy consequences of the
  equality claim in Theorem \ref{theor:ub_from_rec}.
\end{proof}

For any $d\ge{}2$, $\emptyset\subset{}R\subseteq{}[r]$,
$0\le{}k\le{}d+|R|-1$, and $n_i\in\naturals$, $i\in{}R$, with
$n_i\ge{}d+1$, let
\begin{equation*}
  \spans{k}{R}=
  \sum_{\emptyset\subset{}R\subseteq{}[r]}
  (-1)^{r-|R|}f_k\bigl(C_{d+r-1}(n_R)\bigr)+\sum_{i=0}^{\lexp{d+r-2}}
  \tbinom{i}{k-d-r+1+i}\sum_{\emptyset\subset{}R\subset{}[r]}
  \tubh[r-|R|]{i}{R},
\end{equation*}
where $C_\delta(n)$ stands for the cyclic $\delta$-polytope with 
$n$ vertices. It is straightforward to verify that for
$0\le{}k\le{}\lexp{d+|R|-1}$, $\spans{k}{R}$ simplifies to 
$\sum_{\emptyset\subset{}S\subseteq{}R}(-1)^{|R|-|S|}\binom{n_R}{k}$.
We are finally ready to state and prove the main result of the paper.

\begin{theorem}
\label{theor:ubt}
Let $P_1,\ldots,P_r$ be $r$ $d$-polytopes, $r<d$, with
$n_1,\ldots,n_r$ vertices respectively.
Then, for all $1\le{}k\le{}d$, we have:
\begin{equation*}
  f_{k-1}(P_{[r]})\leq{}\spans{k+r}{[r]}.
  \label{equ:ubt}
\end{equation*}
Equality holds for all $0\leq{}k\leq{}d$ if and only if the Cayley polytope 
$\cC_{[r]}$ of $P_1,\ldots,P_r$ is Minkowski-neighborly.
\end{theorem}
\begin{proof}
 We start by recalling that:
  \begin{align*}
    f_{k-1}(\fF_{[r]})&=\sum_{i=0}^{d+r-1}\tbinom{d+r-1-i}{k-i} h_i(\fF_{[r]}).
   \end{align*}
  In view of Theorem \ref{thm:ub-hkF-hkK}, the above expression is
  bounded from above by: 
   \begin{align} 
    &\sum_{i=0}^{\lexp{d+r-1}}\tbinom{d+r-1-i}{k-i}
    \tubh{i}{[r]}
    +\sum_{i=\lexp{d+r-1}+1}^{d+r-1}\tbinom{d+r-1-i}{k-i}\tubh{i}{[r]}
    \label{ubcalc2}\\
    &\qquad=\sum_{i=0}^{\lexp{d+r-1}}\tbinom{d+r-1-i}{k-i}\tubh{i}{[r]}
    +\sum_{i=0}^{\lexp{d+r-2}}\tbinom{i}{k-d-r+1+i}
    \tubh{d+r-1-i}{[r]}
    \label{ubcalc3}\\
    &\qquad=\sideset{}{^{\,*}}{\sum}_{i=0}^{\frac{d+r-1}{2}}
    \biggl(\tbinom{d+r-1-i}{k-i}+\tbinom{i}{k-d-r+1+i}\biggr) 
    \sum_{\emptyset\subset{}R\subseteq{}[r]}(-1)^{r-|R|}\tbinom{n_R-d-r+i}{i}\notag\\
    &\qquad\qquad
    +\sum_{i=0}^{\lexp{d+r-2}}
    \tbinom{i}{k-d-r+1+i}\sum_{\emptyset\subset{}R\subset{}[r]}
    \tubh[r-|R|]{i}{R}
    \label{ubcalc6}\\
    &\qquad=\sum_{\emptyset\subset{}R\subseteq{}[r]}(-1)^{r-|R|}
    f_{k}\bigl(C_{d+r-1}(n_R)\bigr)    
    +\sum_{i=0}^{\lexp{d+r-2}}\tbinom{j}{k-d-r+1+i}
    \sum_{\emptyset\subset{}R\subset{}[r]}\tubh[r-|R|]{i}{R}
    \label{ubcalc7}
   \end{align}
   where to go:
   \begin{itemize}
   \item from \eqref{ubcalc2} to \eqref{ubcalc3} we changed the
     variable of the second sum from $i$ to $d+r-1-i$,
   \item from \eqref{ubcalc3} to \eqref{ubcalc6} we wrote the explicit 
     expression of $\tubh{k}{R}$ from relation
     \eqref{equ:Phi-upper-half-last},
   \item from \eqref{ubcalc6} to \eqref{ubcalc7} we used that the number of 
     $(k-1)$-faces of a cyclic $\delta$-polytope with $n$ 
     vertices is $\sideset{}{^{\,*}}{\sum}_{i=0}^{\frac{\delta}{2}}
     \left(\binom{\delta-i}{k-i}+\binom{i}{k-\delta+i}\right)
     \binom{n-\delta-1+i}{i}$, where 
     $\displaystyle\sideset{}{^{\,*}}{\textstyle\sum}_{^{i=0}}^{_{\frac{\delta}{2}}}T_i$
     denotes the sum of the elements $T_0, T_1,\ldots,T_{\lexp{\delta}}$ 
     where the last term is halved if $\delta$ is even. 
   \end{itemize}
   
   Finally,  observing that the expression in \eqref{ubcalc7} is nothing but
   $\spans{k}{R}$, and recalling that $f_{k-1}(\fF_{[r]})=f_{k-r}(P_{[r]})$,
   we arrive at the upper bound in the statement of the theorem. 
   The equality claim is immediate from Theorem \ref{thm:ub-hkF-hkK}.
\end{proof}

%% file: lbconstruction.tex
\section{Tight bound construction}
\label{sec:tbconstruction}

\newcommand{\DDD}[2]{ D_{\scalebox{0.5}{#1}}(#2)}
\newcommand{\yee}[3]{
{y}_{\scalebox{1}{\hspace{-0.2 cm}  
$_{_\epsilon}$}\scalebox{0.8}{\hspace{-0.1 cm} 
$\scalebox{0.7}{$#1$},\scalebox{0.7}{$#2$}$}}^{\scalebox{0.6}{$#3$}}
}
\newcommand{\yy}[3]{y_{\scalebox{0.5}{$#1,#2$}}^{\scalebox{0.6}{$#3$}}}
\newcommand{\M}[2]{M_{\scalebox{0.7}{$#1$}}^{\scalebox{0.6}{$#2$}}}
\newcommand{\QQ}[1]{\hat P_{#1}}
\renewcommand{\wW}{\hat \fF}
\newcommand{\hC}[1]{\hat \cC_{#1}}

In this section we show that the bounds in Theorem \ref{theor:ubt} are
tight.
Before getting into the technical details, we  outline our approach. We start 
by considering the $(d-r+1)$-dimensional moment curve, which we embed in $r$ 
distinct subspaces of $\reals^d$. 
We consider the $r$ copies of the $\pddo$-dimensional moment curve as different 
curves, and we perturb them appropriately, so that they become $d$-dimensional 
moment-like curves. The perturbation is controlled via a non-negative parameter 
$\zeta$, which will be chosen appropriately. We then choose points on these $r$ 
moment-like curves, all parameterized by a positive parameter $\tau$, which 
will again be chosen appropriately.
These points are the vertices of $r$ $d$-polytopes
$P_1,P_2,\ldots,P_r$, and
we show that, for all $\emptyset\subset{}R\subseteq{}[r]$,
the number of $(k-1)$-faces of $\fF_R$, where $|R|\le{}k\le{}\lexp{d+|R|-1}$, 
becomes equal to $\spans{k}{R}$ for small enough positive values of 
$\zeta$ and $\tau$.
Our construction produces \emph{projected prod-simplicial neighborly} 
polytopes (cf.~\cite{mpp-pnp-11}). For $\zeta=0$ our polytopes
are essentially the same as those in
\cite[Theorem 2.6]{mpp-pnp-11}, while for $\zeta>0$ we get \emph{deformed}
versions of those polytopes. 
The positivity of $\zeta$ allows us to ensure the tightness of the upper
bound on $f_k(\MS)$, not only for small, but also for large values of $k$.

At a more technical level (cf. Section \ref{ss:technical}), the proof that 
$f_{k-1}(\fF_R)=\spans{k}{R}$, for all $|R|\le{}k\le\lexp{d+|R|-1}$, 
is performed in two steps. We first consider the cyclic $(d-r+1)$-polytopes 
$\QQ{1},\ldots,\QQ{r}$, embedded in appropriate subspaces of $\reals^d$. The 
$\QQ{i}$'s are the \emph{unperturbed}, with respect to $\zeta$,
versions of the $d$-polytopes $\range{P}$ (i.e., the polytope $\QQ{i}$
is the polytope we get from $P_i$, when we set $\zeta$ equal to
zero). For each $\emptyset\subset{}R\subseteq[r]$  
we denote by $\hC{R}$ the Cayley polytope of $\QQ{i},i\in{}R$, seen as a 
polytope in $\reals^d$, and we focus on the set $\wW_R$ of its mixed faces. 
Recall that the polytopes $\QQ{i},i\in{}R$, are parameterized by the 
parameter $\tau$; we show that there exists a sufficiently small
positive value $\tau^\star$ for $\tau$, for which the number of
$(k-1)$-faces of $\wW_R$ is equal to $\spans{k}{R}$ for all
$|R|\le{}k\le{}\lexp{d+|R|-1}$. For $\tau$ equal to $\tau^\star$,
we consider the polytopes $\range{P}$ (with $\tau$ set to
$\tau^\star$), and show that for sufficiently small $\zeta$ (denoted
by $\zeta^{\lozenge}$), $f_{k-1}(\fF_R)$ is equal to $\spans{k}{R}$.

\renewcommand{\yee}[3]{
{y}_{\scalebox{1}{\hspace{-0.15 cm}  
$_{_\epsilon}$}\scalebox{0.8}{\hspace{-0.1 cm} 
$\scalebox{0.7}{$#1$},\scalebox{0.7}{$#2$}$}}^{\scalebox{0.6}{$#3$}}
}
\renewcommand{\yy}[3]{y_{\scalebox{0.7}{$#1,#2$}}^{\scalebox{0.6}{$#3$}}}

\renewcommand{\ye}[2]{{\tilde y_{\scalebox{0.7}{$#1,#2$}}}}
\newcommand{\ty}[3]{{\widetilde 
y^{\scalebox{0.6}{$#3$}}_{\scalebox{0.7}{$#1,#2$}}}}

In the remainder of this section we describe our construction in
detail.
For each $1\leq{}i\leq{}r$, we define the $d$-dimensional
moment-like curve\footnote{The curve $\gamma_i(t;\zeta)$,
  $\zeta>0$, is the image under an invertible linear transformation,
  of the curve $\hat\gamma_i(t)=(t,t^2,\ldots,t^{d-r+i},t^{d-r+i+2},\ldots,t^{d+1})$.
  Polytopes whose vertices are $n$ distinct points on this curve are
  combinatorially equivalent to the cyclic $d$-polytope with $n$
  vertices.}:
\begin{align*}
  \mc_i(t;\z)=&\overset{\hspace{0.3cm}\arrow i-th\;
    \mbox{\tiny coordinate}}{
    (\z t^{d-r+2},\ldots,\z t^{d-r+i},t,\z t^{d-r+i+2},\ldots,
    \z t^{d+1},t^2,\ldots,t^{d-r+1})},
\end{align*}
and the $d$-polytope 
\begin{align}
  \label{pol:Pi}
  P_i:= CH\{\mc_i(y_{i,1};\zeta),\ldots, \mc_i(y_{i,n_i};\zeta)\},
\end{align}
where the parameters $y_{i,j}$ belong to the sets
$Y_i=\{y_{i,1},\ldots,y_{i,n_i}\}$, $1\le{}i\le{}r$, whose elements
are determined as follows. Choose
\begin{itemize}
\item 
  $n_{[r]}+d+r$ arbitrary real numbers $x_{i,j}$ and $M_s$, such that:
  \begin{itemize}
  \item 
    $0<x_{i,1}<x_{i,1}+\epsilon<x_{i,2}<x_{i,2}+\epsilon<\cdots<
    x_{i,n_i}+\epsilon$, for $1\leq{}i\leq{}r-1$,
  \item $0<x_{r,1}<x_{r,1}+\epsilon<x_{r,2}<x_{r,2}+\epsilon<\cdots<
    x_{r,n_r}+\epsilon<M'_1<\cdots<M'_{d+r}$,
  \end{itemize}
  where $\epsilon>0$ is sufficiently small and $x_{i,n_i}<x_{i+1,1}$ for all 
  $i$, and 
\item $r$ non-negative integers $\beta_1,\beta_2,\ldots,\beta_r$, such that
  $\beta_1>\beta_2>\cdots>\beta_{r-1}>\beta_r\ge{}0$.
\end{itemize}
We then set $y_{i,j}:=x_{i,j}\tau^{\beta_i}$, 
$\ye{i}{j}:=(x_{i,j}+\epsilon)\tau^{\beta_i}$
and $M_i:=M'_i\tau^{\beta_r}$, where $\tau$ is a positive parameter.
The $y_{i,j}$'s and $\ye{i}{j}$'s are used to define determinants
whose value is positive for a small enough value of $\tau$ (see also
Lemma \ref{lem:asymptotic-tau} in the Appendix).
The positivity of these determinants is crucial in defining
supporting hyperplanes for the Cayley polytopes $\hC{R}$ and $\cC_R$
in Lemmas \ref{lem:Qstar} and \ref{lem:Pstar} below. 

Next, for each $1\leq{}i\leq{}r$, we  define 
$\QQ{i}:=\lim_{\zeta\rightarrow 0^+}P_i$. Clearly, each $\QQ{i}$ is a cyclic 
$(d-r+1)$-polytope embedded in the $\pddo$-flat $F_i$ of $\reals^d$, where 
$F_i=\{x_j=0\mid{}1\le{}j\le{}r\text{ and }j\ne{}i\}.$
The following lemma establishes the first step towards our construction.
\begin{lemma}
\label{lem:Qstar}
There exists a sufficiently small positive value $\tau^\star$ for
$\tau$, such that, for any $\emptyset\subset{}R\subseteq{}[r]$, the set
of mixed faces $\wW_R$ of the Cayley polytope of the polytopes
$\QQ{1},\ldots,\QQ{r}$ constructed above, has
\begin{equation*}
  f_{k-1}(\wW_R)=\spans{k}{R},
  \quad
  |R|\le{}k\le{}\ltexp{d+|R|-1}.
\end{equation*}
\end{lemma}
\begin{proof} 
Let  $\UU_i$  be the set of vertices of $\QQ{i}$ for $1\leq{}i\leq{}r$
and set $\UU:=\cup_{i=1}^r\UU_i$. 
The objective in the proof is, for each $\emptyset\subset{}R\subseteq[r]$
and each spanning subset  $U$ of the partition $\ptn[U]=\cup_{i\in{}R}\UU_i$, 
to exhibit a supporting hyperplane of the $(d+|R|-1)$-dimensional 
Cayley polytope $\hC{R}$, containing exactly the vertices in $U$. In that 
respect, our approach is similar in spirit, albeit much more  technically 
involved, to the proof showing, by defining supporting hyperplanes
constructed from Vandermonde determinants, that the cyclic $n$-vertex
$d$-polytope $C_d(n)$ is neighborly (see, e.g. \cite[Corollary
  0.8]{z-lp-95}).

In our proof we need to involve the parameter $\zeta$ before taking
the limit $\zeta\rightarrow{}0^+$. This is due to the fact that, when 
$\emptyset\subset{}R\subset{}[r]$, the information of the  relative
position of the polytopes $\QQ{i}$, $i\in{}R$,  is lost if we set
$\zeta=0$ from the very first step.
To describe our construction, we write each  spanning subset  $U$ of
$\ptn[U]=\cup_{i\in{}R}\UU_i$ as the disjoint union of non-empty sets
$U_i$, $i\in{}R$, where $|U_i|=\kappa_i\le{}n_i$,
$U_i=\lim_{\zeta\rightarrow{}0}\{\mc_i(y;\zeta):y\in{}Y'_i\}
=\{\mc_i(y;0):y\in{}Y'_i\}$
and $Y'_i=\{y\in{}Y_i\mid{}\mc_i(y;0)\in{}U_i\}$.
For this particular $U$, we define the linear equation:
\begin{equation}
  \label{equ:hyperplane}
  H_U(\mb{x})=\lim_{\zeta\rightarrow{}0^+}
  (-1)^{\tfrac{|R|(|R|-1)}{2}+\sigma(R)}\zeta^{|R|-r}{\sf{}D}_U(\mb{x};\zeta), 
\end{equation}
where $\mb{x}=(x_1,x_2,\ldots,x_{d+|R|-1})$, and
${\sf{}D}_U(\mb{x};\zeta)$ is the $(d+|R|)\times(d+|R|)$ 
determinant\footnote{We refer the reader to Figs. \ref{fig:bigdet} and
  \ref{fig:bigdet0} in the Appendix for an example of
  ${\sf{}D}_U(\mb{x};\zeta)$, $\zeta>0$, and ${\sf{}D}_U(\mb{x};0)$.}:
\begin{itemize}
\item whose first column is $(1,\mb{x})^\intercal$, 
\item the next $\kappa_i$, $i\in{}R$, pairs of columns are 
  $(1,\me_{i-1},\gamma_i(\yy{i}{j}{};\zeta))^\intercal$
  and $(1,\me_{i-1},\gamma_i(\ty{i}{j}{};\zeta))^\intercal$
  where $\me_0,\ldots,\me_{|R|-1}$ is the standard affine 
  basis of $\reals^{|R|-1}$ and $j\in{}Y'_i$,  and
\item 
  the last $s:=d+|R|-1-\sum_{i\in{}R}\kappa_i$
    columns are $(1,\me_{|R|-1},\gamma_{|R|-1}(M_i;\zeta))^\intercal$, 
    $1\leq{}i\leq{}s$; these columns exist only if $s>0$.
\end{itemize}
The quantity $\sigma(R)$ above is a non-negative integer counting the
total number of row swaps required to shift, for all $j\in{}[r]\sm{}R$, 
the $(|R|+j)$-th row of ${\sf{}D}_U(\mb{x};\zeta)$ to the bottom of the
determinant, so that the powers of $y_{i,j}$ in each column are in
increasing order (notice that if $R\equiv{}[r]$ no such row swaps are
required).
Moreover, $\sigma(R)$ depends only on $R$ and not on the
choice of the spanning subset $U$ of $\ptn[U]$.

The equation $H_U(\mb{x})=0$ is the equation of a hyperplane in 
$\reals^{d+|R|-1}$ that passes through the points in $U$.
We claim that, for any choice of $U$, and for all 
vertices $\mbu$ in $\UU\sm{}U$, we have $H_U(\mbu)>0$. 
To prove our claim, notice first that, for each $j\in{}[r]\sm{}R$, 
the $(|R|+j)$-th row of the determinant ${\sf{}D}_U(\mb{u};\zeta)$ will 
contain the parameters  $\yy{i}{j}{d-r+1+j}, \ty{i}{j}{d-r+1+j}$ 
multiplied by $\zeta$.  
After extracting $\zeta$ from each of these rows and shifting it to its 
\emph{proper} position (i.e., the position where the powers along each
column increase), we will have a term $\zeta^{r-|R|}$ and a sign 
$(-1)^{\sigma(R)}$ (induced from the $\sigma(R)$ row swaps required 
altogether). 
These terms cancel out with the term $(-1)^{\sigma(R)}\zeta^{|R|-r}$
in \eqref{equ:hyperplane}.
We can, therefore, transform $H_U(\mbu)$ in  the form of  the determinant
$\DDD{K}{\ptn[Y];\mu_1,\ldots,\mu_m}$ shown below:
\begin{equation*}
\label{detD}
\DDD{K}{\ptn[Y];\mu_1,\ldots,\mu_m}:= (-1)^{\tfrac{n(n-1)}{2}}
\scalebox{0.9}{
$  
\left|
\begin{array}{ccccccccccccccc}
\y{1}{1}^{\mu_1}&\cdots&\y{1}{\kappa_1}^{\mu_1}&
\0  &\cd    &\0  &\cd   &\0  &\cd   &\0\\
\0 &\cd    &\0  & 
\y{2}{1}^{\mu_1}&\cdots&\y{2}{\kappa_2}^{\mu_1}&
\cd   &\0  &\cd   &\0\\
\vd&\dd    &\vd &\vd  &\dd   &\vd &\cd   &\vd &\dd   &\vd\\
\0 &\cd    &\0  &\0  &\cd    &\0  &\cd  &
\y{n}{1}^{\mu_1}&\cdots&\y{n}{\kappa_1}^{\mu_1} \\
\y{1}{1}^{\mu_2}&\cdots&\y{1}{\kappa_1}^{\mu_2}&
\0&\cd&\0&  \cd&
\0&\cd&\0\\
\0&\cd&\0&
\y{2}{1}^{\mu_2}&\cdots&\y{2}{\kappa_2}^{\mu_2}&  \cd&
\0&\cd&\0\\
\vd&\dd&\vd& 
\vd&\dd&\vd& \cd&
\vd&\dd&\vd\\
\0&\cd&\0&
\0&\cd&\0& \cd&
\y{n}{1}^{\mu_2}&\cdots&\y{n}{\kappa_n}^{\mu_2}\\
\y{1}{1}^{\mu_3}&\cdots&\y{1}{\kappa_1}^{\mu_3}&  
\y{2}{1}^{\mu_3}&\cdots&\y{2}{\kappa_2}^{\mu_3}&   \cd& 
\y{n}{1}^{\mu_3}&\cdots&\y{n}{\kappa_n}^{\mu_3}\\
\vd&\dd&\vd& \vd&\dd&\vd&  \cd& \vd&\dd&\vd\\
\y{1}{1}^{\mu_m}&\cdots&\y{1}{\kappa_1}^{\mu_m}&  
\y{2}{1}^{\mu_m}&\cdots&\y{2}{\kappa_2}^{\mu_m}&   \cd& 
\y{n}{1}^{\mu_m}&\cdots&\y{n}{\kappa_n}^{\mu_m}
\end{array}\right|,
$}
\end{equation*}
by means of the following determinant transformations:
\begin{enumerate}
\item
  By subtracting rows $2$ to $|R|$ of $H_U(\mbu)$ from its first row.
\item 
  By shifting the first column of $H_U(\mb{u})$ to the right,
  so that all columns of  $H_U(\mb{u})$  are arranged in 
  increasing order according to their parameter. 
  Clearly, this can be done with an \emph{even} number of
  column swaps.
\end{enumerate}
The determinant $\DDD{K}{\ptn[Y];\mu_1,\ldots,\mu_m}$ is strictly
  positive for all $\tau$ between $0$ and some value
  $\hat\tau(R,U,\mb{u})$, that, depends (only) on the choice of $R$, $U$ and
  $\mb{u}$.
  Since there is a finite number of possible such determinants, the
  value $\hat\tau^\star:=\min_{R,U,\mb{u}}\hat\tau(R, U, \mb{u})$ is
  necessarily positive. Choosing some
  $\tau^\star\in{}(0,\hat\tau^\star)$ makes all these determinants
  simultaneously positive; this completes our proof.
\end{proof} 

The following lemma establishes the second (and last) step of  our construction.

\begin{lemma}\label{lem:Pstar}
  There exists a sufficiently small positive value 
  $\zeta^{\lozenge}$ for
  $\zeta$, such that, for any $\emptyset\subset{}R\subseteq{}[r]$,
  the set $\fF_R$ of mixed faces of the Cayley polytope $\cC_R$ 
  of the polytopes $P_1,\ldots,P_r$  in \eqref{pol:Pi} has
  \begin{equation*}
    f_{k-1}(\fF_R)=\spans{k}{R},\quad\mbox{ for all }\quad 
    |R|\le{}k\le{}\ltexp{d+|R|-1}.
  \end{equation*}
\end{lemma}
\begin{proof} 
Briefly speaking, the value $\zeta^{\lozenge}$  is determined by
replacing the limit $\zeta\rightarrow{}0^+$ in the previous proof, by a
specific value of $\zeta$ for which the determinants we consider are
positive.

More precisely, let $\UU_i$  be the set of vertices of $P_i$, 
$1\leq{}i\leq{}r$,  and set $\UU:=\cup_{i=1}^r\UU_i$. Our goal is, for each 
$\emptyset\subset{}R\subseteq[r]$ and each spanning subset  $U$ of the 
partition $\ptn[U]=\cup_{i\in{}R}\UU_i$, to exhibit a supporting hyperplane of 
the Cayley polytope $\cC_R$, containing exactly the vertices in $U$. 
To this end, we define the linear equation
${\widetilde H}_U(\mb{x};\zeta)=0$,
$\mb{x}=(x_1,x_2,\ldots,x_{d+|R|-1})$, with
\begin{equation}
  \label{equ:hyperplane2}
        {\widetilde H}_U(\mb{x};\zeta)=
        (-1)^{\frac{|R|(|R|-1)}{2}+\sigma(R)}\zeta^{|R|-r} 
        {\sf{}D}_U(\mb{x};\zeta),\quad \zeta>0,
\end{equation}
where ${\sf{}D}_U(\mb{x};\zeta)$ is the determinant in the
proof of Lemma \ref{lem:Qstar}, where we have set $\tau$ to $\tau^\star$. 
Clearly, for each $\mb{u}\in \UU\sm U$, we have
$\lim_{\zeta\rightarrow{}0^+}{\widetilde H}_U(\mb{u};\zeta)= H_U(\mb{u})>0$.
This immediately implies that for each combination of $U$ and $\mb{u}$
there exists a value $\hat\zeta(U,\mb{u})$ such that, for all
$\zeta\in(0,\hat\zeta(U,\mb{u}))$, ${\widetilde H}_U(\mb{u};\zeta)>0$,
which, due to the positivity of $\zeta$, yields that
$\zeta^{r-|R|}{\widetilde H}_U(\mb{u};\zeta)>0$.
Since the number of possible combinations for $U$ and $\mb{u}$ is
finite, the minimum $\hat\zeta^\lozenge:=\min_{U,\mb{u}}\{\hat\zeta(U,\mb{u})\}$
is well defined and positive. Taking $\zeta^{\lozenge}$ to be any
value in $(0,\hat\zeta^\lozenge)$, satisfies our demands.
\end{proof}
\renewcommand{\arraystretch}{1.1}
\renewcommand{\yy}[3]{y_{\scalebox{0.7}{$#1,#2$}}^{\scalebox{0.6}{$#3$}}}
\renewcommand{\ye}[2]{
y_{ \scalebox{1}{\hspace{-0.2 cm}  
$_{_\epsilon}$}\scalebox{0.8}{\hspace{-0.1 cm} 
$\scalebox{0.7}{$#1$},\scalebox{0.7}{$#2$}$}}
}
\renewcommand{\yee}[3]{
{y}_{\scalebox{1}{\hspace{-0.2 cm}  
$_{_\epsilon}$}\scalebox{0.8}{\hspace{-0.1 cm} 
$\scalebox{0.7}{$#1$},\scalebox{0.7}{$#2$}$}}^{\scalebox{0.6}{$#3$}}
}

\subsection{Examples of determinants appearing in the tightness construction}

 \renewcommand\undermat[3][0pt]{%
  \makebox[0pt][l]{$\smash{\underbrace{\phantom{%
    \begin{matrix}\phantom{\rule{-100pt}{#1}}#3\end{matrix}}}_{\text{#2}}}$}#3}
\renewcommand{\arraystretch}{1.4}


  The determinant in Fig. \ref{fig:bigdet} is the determinant
  ${\sf{}D}_U(\mb{x};\zeta)$ that corresponds to the linear 
  equation $H_U(\mb{x})$ defined in the proof of Lemma \ref{lem:Qstar}, 
  in the case where $R=[r]$ and  $Y_i=\{y_{i,1},\ldots,y_{i,\kappa_i}\}$, for 
  all
  $1\leq{}i\leq{}r$.
  The determinant in Fig. \ref{fig:bigdet0} is the same as in Fig.
  \ref{fig:bigdet} after having taken the limit $\zeta\rightarrow{}0^+$.
\begin{landscape}
\vspace*{1cm}
\begin{figure}[H]
\scalebox{0.90}
{ 
$  
\hspace{0cm} 
\left|
\begin{array}{c:r@{}lr@{}lcr@{}lr@{}l:r@{}lr@{}lcr@{}lr@{}lc:r@{}lr@{}lcr@{}lr@{}lcr@{}l}
1&&
1&& 1&\cdots&& 1&& 1&&
1&& 1&\cdots&& 1&& 1&   \cdots& &
1&& 1&\cdots&& 1&& 1    &\cdots && 1\\ 
x_1&&
\0&& \0&\cd&&\0&&\0&&
1 && 1& \cd&& 1&& 1& \cdots&&
\0&& \0&\cd&&\0&&\0   
&\cdots && 1 \\ 
\vdots&&
\vd&&\vd&\dd&&\vd&&\vd&&
\0&& \0&\cd&& \0&&\0&  \cdots&&
\0&& \0&\cd&& \0&&\0&  
 \ddots&&\vdots \\
x_{r-1}&&
\0&& \0&\cd&& \0&&\0&&
\0&& \0&\cd&& \0&&\0&  \cdots&&      
  1&& 1&\cdots&& 1&& 1&      
 \cdots&& 1 \\
\hdashline
x_{r}&&
\yy{1}{1}{}&&\ty{1}{1}{}&\cdots&&\yy{1}{\kappa_1}{}&&\ty{1}{\kappa_1}{}&
\zz{}&\yy{2}{1}{d-r+2}&\zz{}&\ty{2}{1}{d-r+2}{}
&\cdots&\zz{}&\yy{2}{\kappa_2}{d-r+2}&\zz{}&\ty{2}{\kappa_2}{d-r+2}
&\cdots&
\zz{}&\yy{r}{1}{d-r+2}&\zz{}&\ty{r}{1}{d-r+2}&
\cdots&\zz{}&\yy{r}{\kappa_r}{d-r+2}&\zz{}&\ty{r}{\kappa_r}{d-r+2}
&\cdots &\zz{}& \M{s}{d-r+2} \\ 
x_{r+1}&
\zz{}&\yy{1}{1}{d-r+3}&\zz{}&\ty{1}{1}{d-r+3}&\cdots&
\zz{}&\yy{1}{\kappa_1}{d-r+3}&\zz{}&\yee{1}{\kappa_1}{d-r+3}&
&\yy{2}{1}{}&&\ty{2}{1}{}&\cdots&&
\yy{2}{\kappa_2}{}&&\ty{2}{\kappa_2}{} 
&\cdots&
\zz{}&\yy{r}{1}{d-r+3}&\zz{}&\ty{r}{1}{d-r+3}&
\cdots&\zz{}&\yy{r}{\kappa_r}{d-r+3}&\zz{}&\ty{r}{\kappa_r}{d-r+3} 
&\cdots &\zz{}& \M{s}{d-r+3}\\
x_{r+2}&
\zz{}&\yy{1}{1}{d-r+4}&\zz{}&\ty{1}{1}{d-r+4}&\cdots&
\zz{}&\yy{1}{\kappa_1}{d-r+4}&\zz{}&\ty{1}{\kappa_1}{d-r+4}&
\zz{}&\yy{2}{1}{d-r+4}&\zz{}&\ty{2}{1}{d-r+4}   
&\cdots& 
\zz{}&\yy{2}{\kappa_2}{d-r+4}&\zz{}&\ty{2}{\kappa_2}{d-r+4} 
&\cdots&
\zz{}&\yy{r}{1}{d-r+4}&\zz{}&\ty{r}{1}{d-r+4}&
\cdots&\zz{}&\yy{r}{\kappa_r}{d-r+4}&\zz{}&\ty{r}{\kappa_r}{d-r+4} 
&\cdots &\zz{}& \M{s}{d-r+4} \\
\vdots &&
\vd&&\vd&\dd&&\vd&&\vd&&
\vd&&\vd&\dd&&\vd&&\vd
&\ddots&
&\vdots&&
\vdots&&\vdots&&\vdots 
&&\ddots && \vdots\\
x_{2r-1}&
\zz{}&\yy{1}{1}{d+1}&\zz{}&\ty{1}{1}{d+1}&\cdots&
\zz{}&\yy{1}{\kappa_1}{d+1}&\zz{}&\ty{1}{\kappa_1}{d+1}&
\zz{}&\yy{2}{1}{d+1}&\zz{}&\ty{2}{1}{d+1}&\cdots&
\zz{}&\yy{2}{\kappa_2}{d+1}&\zz{} &\ty{2}{\kappa_2}{d+1}
&\cdots&
&\yy{r}{1}{d+1}&&\ty{r}{1}{d+1}&
\cdots&&\yy{r}{\kappa_r}{d+1}&&\ty{r}{\kappa_r}{d+1} 
&\cdots && \M{s}{}\\
\hdashline
x_{2r}&&
\yy{1}{1}{2}&&\ty{1}{1}{2}&\cdots&&\yy{1}{\kappa_1}{2}
&&\ty{1}{\kappa_1}{2}&
&\yy{2}{1}{2}&& \ty{2}{1}{2}&\cdots
&&\yy{2}{\kappa_2}{2}&&\ty{2}{\kappa_2}{2}
&\cdots&
&\yy{r}{1}{2}&&\ty{r}{1}{2}& \cdots 
&&\yy{r}{\kappa_r}{2}&&\ty{r}{\kappa_r}{2}
&\cdots && \M{s}{2}\\
x_{2r+1}&&
\yy{1}{1}{3}&&\ty{1}{1}{3}&\cdots&&\yy{1}{\kappa_1}{3}
&&\ty{1}{\kappa_1}{3}&
&\yy{2}{1}{3}&&\ty{2}{1}{3}&\cdots
&&\yy{2}{\kappa_2}{3}&&\ty{2}{\kappa_2}{3}
&\cdots&
&\yy{r}{1}{3}&&\ty{r}{1}{3}& \cdots 
&&\yy{r}{\kappa_r}{3}&&\ty{r}{\kappa_r}{3}
&\cdots && \M{s}{3}\\
\vdots&&
\vd&&\vd&\dd&&\vd&&\vd
&&\vd&&\vd&\dd&&\vd&&\vd
  &\ddots&
  &\vdots&& \vdots & \ddots && \vdots && \vdots    
&\ddots && \vdots\\
x_{d+r-1}&&
\yy{1}{1}{d-r+1}&&\ty{1}{1}{d-r+1}&\cdots
&&\yy{1}{\kappa_1}{d-r+1}&&\ty{1}{\kappa_1}{d-r+1}&
&\yy{2}{1}{d-r+1}&&\ty{2}{1}{d-r+1}
&\cdots&&\yy{2}{\kappa_2}{d-r+1} &&\ty{2}{\kappa_2}{d-r+1}

&\cdots&
&\yy{r}{1}{d-r+1}&&\ty{r}{1}{d-r+1}& \cdots 
&&\yy{r}{\kappa_r}{d-r+1}&&\ty{r}{\kappa_r}{d-r+1}
&\cdots&& \M{s}{d-r+1}
\end{array}\right|$}
\vspace*{5mm}
\caption{The determinant ${\sf{}D}_U(\mb{x};\zeta)$, for $R=[r]$.}
\label{fig:bigdet}
\end{figure}
\end{landscape}

\begin{figure}[H]
\scalebox{0.78}{ 
$  
\left| 
\begin{array}
{c:c@{}c@{}c@{}c@{}c@{}:c@{}c@{}c@{}c@{}ccc@{}
c@{}c@{}c@{}c@{}c@{}c@{}c@{}c@{}c@{}c@{}c@{}c@{}c@{}c@{}c@{}ccccc}
1&
1& 1&\cdots& 1& 1&
1& 1&\cdots& 1& 1& \cdots&
1& 1&\cdots& 1& 1&
1& \cdots& 1\\
x_1&
\0&\0&\cd&\0&\0&
1& 1&\cdots& 1& 1& \cd&
\0&\0&\cd&\0&\0&
\0& \0& \0\\
\vdots&
\vd&\vd&\cd&\vd&\vd&
\vd&\vd&\cd&\vd&\vd&   \ddots&
\vd&\vd&\cd&\vd&\vd&
\0&\0&\0\\
x_{r-1}&
\0&\0&\cd&\0&\0&
\0&\0&\cd&\0&\0& \cd& 
1& 1&\cdots& 1& 1&
1& 1& 1\\
x_{r}& 
\yy{1}{1}{}&\ty{1}{1}{}&\cdots&\yy{1}{\kappa_1}{}&\ty{1}{\kappa_1}{}&
\0&\0&\cd&\0&\0&  \cd&
\0&\0&\cd&\0&\0&
\0&\0&\0\\
 x_{r+1}&
\0&\0&\cd&\0&\0& 
\yy{2}{1}{}&\ty{2}{1}{}&\cdots&\y{2}{\kappa_2}&\ty{2}{\kappa_2}{}&
\cd&\0&\0&\cd&\0&\0&
\0&\0 &\0\\
\vdots&
\vd&\vd&\cd&\vd&\vd&
\vd&\vd&\cd&\vd&\vd&  \ddots&
\vd&\vd&\cd&\vd&\vd& 
\vd&\cd&\vd\\
x_{2r-1}&
\0&\0&\cd&\0&\0&
\0&\0&\cd&\0&\0& \cd&
\yy{r}{1}{}&\ty{r}{1}{}&\cdots&\yy{r}{\kappa_r}{}&\ty{r}{\kappa_r}{}&
\M{1}{}&\cdots& \M{s}{}\\
x_{2r}&
\yy{1}{1}{2}&\ty{1}{1}{2}&\cdots&\yy{1}{\kappa_1}{2}&\ty{1}{\kappa_1}{2}&
\yy{2}{1}{2}&\ty{2}{1}{2}&\cdots&\yy{2}{\kappa_2}{2}&\ty{2}{\kappa_2}{2}&\cdots&
\yy{r}{1}{2}&\ty{r}{1}{2}&\cdots&\yy{r}{\kappa_r}{2}&\ty{r}{\kappa_r}{2}& 
\M{1}{2} & \cdots & M_s^2\\
x_{2r+1}&
\yy{1}{1}{3}&\ty{1}{1}{3}&\cdots&\yy{1}{\kappa_1}{3}&\ty{1}{\kappa_1}{3}&
\yy{2}{1}{3}&\ty{2}{1}{3}&\cdots&\yy{2}{\kappa_2}{3}&\ty{2}{\kappa_2}{3}&\cdots&
\yy{r}{1}{3}&\ty{r}{1}{3}&\cdots&\yy{r}{\kappa_r}{3}&\ty{r}{\kappa_r}{3}& 
M_1^3 & \cdots & M_s^3\\
\vdots&\vdots&\vdots&\cdots&\vdots&\vdots&\vdots&\vdots&\cdots&\vdots&
\vdots&\ddots&\vdots&\vdots&\cdots&\vdots&\vdots&\vdots&\cdots&\vdots\\
x_{d+r-1}&\undermat[0pt]{$\kappa_1$ pairs of points from  
\scalebox{1}{$(\me_0,\mc_1(\cdot))$}}{\yy{1}{1}{d-r+1}&\ty{1}{1}{d-r+1}&
         \cdots&\yy{1}{\kappa_1}{d-r+1}&\ty{1}{\kappa_1}{d-r+1}}&
 \undermat[10pt]{$\kappa_2$ 
}
{\yy{2}{1}{d-r+1}&\ty{2}{1}{d-r+1}&\cdots&\yy{2}{\kappa_2}{d-r+1}&\ty{2}{\kappa_2}{d-r+1}}
 &\cdots&
\undermat[10pt]{$\kappa_r$
 }
 {\yy{r}{1}{d-r+1}&\ty{r}{1}{d-r+1}&\cdots&\yy{r}{\kappa_r}{d-r+1}&\ty{r}{\kappa_r}{d-r+1}}&
 \undermat[10pt]{$\substack{\mbox{auxiliary  columns}\\\mbox{ if necessary} 
 }$}{M_1^{d-r+1} & \cdots & M_s^{d-r+1}}\\
\end{array}\right|$}
\vspace{0.6cm}
\caption{The determinant
  ${\sf{}D}_U(\mb{x};0)=\lim_{\zeta\rightarrow{}0^+}{\sf{}D}(\mb{u};\zeta)$,
  for $R=[r]$.}
\label{fig:bigdet0}
\end{figure}


%% file: appendix_Dehn_Som.tex
\section{Special sets related to the derivation of the
  Dehn-Sommerville equations}
\label{app:sec:DS}


To prove Lemma \ref{lem:stirling}, we introduce a 
couple of sets that appear in the face counting of $\bx\qQ_R$.
For any $m\in\naturals$, and $S\subseteq{}[m]$ we define:
\begin{equation}
  \label{chain1}
  \spsa_m(S,k):=\{(S_1,S_2,\ldots,S_k)
  \mid{}S\subseteq{}S_1\subset{}S_2\subset{}\cdots\subset{}S_k\subset{}[m]\},
\end{equation}
\begin{equation}
  \label{chain2}
  \spsb_m(S,k):= \{(S_0,S_1,\ldots,S_{k-1})\mid{}
  S={}S_0\subset{}S_1\subset{}\cdots\subset{}S_{k-1}\subset{}[m]\}.
\end{equation}
Furthermore, we denote by $\spsac_m(S,k)$ and $\spsbc_m(S,k)$ the
cardinalities of $\spsa_m(S,k)$ and $\spsb_m(S,k)$ respectively.
It is immediate to see that:
\begin{align*}
  \spsac_m(S,k)&=\spsac_{\fs m-|S|}(\emptyset,k),\quad\text{and}\\
  \spsbc_m(S,k)&=\spsbc_{\fs m-|S|}(\emptyset,k).
\end{align*}


\begin{lemma} For any $k,m\in\naturals$, with $k\leq{}m$, we have:
\label{lem:chain_counting}
\begin{enumerate}[(i)]
\item $\spsbc_m(\emptyset,k)=k!\,\stirl{m}{k}$,
\item $\spsac_m(\emptyset,k)=k!\,\stirl{m+1}{k+1}$.
\end{enumerate}
\end{lemma} 

\begin{proof} 
  Recall that the Stirling number $\stirl{m}{k}$ counts the number of elements
  of the set $\binom{[m]}{k}$ of all partitions of $[m]=\{1,2,\ldots,m\}$
  into $k$ subsets.
	
  In order to prove {\rm(i)} let $ \sigma:[k]\rightarrow[k]$ be a 
  permutation of the integers in $[k]$ and $T=(T_1,\ldots,T_k)$ be 
  a partition of $[m]$ into $k$ subsets. We claim 
  that  the map  $\varphi$ which sends each pair $(\sigma,T)$ 
  to  the chain 
  \begin{equation*}
    \emptyset\subset{}T_{\sigma(1)}\subset(T_{\sigma(1)}\cup  
    T_{\sigma(2)})\subset\cdots\subset\bigcup_{i=1}^{k-1} 
    T_{\sigma(i)}\subset\bigcup_{i=1}^{k}T_{\sigma(i)}=[m]
  \end{equation*}
  is a  bijection between 
  $[k]\times\binom{[m]}{k}$ and $\spsb_m(\emptyset,k).$ 
	 
  To prove our claim, notice first that, since the sets $T_1,\ldots,T_k$
  are non-empty, the inclusions in the chain $\varphi(\sigma,T)$
  are strict and thus $\varphi$ is well defined. 
  To prove that $\varphi$ is injective, let $\sigma,\tau$ be two permutations 
  of $[k]$, and $T=(T_1,\ldots,T_k),$ $T'=(T'_1,\ldots,T'_k)$ be two partitions 
  of $[m]$ into $k$ subsets. We  assume that 
  $\varphi(\sigma,T)=\varphi(\tau,T')$ and we will prove that $\sigma=\tau$ 
  and $\{T_1,\ldots,T_k\}=\{T'_1,\ldots,T'_k\}.$ We use induction on the size 
  of $[m]$, the case $m=1$ being trivial.	We next assume that our 
  assumption holds true for any proper subset of $[m]$ and any $k<m$ and we 
  prove it for $[m].$ To this end, since
  $\varphi(\sigma,T)=\varphi(\tau,T')$,
  we have that the chains 
  \begin{align*}
    \emptyset\subset{}T_{\sigma(1)}\subset(T_{\sigma(1)}\cup{}
    T_{\sigma(2)})\subset\cdots\subset\bigcup_{i=1}^{k-2}T_{\sigma(i)},
    \\
    \emptyset\subset{}T'_{\tau(1)}\subset(T'_{\tau(1)}\cup{}T'_{\tau(2)})
    \subset\cdots\subset\bigcup_{i=1}^{k-2}T'_{\tau(i)}
  \end{align*}
  are identical. Thus, using the induction hypothesis, we deduce that 
  $T_{\sigma(i)}=T'_{\tau(i)}$ and $\sigma(i)=\tau(i)$ for all $1\leq{}i\leq{}k-1$.
  Clearly, $\sigma(k),\tau(k)\in{}K=[k]\sm{}\{\sigma(i):1\leq{}i\leq{}k-1\}.$
  But since $|K|=1$, we have that $\sigma(k)=\tau(k).$ Moreover, since 
  $[m]=\bigcup_{i=1}^{k-1}T_{\sigma(i)}=\bigcup_{i=1}^{k}T'_{\tau(i)}$
  we deduce that $T_{\sigma(k)}=T'_{\tau(k)}.$ This completes our induction.
  Finally, to prove that $\varphi$ is onto, we consider a chain
  $\emptyset\subset{}S_1\subset{}S_2\subset\cdots\subset{}S_{k-1}\subset[m]$
  in $\spsb_{m}(\emptyset,k)$ and we set
  $T_k:=[m]\setminus{}S_{k-1},T_{k-1}:=S_{k-1}\setminus{}S_{k-2},\dots,
  \allowbreak{}T_{2}:=S_{2}\setminus{}S_{1}$ and $T_{1}:=S_{1}.$
  It is immediate to see that $T_1,\ldots,T_k$ is a partition of $[m]$
  into $k$ non-empty sets and that
  $\varphi(A,{\tt id})=\emptyset\subset{}S_1\subset{}S_2\subset\cdots\subset{}S_k$,
  where ${\tt id}$ is the identity permutation in $[k].$
  %

	
  To prove {\rm (ii)}, notice that
  \begin{align*}
    \spsa_m(\emptyset,&\,k)
    =\,\{(S_1,\ldots,S_k)\mid{}
    \emptyset\subseteq{}S_1\subset{}\cdots\subset{}S_k\subset{}[m]\}\\  
    =&\,\{(S_1,\ldots,S_k)\mid{}
    \emptyset\subset{}S_1\subset{}\cdots\subset{}S_k\subset{}[m]\}
    \,\bigcup\,\{(S_1,\ldots,S_k)\mid{}\emptyset=S_1
    \subset{}\cdots\subset{}S_k\subset{}[m]\}\\
    =&\,\{(S_1,\ldots,S_k)\mid{}
    \emptyset\subset{}S_1\subset{}\cdots\subset{}S_k\subset{}[m]\}
    \,\bigcup\,\{(S_2,\ldots,S_k)\mid{}\emptyset\subset{}
    S_2\subset{}\cdots\subset{}S_k\subset{}[m]\}.
  \end{align*}
  Using {\rm(i)}, we have:
  \begin{align*}
    \spsac_m(\emptyset,k)
    &=\spsbc_m(\emptyset,k)+\spsbc_m( \emptyset,k-1)\\
    &=\sum_{j=0}^{k+1}(-1)^{k+1-j}\tbinom{k+1}{j}j^m+\sum_{j=0}^{k}(-1)^{k-j}
    \tbinom{k}{j}j^m\\
    &=(k+1)^m+\sum_{j=0}^k(-1)^{k+1-j}
    \biggl(\tbinom{k+1}{j}-\tbinom{k}{j}\biggr)\,j^m\\
    &=(k+1)^m+\sum_{j=0}^k(-1)^{k+1-j}\tbinom{k}{j-1}j^m\\
    &=(k+1)^m+\sum_{i=0}^{k-1}(-1)^{k-i}\tbinom{k}{i}(i+1)^m\\
    &=\sum_{i=0}^{k}(-1)^{k-i}\tbinom{k}{i}(i+1)^m\\
    &=k!\,\stirl{m+1}{k+1}.\qedhere
  \end{align*}  
\end{proof}
\medskip

The following combinatorial identities are used in the
proof of Lemma \ref{lem:hQ}. 
\begin{lemma}\label{stirl-euler}
For any $m\ge{}1$, we have:
\begin{equation}
  \sum_{i=0}^{m}i!\,\stirl{m+1}{i+1}(t-1)^{m-i}
  =\sum_{j=0}^{m-1}\eul{m}{j}t^{m-j},
  \label{WE1}
\end{equation} 
and
\begin{equation} 
  \label{WE2} 
  \sum_{i=0}^{m-1}(i+1)!\,\stirl{m}{i+1}(t-1)^{m-i-1}
  =\sum_{j=0}^{m-1}\eul{m}{j}t^{m-1-j}.
\end{equation} 
\end{lemma} 
\begin{proof}
Observe that:
\begin{align*}
  \sum_{i=0}^{m}i!\,\stirl{m+1}{i+1}(t-1)^{m-i}
  &=\sum_{i=0}^{m}\left(\frac{1}{i+1}\sum_{j=0}^{i+1}(-1)^{i+1-j}\tbinom{i+1}{j}
  j^{m+1}\right)\left(\sum_{j'=0}^{m-i}(-1)^{m-i-j'}\tbinom{m-i}{j'}t^{j'}
  \right)\\
  &=\sum_{i=0}^{m}\left(\sum_{j=0}^{i+1}(-1)^{i+1-j}\tbinom{i}{j-1}j^{m}\right)
  \left(\sum_{j'=0}^{m-i} (-1)^{m-i-j'}\tbinom{m-i}{j'}\,t^{j'}\right)\\
  &=\sum_{j=0}^{m+1}\sum_{j'=0}^{m}(-1)^{m+1-j-j'}j^{m}
  \sum_{i=0}^{m}\tbinom{m-i}{j'}\tbinom{i}{j-1}\,t^{j'}\\
  &=\sum_{j=1}^{m+1}\sum_{j'=0}^{m}(-1)^{m+1-j-j'}j^{m}\tbinom{m+1}{j+j'}\,
  t^{j'}\\
  &\hspace{-0.5cm}\overset{\ell:=m-j'}{=}\sum_{j=1}^{m+1}
  \sum_{\ell=0}^{m}(-1)^{\ell-j+1}j^{m}\tbinom{m+1}{m-\ell+j}\,t^{m-\ell}\\
  &\hspace{-0.5cm}\overset{i:=\ell-j+1}{=}\sum_{\ell=0}^{m}\sum_{i=0}^{\ell}
  (-1)^i(\ell-i+1)^{m}\tbinom{m+1}{i}\,t^{m-\ell}\\
  &=\sum_{\ell=0}^{m}\eul{m}{\ell}t^{m-\ell}=
  \sum_{\ell=0}^{m-1}\eul{m}{\ell}t^{m-\ell},
\end{align*}
where in the last sum we used the fact that $\eul{m}{m}=0$ for all  $m\geq{}1$.
	
To prove \eqref{WE2}, we distinguish between the cases $m=1$ and $m>1$.
For $m=1$ we have:
\begin{align*}
  \sum_{i=0}^{m-1}(i+1)!\stirl{m}{i+1}(t-1)^{m-i-1}&=(t-1)^0=1=t^0
  =\sum_{i=0}^{m-1}\eul{m}{i}t^{m-1-i},
\end{align*}
where we used the fact that $\stirl{1}{1}=\eul{1}{0}=1$.
For $m>1$, we set $\Sgen{m}{t}:=\sum_{i=0}^{m}i!\,\stirl{m+1}{i+1}t^{m-i}$ 
and $\E{m}{t}:=\sum_{i=0}^{m-1}\eul{m}{i}t^{m-i}$, so that relation
\eqref{WE1} is equivalent to $\Sgen{m}{t-1}=\E{m}{t}.$ We then have:
\begin{align*}
  \sum_{i=0}^{m-1}(i+1)!\stirl{m}{i+1}(t-1)^{m-i-1}&=
  m\sum_{i=0}^{m-1}i!\,\stirl{m}{i+1}(t-1)^{m-i-1}-
  \sum_{i=0}^{m-1}(m-i-1)\,i!\,\stirl{m}{i+1}(t-1)^{m-i-1}\\
  &=m\,\Sgen{m-1}{t-1}+(t-1)\,{\mathtt S}'_{m-1}(t-1)\\
  &=m\,\E{m-1}{t}-(t-1)\EE{m-1}{t}\\   
  &=\sum_{i=0}^{m-1}\biggl(m\eul{m-1}{i}-(m-1-i)\eul{m-1}{i}+(m-i)
  \eul{m-1}{i-1}\biggr)t^{m-1-i}\\
  &=\sum_{i=0}^{m-1}\eul{m}{i}t^{m-1-i},
\end{align*}
where in the last equality we used the recurrence relation of Eulerian
numbers:
\[\eul{m}{i}=(m-i)\eul{m-1}{i-1}+(i+1)\eul{m-1}{i}.\qedhere\]
\end{proof}

%% file: appendix_recurrence.tex
\section{Relations appearing in the derivation
  of the recurrence relation for the \texorpdfstring{$h$}{h}-vector of
  \texorpdfstring{$\fF_R$}{FR}}
\label{app:sec:recF}

\subsection{McMullen's relation restated}
\label{app:sec:McM-restated}

McMullen \cite{m-mnfcp-70}, in his original proof of the Upper Bound
Theorem for polytopes, proved that for any simplicial $d$-polytope $P$
the following relation holds:
\begin{equation}
  \label{equ:hkrelP}
  (k+1)h_{k+1}(\bx{P})+(d-k)h_k(\bx{P})
  =\sum_{v\in\text{vert}(\bx{P})}h_k(\bx{P}/v),\qquad 0\le{}k\le{}d-1.
\end{equation}
Below we rewrite these relations in terms of generating functions.

\begin{lemma}[McMullen 1970]
\label{lem:McM}
For any simplicial $d$-polytope $P$ 
\begin{equation}
  d\,\h{\bx{P}}+(1-t)\hh{\bx{P}}=\sum_{v\in\text{vert}(\bx{P})}\h{\bx{P}/v}.
  \label{McM}
\end{equation}
\end{lemma} 
\begin{proof}
  Multiplying both sides of \eqref{equ:hkrelP} by $t^{d-k-1}$, and
  summing over  all $ 0 \leq k \leq d$, we get:
  \begin{equation}
    \label{equ:rel1}
    \sum_{k=0}^{d} (k+1)h_{k+1}(\bx{P}) 
    t^{d-k-1}+\sum_{k=0}^{d}(d-k)h_k(\bx{P})t^{d-k-1} 
    =\sum_{k=0}^{d} \sum_{v\in\text{vert}(\bx{P})}h_k(\bx{P}/v)t^{d-k-1}.
  \end{equation}
  For the right-hand side of \eqref{equ:rel1} we have:
  \begin{equation} 
    \label{equ:rel1-rhs}
    \sum_{k=0}^{d} \sum_{v\in\text{vert}(\bx{P})}h_k(\bx{P}/v)t^{d-k-1}
    =\sum_{v\in\text{vert}(\bx{P})}\sum_{k=0}^{d}h_k(\bx{P}/v) t^{d-1-k}\\
    =\sum_{v\in\text{vert}(\bx{P})}   \h{\bx{P}/v},
  \end{equation}
  whereas for the left-hand side of \eqref{equ:rel1} we get:
  \begin{equation}\label{equ:rel1-lhs}
    \begin{aligned}
      \sum_{k=0}^{d} (k+1)h_{k+1}(\bx{P}) t^{d-k-1}
      &+   \sum_{k=0}^{d} (d-k)h_k(\bx{P})t^{d-k-1} \\
      &=
      \sum_{k=0}^{d} k h_{k}(\bx{P}) t^{d-k} +
      \sum_{k=0}^{d} (d-k)h_k(\bx{P})t^{d-k-1}\\
      &=
      d \sum_{k=0}^{d}  h_{k}(\bx{P}) t^{d-k}
      +(1-t)\sum_{k=0}^{d} (d-k) h_{k}(\bx{P}) t^{d-k-1}\\
      &=
      d\, \h{\bx{P}}+(1- t)\,  \hh{\bx{P}}.
    \end{aligned}
  \end{equation}
  Substituting \eqref{equ:rel1-rhs} and \eqref{equ:rel1-lhs} in
  \eqref{equ:rel1} we recover the relation in the statement of  the lemma.
\end{proof}
\subsection{One more auxiliary set}
\label{app:sec:setD}

Recall that $\D{R}{T}{X}{\ell}$ denotes the cardinality of the set:
\begin{equation*}
  \mathcal{D}(R,T,X,\ell):=\{(S_1,\ldots,S_\ell):
  X\subseteq{}S_1\subset{}S_2\subset\cdots\subset{}S_{\ell}\subset{}R
  \mbox{ and }S_i=T\mbox{ for some }1\leq{}i\leq\ell\}.
\end{equation*} 
The following lemma expresses the sum of the cardinalities
$\D{R}{T}{X}{\ell}$, over all $T$ with $X\subseteq{}T\subset{}R$, in
terms of the Stirling numbers of the second kind.

\begin{lemma} For any $\ell\in\naturals$, and $X$, $R$ with
  $\emptyset\subseteq{}X\subset{}R$, we have:
\label{lem:D-S}
\begin{equation}
  \sum_{X\subseteq{}T\subset{}R}\D{R}{T}{X}{\ell}
  =\ell\,\ell!\,\stirl{|R|-|X|+1}{\ell+1}.
\label{equ:DS}
\end{equation}
\end{lemma}
\begin{proof}
  The left-hand side of \eqref{equ:DS} is the cardinality of the set
  \begin{equation*} 
    {\mathcal Y}=\{(S_1,\ldots,S_\ell):X\subseteq{}T\subset{}R,
    X\subseteq{}S_1\subset{}S_2\subset\cdots\subset{}S_{\ell}\subset{}R
    \mbox{ and }S_i=T\mbox{ for some }1\leq{}i\leq\ell\},
  \end{equation*} 
  which is nothing but $\ell$ copies of the set
  \begin{equation*}
    {\mathcal Z}=\{(S_1,S_2,\ldots,S_{\ell})		
    \mid{}X\subseteq{}S_1\subset{}S_2\subset{}\cdots\subset{}S_{\ell}\subset{}R\}.
  \end{equation*} 
  Indeed, 
  \begin{align*} 
    {\mathcal Y}=&\{(S_1,\ldots,S_\ell):X\subseteq{}T\subset{}R,
    X\subseteq{}S_1\subset{}S_2\subset\cdots\subset{}S_{\ell}\subset{}R
    \mbox{ and }S_i=T\mbox{ for some }1\leq{}i\leq\ell\}\\
    =&\{i:1\leq{}i\leq{}\ell\}\times\{(S_1,\ldots,S_\ell):
    X\subseteq{}T\subset{}R,X\subseteq{}S_1\subset{}S_2\subset\cdots
    \subset{}S_{\ell}\subset{}R\mbox{ and }S_i=T\}\\
    =&\{i:1\leq{}i\leq{}\ell\}\times\{(S_1,\ldots,S_\ell):
    X\subseteq{}S_1\subset{}S_2\subset\cdots\subset{}S_{\ell}\subset{}R\}\\
    =&\{i:1\leq{}i\leq{}\ell\}\times{\mathcal Z}.
  \end{align*} 
  By Lemma \ref{lem:chain_counting}{\rm(ii)}, the cardinality of 
  $\mathcal{Z}$ is $\ell!\,\stirl{|R|-|X|+1}{\ell+1}$ 
  and this completes our proof.
\end{proof}

%% file: appendix_matrix.tex
\section{Determinants used in the tightness construction}
\label{ss:technical} 

\renewcommand{\ye}[2]{{\tilde y_{\scalebox{0.7}{$#1,#2$}}}}


\begin{definition} 
\label{def:det}
Let $Y_i=\{y_{i,1},\ldots,y_{i,\kappa_1}\}$, $1\leq{}i\leq{}n$, be non-empty 
disjoint sets of real numbers. Set  $K:=\kappa_1+\kappa_2+\cdots+\kappa_n$, 
 $m:=K-2n-2$ and let $\mu_1<\mu_2<\cdots<\mu_{m}$ be non-negative integers. We 
 denote by  $\ptn[Y]$ the partition  $Y_1\cup\cdots\cup{}Y_n$ and we define 
 the  $K\times{}K$ matrix $\De{K}{\ptn[Y];\mu_1,\ldots,\mu_m}$ as follows:
\renewcommand{\y}[3]{y_{\scalebox{0.6}{$#1,#2$}}^{\scalebox{0.6}{$#3$}}}

\scalebox{0.85}{
$  
\De{K}{\ptn[Y];\mu_1,\ldots,\mu_m}:=
\left( \begin{array}{cccccccccccccc}
\y{1}{1}{\mu_1}&\cdots&\y{1}{\kappa_1}{\mu_1}&
\0  &\cd    &\0  &\0   &\cd
&\0  &\cd   &\0  &\cd   &\0\\
\0 &\cd    &\0  & 
\y{2}{1}{\mu_1}&\cdots&\y{2}{\kappa_2}{\mu_1}&
\0   &\cd   &\0  &\cd   &\0  &\cd   &\0\\
\0 &\cd    &\0  &\0  &\cd    &\0  &
\y{3}{1}{\mu_1}&\cdots&\y{3}{\kappa_3}{\mu_1}&
 \cd   &  \0&\cd   &\0\\
\vd&\dd    &\vd &\vd &\dd    &\vd &\vd  &\dd   &\vd &\cd   &\vd &\dd   &\vd\\
\0 &\cd    &\0  &\0  &\cd    &\0  &\0   &\cd   &\0  &\cd  &
\y{n}{1}{\mu_1}&\cdots&\y{n}{\kappa_1}{\mu_1} \\
\y{1}{1}{\mu_2}&\cdots&\y{1}{\kappa_1}{\mu_2}&
\0&\cd&\0&
\0&\cd&\0&   \cd&
\0&\cd&\0\\
\0&\cd&\0&
\y{2}{1}{\mu_2}&\cdots&\y{2}{\kappa_2}{\mu_2}&
\0&\cd&\0&  \cd&
\0&\cd&\0\\
\0&\cd&\0& 
\0&\cd&\0& 
\y{3}{1}{\mu_2}&\cdots&\y{3}{\kappa_3}{\mu_2}&  \cd&
\0&\cd&\0\\
\vd&\dd&\vd& 
\vd&\dd&\vd&  
\vd&\dd&\vd& \cd&
\vd&\dd&\vd\\
\0&\cd&\0&
\0&\cd&\0&
\0&\cd&\0& \cd&
\y{n}{1}{\mu_2}&\cdots&\y{n}{\kappa_n}{\mu_2}\\
\y{1}{1}{\mu_3}&\cdots&\y{1}{\kappa_1}{\mu_3}&  
\y{2}{1}{\mu_3}&\cdots&\y{2}{\kappa_2}{\mu_3}&
\y{3}{1}{\mu_3}&\cdots&\y{3}{\kappa_3}{\mu_3}&  \cd& 
\y{n}{1}{\mu_3}&\cdots&\y{n}{\kappa_n}{\mu_3}\\
\y{1}{1}{\mu_4}&\cdots&\y{1}{\kappa_1}{\mu_4}&  
\y{2}{1}{\mu_4}&\cdots&\y{2}{\kappa_2}{\mu_4}&
\y{3}{1}{\mu_4}&\cdots&\y{3}{\kappa_3}{\mu_4}&  \cd& 
\y{n}{1}{\mu_4}&\cdots&\y{n}{\kappa_n}{\mu_4}\\
\vd&\dd&\vd& \vd&\dd&\vd&  \vd&\dd&\vd& \cd& \vd&\dd&\vd\\
\y{1}{1}{\mu_m}&\cdots&\y{1}{\kappa_1}{\mu_m}&  
\y{2}{1}{\mu_m}&\cdots&\y{2}{\kappa_2}{\mu_m}&
\y{3}{1}{\mu_m}&\cdots&\y{3}{\kappa_3}{\mu_m}&  \cd& 
\y{n}{1}{\mu_m}&\cdots&\y{n}{\kappa_n}{\mu_m}
\end{array}\right)
\overset
{\substack{\mbox{\tiny row}\\\mbox{\tiny index}}}
{
\begin{array}{c}  
\mbox{\tiny 1}\\
\mbox{\tiny 2}\\
\mbox{\tiny 3}\\  \vd\\
\mbox{\tiny n}\\
\mbox{\tiny n+1}\\
\mbox{\tiny n+2}\\
\mbox{\tiny n+3}\\ \vd\\
\mbox{\tiny 2n}\\
\mbox{\tiny 2n+1}\\
\mbox{\tiny 2n+2}\\ \vd\\
\mbox{\tiny K}
\end{array}
}
$}
\\
We  denote by 
$\DDD{K}{\ptn[Y];\mu_1,\ldots,\mu_m}$ the signed determinant  
$(-1)^{\frac{n(n-1)}{2}}\left|\De{K}{\ptn[Y];\mu_1,\ldots,\mu_m}\right|. $  
\end{definition} 

We, now, parameterize all $y_{i,j}$'s
as follows: for each $1\leq{}i\leq{}n$ we choose arbitrary real numbers 
$0<x_{i,1}<x_{i,2}<\cdots<x_{i,\kappa_i}$ and  non-negative integers 
$0\leq\beta_n<\beta_{n-1}<\cdots<\beta_1$. Then, we set
$y_{i,j}:=x_{i,j}\tau^{\beta_i}$ where $\tau$ is a positive parameter, 
and consider  $\DDD{K}{\ptn[Y];\mu_1,\ldots,\mu_m}$ as a polynomial in 
$\tau$. In the next lemma we essentially show that, for sufficiently
small $\tau$, the determinant $\DDD{K}{\ptn[Y];\mu_1,\ldots,\mu_m}$ is
strictly positive.
%

\begin{lemma} \label{lem:asymptotic-tau}
	If the elements of the sets $Y_i$, $1\leq{}i\leq{}r$, are parameterized as 
	above  and 	$\ptn[Y]=Y_1\cup\cdots\cup{}Y_n$,
	 then $\DDD{K}{\ptn[Y];\mu_1,\ldots,\mu_m}= 
	A\tau^a+O(\tau^{a+1})$, where $A>0$ and $a$ is a positive integer. 
\label{lem:dominant_term}
\end{lemma}
\begin{proof}
\newcommand{\mycd}{\cdots}
	To prove our claim we use the Binet-Cauchy 
	theorem \cite{cauchy-binet}. More precisely, let $J$ be a subset of 
	$\{1,2,\ldots,n(m+1)\}$ of size $K$. We denote by $L_{[K],J}$ the 
	$K\times{}K$ matrix whose columns are the columns of $L$  at indices from 
	$J$ and  by $ R_{J,[K]} $   the  $ K \times K$ matrix whose rows are the 
	rows  of $R$  at indices from $J.$ The Binet-Cauchy theorem   states that:
	\begin{equation}
	\det(L R) = \sum \limits_{J} \det(L_{[K],J}) \det(R_{J,[K]}), 
	\label{Binet}     
	\end{equation}
	where  the sum is taken over all subsets $J$ of $\{1,2,\ldots, n(m+1)\}$ of 
	size $K.$ 
	
	To apply the Binet-Cauchy theorem in our case, notice that the matrix
    $\De{K}{\ptn[Y];\mu_1,\ldots,\mu_m}$ can be factorized into a product of a 
    $K\times{}n(m+1)$ matrix $L$ and an $n(m+1)\times{}K$ matrix $R$ as shown 
    below:
\newcommand{\shi}[2]{#1^{{\langle #2\rangle}}}
\newcommand{\hi}[2]{\rotatebox{0}{\scalebox{0.5}{$\,#1\,\,$}}}
\newcommand{\hir}[2]{\rotatebox{90}{\scalebox{0.5}{$#1$}}}
\newcommand{\hti}[2]{\rotatebox{90}{\scalebox{0.5}{$#1 + #2(m+1)$}}}
\newcommand{\ti}[1]{\rotatebox{0}{\scalebox{0.5}{$#1$}}}
\begin{center}
$L=$\scalebox{0.9}{
$
\begin{array}{c}
\ti{1}\\
\ti{2}\\
\ti{3}\\ \scalebox{0.7}{$\vd$}\\
\ti{n}\\
\ti{n+1}\\
\ti{n+2}\\
\ti{n+3}\\ \scalebox{0.7}{$\vd$}\\
\ti{2n}\\
\ti{2n+1}\\
\ti{2n+2}\\ \scalebox{0.7}{$\vd$}\\
\ti{K}
\end{array}
\overset{
\begin{array}{cccccccccccccccccccccccccccccccccc}
\hi{1}{}&\hi{2}{}&\hi{3}{}&\hi{4}{}&\cdots&\hir{m+1}{}&
\hti{1}{1}&\hti{2}{1}&\hti{3}{1}&\hti{4}{1}&\cdots&\hti{m+1}{1}&
\hti{1}{2}&\hti{2}{2}&\hti{3}{2}&\hti{4}{2}&\cdots&\hti{m+1}{2}&
\cdots&
\hti{1}{(n-1)}&\hti{2}{(n-1)}&\hti{3}{(n-1)}&\hti{4}{(n-1)}&\cdots&\hti{m+1}{(n-1)}
\end{array}}
{\left(\begin{array}{ccccccccccccccccccccccccccccccccc}
1&\0&\0&\0&\cd&\0&
\0&\0&\0&\0&\cd&\0&
\0&\0&\0&\0&\cd&\0&  \mycd&
\0&\0&\0&\0&\cd&\0\\
\0&\0&\0&\0&\cd&\0&
1&\0&\0&\0&\cd&\0&
\0&\0&\0&\0&\cd&\0&  \mycd&
\0&\0&\0&\0&\cd&\0\\
\0&\0&\0&\0&\cd&\0&
\0&\0&\0&\0&\cd&\0&
1&\0&\0&\0&\cd&\0&   \mycd&
\0&\0&\0&\0&\cd&\0\\
\vd&\vd&\vd&\vd&\dd&\vd&
\vd&\vd&\vd&\vd&\dd&\vd&
\vd&\vd&\vd&\vd&\dd&\vd&  \dd&
\vd&\vd&\vd&\vd&\dd&\vd\\
\0&\0&\0&\0&\cd&\0&
\0&\0&\0&\0&\cd&\0&
\0&\0&\0&\0&\cd&\0&  \mycd&
1&\0&\0&\0&\cd&\0\\
\0&1&\0&\0&\cd&\0&
\0&\0&\0&\0&\cd&\0&
\0&\0&\0&\0&\cd&\0&  \mycd&
\0&\0&\0&\0&\cd&\0\\
\0&\0&\0&\0&\cd&\0&
\0& 1&\0&\0&\cd&\0& 
\0&\0&\0&\0&\cd&\0& \mycd&
\0&\0&\0&\0&\cd&\0\\
\0&\0&\0&\0&\cd&\0&
\0&\0&\0&\0&\cd&\0&
\0& 1&\0&\0&\cd&\0& \mycd&
\0&\0&\0&\0&\cd&\0\\
\vd&\vd&\vd&\vd&\dd&\vd&
\vd&\vd&\vd&\vd&\dd&\vd&
\vd&\vd&\vd&\vd&\dd&\vd& \dd&
\vd&\vd&\vd&\vd&\dd&\vd\\
\0&\0&\0&\0&\cd&\0&
\0&\0&\0&\0&\cd&\0&
\0&\0&\0&\0&\cd&\0&  \mycd&
\0&1&\0&\0&\cd&\0\\
\0&\0&1&0&\cdots&0&
\0&\0&1&0&\cdots&0&
\0&\0&1&0&\cdots&0& \mycd&
\0&\0&1&0&\cdots&0\\
\0&\0& 0& 1&\cdots& 0&
\0&\0& 0& 1&\cdots& 0&
\0&\0& 0& 1&\cdots& 0&  \mycd&
\0&\0& 0& 1&\cdots& 0\\
\vd&\vd&\vdots&\vdots&\ddots&\vdots&
\vd&\vd&\vdots&\vdots&\ddots&\vdots&
\vd&\vd&\vdots&\vdots&\ddots&\vdots& \dd& 
\vd&\vd&\vdots&\vdots&\ddots&\vdots\\
\0&\0& 0& 0&\cdots& 1&
\0&\0& 0& 0&\cdots&1&
\0&\0& 0& 0&\cdots&1& \mycd&
\0&\0& 0& 0&\cdots& 1
\end{array}\right)}
$},\\
[15pt]

$R=$\,\scalebox{0.8}{%
$\left(\begin{array}{cccccccccccccc}
\yy{1}{1}{\mu_1}&\cdots&\yy{1}{\kappa_1}{\mu_1}&
\0&\cd   &\0&
\0&\cd   &\0&   \mycd&
\0&\cd&\0\\
\yy{1}{1}{\mu_2}&\cdots&\yy{1}{\kappa_1}{\mu_2}&
\0&\cd&\0&
\0&\cd&\0& \mycd
&\0&\cd&\0\\
\yy{1}{1}{\mu_3}&\cdots&\yy{1}{\kappa_1}{\mu_3}&
\0&\cd&\0&
\0&\cd&\0&  \mycd
&\0&\cd&\0\\
\vdots&\ddots&\vdots&
\vd&\dd&\vd&
\vd&\dd&\vd& \dd
&\vd&\dd&\vd\\
\yy{1}{1}{\mu_m}&\cdots&\yy{1}{\kappa_1}{\mu_m}&
\0&\cd&\0&
\0&\cd&\0&  \mycd
&\0&\cd&\0\\
\0&\cd&\0&
\yy{2}{1}{\mu_1}&\cdots&\yy{2}{\kappa_2}{\mu_1}&
\0&\cd&\0&   \mycd&
\0&\cd&\0\\
\0&\cd&\0&
\yy{2}{1}{\mu_2}&\cdots&\yy{2}{\kappa_2}{\mu_2}&
\0&\cd&\0&   \mycd&
\0&\cd&\0\\
\0&\cd&\0&
\yy{2}{1}{\mu_3}&\cdots&\yy{2}{\kappa_2}{\mu_3}&
\0&\cd&\0&   \mycd&
\0&\cd&\0\\
\vd&\dd&\vd&
\vdots&\ddots&\vdots&
\vd&\dd&\vd& \dd
&\vd&\dd&\vd\\
\0&\cd&\0&
\yy{2}{1}{\mu_m}&\cdots&\yy{2}{\kappa_2}{\mu_m}&
\0&\cd&\0&  \mycd&
\0&\cd&\0\\
\0&\cd&\0&
\0&\cd&\0&
\yy{3}{1}{\mu_1}&\cdots&\yy{3}{\kappa_3}{\mu_1}& \mycd &
\0&\cd&\0\\

\0&\cd&\0&
\0&\cd&\0&
\yy{3}{1}{\mu_2}&\cdots&\yy{3}{\kappa_3}{\mu_2}& \mycd
&\0&\cd&\0\\
\0&\cd&\0&
\0&\cd&\0&
\yy{3}{1}{\mu_3}&\cdots&\yy{3}{\kappa_3}{\mu_3}& \mycd
&\0&\cd&\0\\
\vd&\dd&\vd
&\vd&\dd&\vd
&\vdots&\ddots&\vdots& \mycd&
\vd&\dd&\vd\\
\0&\cd&\0&
\0&\cd&\0&
\yy{3}{1}{\mu_m}&\cdots&\yy{3}{\kappa_3}{\mu_m}& \mycd
&\0&\cd&\0\\
\vd&\dd&\vd&
\vd&\dd&\vd&
\vd&\dd&\vd& \ddots
&\vd&\dd&\vd\\
\0&\cd&\0&
\0&\cd&\0&
\0&\cd&\0& \mycd&
\yy{n}{1}{\mu_1}&\cdots&\yy{n}{\kappa_n}{\mu_1}\\
\0&\cd&\0&
\0&\cd&\0&
\0&\cd&\0& \mycd&
\yy{n}{1}{\mu_2}&\cdots&\yy{n}{\kappa_n}{\mu_2}\\
\0&\cd&\0&
\0&\cd&\0&
\0&\cd&\0& \mycd&
\yy{n}{1}{\mu_3}&\cdots&\yy{n}{\kappa_n}{\mu_3}\\
\vd&\dd&\vd&
\vd&\dd&\vd&
\vd&\dd&\vd& \vd&
\vdots&\ddots&\vdots\\
\0&\cd&\0&
\0&\cd&\0&
\0&\cd&\0& \mycd&
\yy{n}{1}{\mu_m}&\cdots&\yy{n}{\kappa_n}{\mu_m}\
\end{array}\right)$}.
\end{center}
The numbers over and sideways of $L$ indicate the column and row
numbers, respectively.

Recall that  $y_{i,j}=x_{i,j}\tau^{\beta_i}.$ Then it is not hard to see 
that, for each $J\subseteq\{1,\ldots,n(m+1)\}$  with $|J|=K$,  the  
sub-matrix $L_{[K],J}$ is independent of $\tau$ while $R_{J,[K]}$  is a 
block-diagonal matrix whose blocks are  generalized  Vandermonde  
determinants (cf. \cite{g-atm-05}) from which we can extract powers of $\tau$.  
More precisely, we set $\shi{k}{i}:=k+(i-1)(m+1)$ and we write each index set 
$J$ of the Binet-Cauchy expansion as  $J_1\cup{}J_2\cup\cdots\cup{}J_{n}$, where 
$J_i\subseteq \left\{\shi{1}{i},\ldots,\shi{(m+1)}{i} \right\}$.
We then have:
\begin{equation}
  \det(R_{J,[K]}) = \tau^{a(J)} \prod_{i=1}^n \GVD(J_i),
  \label{a}
\end{equation}
where
\begin{equation*}
  a(J)=\beta_1\sum\limits_{\shi{j}{1}\in{}J_1}\mu_j
  +\beta_2\sum \limits_{\shi{j}{2}\in J_2}\mu_j
  +\cdots+\beta_n\sum\limits_{\shi{j}{n}\in J_n}\mu_j,
\end{equation*}
and $\GVD(J_i)$ is a positive generalized Vandermonde determinant%
\footnote{It is a well-known fact that, if the parameters in the
  columns of the generalized Vandermonde determinant are in
  strictly increasing order, then the Vandermonde determinant is itself 
  strictly positive (see \cite{g-atm-05} for a proof of this fact).},
independent of $\tau$, depending on the $y_{i,j}$'s with $j\in{}J_i$.
Thus, combining  \eqref{Binet} and \eqref{a}, we deduce that 
$\DDD{K}{\ptn[Y];\mu_1,\ldots,\mu_m}$ is a polynomial in $\tau.$ 
To prove our claim it suffices to find the subset $J$ for which
$a(J)$ is minimal and, for this  $J$,
evaluate the sign of the coefficient of $\tau^{a(J)}$.

Notice that a term $\det(L_{[K],J}) \det(R_{J,[K]})$ in the Cauchy-Binet  
expansion of $\DDD{K}{\ptn[Y];\mu_1,\mu_2,\allowbreak\ldots,\allowbreak\mu_m}$
vanishes in the following two cases:
\begin{itemize} 
\item[{\rm (i)}]
  $\shi{k}{i},\shi{k}{j}\in{}J$ for some $3\leq{}k\leq{}m+1$; in this case
  the $\shi{k}{i}$-th and $\shi{k}{j}$-th columns of $L_{[K],J}$
  are identical, and thus  $\det(L_{[K],J})=0$.
\item[{\rm (ii)}] 
  $|J_i|\neq{}k_i$ for at least some $1\leq{}i\leq{}n$; in this case 
  $R_{J,[K]}  $ is a block-diagonal square matrix  with non-square  
  non-zero blocks. The determinant of such a  matrix is always 	  
  zero.\footnote{To see this, consider the  Laplace expansion of the 
    matrix with respect to the  columns of its top-left block.}
\end{itemize}
Among all possible  index sets $J=J_1\cup\cdots\cup{}J_n$ for which   the
product $\det(L_{[K],J}){}\det(R_{J,[K]})$ does not vanish, we have to find the
one for which the exponent  $a(J)$ in \eqref{a} is the minimum possible.
To do this,  we combine condition {\rm(i)} above with the fact 
that  $\beta_1>\cdots>\beta_n$ and we deduce that the 
minimum exponent  $M(J)$ is attained if, for all $1\leq{}i\leq{}r:$
\begin{itemize}
\item
 $\shi{1}{i},\shi{2}{i}\in{}J_i$, and 
\item 
if $\shi{\kappa}{i}\in{}J_i$ and $\shi{\lambda}{i+1}\in{}J_{i+1}$ 
for some $\kappa,\lambda>2$,  then $\kappa<\lambda$.
\end{itemize}
Moreover, since from condition {\rm(ii)} we have $|J_i|=k_i$, 
we conclude that:
\begin{itemize}
	\item $J_1=J_1^\star:=
	\{\shi{1}{1},\shi{2}{1},\shi{3}{1},\ldots,\shi{k_1}{1}\}$
	 $=\{1,\ldots,k_1\}$, 
	\item 
	$J_2=J_2^\star:=\{\shi{1}{2},\shi{2}{2},\shi{(k_1+1)}{2},
	\ldots,\shi{(k_1+k_2-2)}{2}\}$,
	\item 
	$J_3=J_3^\star:=\{\shi{1}{3},\shi{2}{3},\shi{(k_1+k_2-1)}{3},
	\ldots,\shi{(k_1+k_2+k_3-4)}{3}\}$
	\item[] etc.
	\end{itemize}
For the above choice of $J^\star=J_1^\star\cup{}\cdots\cup{}J_n^\star$, the 
matrix $L_{[K],J}$ is:

\renewcommand{\ti}[1]{\rotatebox{0}{\scalebox{0.8}{$#1$}}}

\begin{equation*}
  L_{[K],J^\star} = 
  \scalebox{0.7}{$
  \left(\begin{array}{ccc:ccc:ccccccccc}
        1 &\0 &\0 &
       \0 &\0 &\0 &
       \0 &\0 &\0 & 
       \cd &
       \0 &\0 &\0 \\
       \0 &\0 & \0&
        1 &\0 &\0 &
        \0&\0 &\0 &
        \cd &
       \0 &\0 &\0 \\
       \0 &\0 &\0 &
       \0 &\0 &\0 &
        1 &\0 &\0 &
       \cd &
       \0 &\0 &\0 \\
		\vd &\vd &\vd &
		\vd &\vd &\vd &
		\vd &\vd &\vd &
		\dd &
		\vd &\vd &\vd \\
		\0 &\0 &\0 &
		\0 &\0 &\0 &
		\0 &\0 &\0 &
		\cd &
		1 &\0 &\0 \\
\hdashline
       \0 & 1 &\0 &
       \0 &\0 &\0 &
       \0 & \0&\0 &
	   \cd &
       \0 &\0 &\0 \\
       \0 & \0&\0 &
       \0 & 1 &\0 &
       \0 &\0 &\0 &
       \cd &
       \0 &\0 &\0 \\
       \0 & \0&\0 &
       \0 & \0&\0 &
       \0 & 1 &\0 &
       \cd & \0 & \0 & \0\\
		\vd &\vd &\vd &
		\vd &\vd &\vd &
		\vd &\vd &\vd &
		\dd &
		\vd &\vd &\vd \\
		\0 &\0 &\0 &
		\0 &\0 &\0 &
		\0 &\0 &\0 &
		\cd &
		\0 & 1 &\0 \\
           \hdashline   
      \0  &\0 &{\rm{}I}_{k_1-2}&
       \0 &\0 &\0 &
       \0 &\0 &\0 &
       \cd &
       \0 & \0 & \0
      \\
       \0 &\0 & \0&
       \0 &\0 &{\rm{}I}_{k_2-2}&
       \0 &\0 &\0 &
       \cd &
		\0 & \0 & \0       
       \\
       \0 &\0 & \0&
       \0 &\0 & \0& 
       \0 &\0 & {\rm{}I}_{k_3-2} &
      \cd &
      \0 & \0 & \0
       \\
		\vd &\vd &\vd &
		\vd &\vd &\vd &
		\vd &\vd &\vd &
		\dd &
		\vd &\vd &\vd \\
		\0 &\0 &\0 &
		\0 &\0 &\0 &
		\0 &\0 &\0 &
		\cd &
		\0 & \0 & {\rm{}I}_{k_n-2}
       \end{array}\right)
		\begin{array}{c}
		\mbox{\small row }\\
		\mbox{\small index}\\
		\ti{1}\\
		\ti{2}\\
		\ti{3}\\
		\vdots\\
		\ti{n}\\
		\hdashline
		\ti{n+1}\\
		\ti{n+2}\\
		\ti{n+3}\\
		\vdots\\
		\ti{2n}\\
		\hdashline
		\vdots\\
		\\
		\\\\ \\ \\ \\ 
		\end{array}
$}
\end{equation*}
Thus, in order to find the sign of our original determinant, we have to
evaluate $\det(  L_{[K],J^\star})$. To do this, we perform the appropriate 
row and column swaps so that $L_{[K],J^\star} $ becomes the identity matrix.  
More precisely, 
\begin{itemize}
\item we perform $n-1+(n-2)+(n-3)+\cdots+1=\tfrac{n(n-1)}{2}$
  row swaps so that, for all $1\le{}i\le{}n$, row $n+i$ is shifted upwards 
  and paired with row $i$, to become a $2\times{}2$ identity matrix,
\item we then perform an even number of column swaps to shift each $I_{k_i-2}$
  to its ``proper'' position (i.e., so that we get an identity matrix
  along with the corresponding $2\times{}2$ block of the previous step).
\end{itemize}
We therefore conclude that the sign of the dominant term of the expansion of 
the determinant of the matrix $\De{K}{\ptn[Y];\mu_1,\ldots,\mu_m}$ 
as a polynomial in $\tau$, is $(-1)^{\tfrac{n(n-1)}{2}}$ and 
this completes our proof. 
\end{proof}